\documentclass[12pt]{ociamthesis}
\usepackage{kpfonts}
\usepackage[T1]{fontenc}
\usepackage[utf8]{inputenc}
\usepackage{microtype}
\usepackage[all]{nowidow}

\usepackage{amsmath}
\usepackage{amssymb}
\usepackage{amsthm}
\usepackage{mathtools}
\usepackage{thmtools}
\usepackage{cancel}
\usepackage{faktor}
\usepackage{bm}
\usepackage[pdfauthor={Samo Novák},
            pdftitle={Homological Quantum Error Correction with Torsion}]{hyperref}
\usepackage{cleveref}
\usepackage{graphicx}
\usepackage{pdfpages}
\usepackage{xcolor}
\hypersetup{
    colorlinks,
    linkcolor={red!50!black},
    citecolor={blue!50!black},
    urlcolor={blue!80!black}
}
\usepackage{tikz}
\usepackage{tikz-cd}
\usetikzlibrary{
  backgrounds,
  fit,
  decorations.pathreplacing,
  decorations.markings,
  decorations.shapes,
  arrows.meta,
  calc,
  patterns,
  fadings,
  3d}
\usepackage{quiver}
\usepackage{enumitem}
\usepackage{blindtext}
\usepackage{caption}
\usepackage{subcaption}
\usepackage{multirow}
\usepackage{tikzit}

\tikzstyle{vertex}=[shape=circle, fill=black, minimum size=6pt, inner sep=0pt, outer sep=0pt]

\tikzstyle{open boundary}=[-, dashed]
\tikzstyle{fill helper}=[-, draw=none, tikzit draw=green]
\tikzstyle{edge}=[->]

\definecolor{zx_red}{RGB}{255, 136, 136}
\definecolor{zx_green}{RGB}{136, 255, 136}
\definecolor{zx_blue}{RGB}{136, 136, 255}
\newcommand{\colora}{zx_green}
\newcommand{\colorb}{zx_red}
\newcommand{\colorc}{zx_blue}
\newcommand{\colord}{yellow}

\tikzstyle{vertex}=[fill=black,draw=black,circle,minimum size=4pt,inner sep=0]
\tikzstyle{compedge}=[
postaction={decorate},
line width=0.7pt,
draw]
\tikzstyle{->-}=[compedge,
decoration={
  markings,
  mark=at position 0.55 with {\arrow{Triangle[scale=1]}}}]
\tikzstyle{-g>-}=[compedge,
decoration={
  markings,
  mark=at position 0.55 with {\arrow{Triangle[scale=1.5, fill=white]}}}]
\tikzstyle{-g>>-}=[compedge,
decoration={
  markings,
  mark=at position 0.50 with {\arrow{Triangle[scale=1.5, fill=white]}},
  mark=at position 0.60 with {\arrow{Triangle[scale=1.5, fill=white]}}}]
\tikzstyle{-g>>>-}=[compedge,
decoration={
  markings,
  mark=at position 0.5 with {\arrow{Triangle[scale=1.5, fill=white]}},
  mark=at position 0.55 with {\arrow{Triangle[scale=1.5, fill=white]}},
  mark=at position 0.6 with {\arrow{Triangle[scale=1.5, fill=white]}}}]
\tikzstyle{walk}=[
decoration={
  markings,
  mark=between positions 0 and 1 step 3pt with {
    \node[fill=black,
    signal from=west,
    signal to=east,
    signal,
    transform shape,
    minimum width=1pt,
    minimum height=3pt,
    inner sep=0]{};}},
decorate]
\tikzstyle{walk->-}=[
decoration={
  markings,
  mark=between positions 0 and 1 step 3pt with {
    \node[fill=black,
    signal from=west,
    signal to=east,
    signal,
    transform shape,
    minimum width=1pt,
    minimum height=3pt,
    inner sep=0]{};},
  mark=at position 0.55 with {\arrow{Triangle[scale=1.5]}}},
decorate]
\tikzstyle{crawl}=[
decoration={
  markings,
  mark=between positions 0 and 1 step 4pt with {
    \node[fill=black,
    signal from=west,
    signal to=east,
    signal,
    transform shape,
    minimum width=1pt,
    minimum height=2pt,
    inner sep=0] {};}},
decorate]

\tikzstyle{erroredge}=[red, decorate,
decoration={crosses,
  segment length=2mm,
  raise=0pt,
  shape size=2mm}]
\tikzstyle{erroredgesubtle}=[red, decorate,
decoration={crosses,
  segment length=2mm,
  raise=0pt,
  shape size=1mm}]
\tikzstyle{errorvertexsubtle}=[vertex, red]
\tikzstyle{errorvertex}=[fill=red, minimum size=3mm]

\tikzstyle{walk-g>-}=[
decoration={
  markings,
  mark=between positions 0 and 1 step 3pt with {
    \node[fill=black,
    signal from=west,
    signal to=east,
    signal,
    transform shape,
    minimum width=1pt,
    minimum height=3pt,
    inner sep=0]{};},
  mark=at position 0.55 with {\arrow{Triangle[scale=1.5, fill=white]}}},
decorate]

\tikzstyle{walk-g>>-}=[
decoration={
  markings,
  mark=between positions 0 and 1 step 3pt with {
    \node[fill=black,
    signal from=west,
    signal to=east,
    signal,
    transform shape,
    minimum width=1pt,
    minimum height=3pt,
    inner sep=0]{};},
  mark=at position 0.50 with {\arrow{Triangle[scale=1.5, fill=white]}},
  mark=at position 0.60 with {\arrow{Triangle[scale=1.5, fill=white]}}},
decorate]

\tikzstyle{walk-g<-}=[
decoration={
  markings,
  mark=between positions 0 and 1 step 3pt with {
    \node[fill=black,
    signal from=east,
    signal to=west,
    signal,
    transform shape,
    minimum width=1pt,
    minimum height=3pt,
    inner sep=0]{};},
  mark=at position 0.55 with {\arrow{Triangle[scale=1.5, fill=white]}}},
decorate]

\tikzstyle{walk-g<3-}=[
decoration={
  markings,
  mark=between positions 0 and 1 step 3pt with {
    \node[fill=black,
    signal from=east,
    signal to=west,
    signal,
    transform shape,
    minimum width=1pt,
    minimum height=3pt,
    inner sep=0]{};},
  mark=at position 0.5 with {\arrow{Triangle[scale=1.5, fill=white]}},
  mark=at position 0.55 with {\arrow{Triangle[scale=1.5, fill=white]}},
  mark=at position 0.6 with {\arrow{Triangle[scale=1.5, fill=white]}}},
decorate]

\tikzstyle{errorcirc}=[circle, draw=black, fill=red, minimum size=4mm, inner sep=0pt, line width=0.7pt]

\newcommand{\gsNoglue}{\begin{tikzpicture}
    \node[vertex, label={45+90*1:$e^0_1$}] at (45+90*1:1.5) (v1) {};
    \node[vertex, label={45+90*2:$e^0_2$}] at (45+90*2:1.5) (v2) {};
    \node[vertex, label={45+90*3:$e^0_3$}] at (45+90*3:1.5) (v3) {};
    \node[vertex, label={45+90*4:$e^0_4$}] at (45+90*4:1.5) (v4) {};

    \draw[->-] (v1) -- (v2);
    \node at (90*2:1.6) {$e^1_1$};

    \draw[->-] (v2) -- (v3);
    \node at (90*3:1.6) {$e^1_2$};

    \draw[->-] (v3) -- (v4);
    \node at (90*4:1.6) {$e^1_3$};

    \draw[->-] (v4) -- (v1);
    \node at (90*1:1.6) {$e^1_4$};

    \node (f) at (0,0) {$e^2$};
    \node[scale=3, rotate=-45] at (f.center) {$\circlearrowleft$};
    \begin{pgfonlayer}{background}
      \fill[\colora, fill opacity=0.5] (v1.center) -- (v2.center) -- (v3.center) -- (v4.center) -- cycle;
    \end{pgfonlayer}
  \end{tikzpicture}
}

\newcommand{\gsCylinder}{\begin{tikzpicture}
    \node[vertex, label={45+90*1:$e^0_1$}] at (45+90*1:1.5) (v1) {};
    \node[vertex, label={45+90*2:$e^0_2$}] at (45+90*2:1.5) (v2) {};
    \node[vertex, label={45+90*3:$e^0_2$}] at (45+90*3:1.5) (v3) {};
    \node[vertex, label={45+90*4:$e^0_1$}] at (45+90*4:1.5) (v4) {};

    \draw[-g>-] (v1) -- (v2);
    \node at (90*2:1.6) {$e^1_1$};

    \draw[->-] (v2) -- (v3);
    \node at (90*3:1.6) {$e^1_2$};

    \draw[-g>-] (v4) -- (v3);
    \node at (90*4:1.6) {$e^1_1$};

    \draw[->-] (v4) -- (v1);
    \node at (90*1:1.6) {$e^1_4$};

    \node (f) at (0,0) {$e^2$};
    \node[scale=3, rotate=-45] at (f.center) {$\circlearrowleft$};
    \begin{pgfonlayer}{background}
      \fill[\colora, fill opacity=0.5] (v1.center) -- (v2.center) -- (v3.center) -- (v4.center) -- cycle;
    \end{pgfonlayer}
  \end{tikzpicture}
}

\newcommand{\gsMobius}{\begin{tikzpicture}
    \node[vertex, label={45+90*1:$e^0_1$}] at (45+90*1:1.5) (v1) {};
    \node[vertex, label={45+90*2:$e^0_2$}] at (45+90*2:1.5) (v2) {};
    \node[vertex, label={45+90*3:$e^0_1$}] at (45+90*3:1.5) (v3) {};
    \node[vertex, label={45+90*4:$e^0_2$}] at (45+90*4:1.5) (v4) {};

    \draw[-g>-] (v1) -- (v2);
    \node at (90*2:1.6) {$e^1_1$};

    \draw[->-] (v2) -- (v3);
    \node at (90*3:1.6) {$e^1_2$};

    \draw[-g>-] (v3) -- (v4);
    \node at (90*4:1.6) {$e^1_1$};

    \draw[->-] (v4) -- (v1);
    \node at (90*1:1.6) {$e^1_4$};

    \node (f) at (0,0) {$e^2$};
    \node[scale=3, rotate=-45] at (f.center) {$\circlearrowleft$};
    \begin{pgfonlayer}{background}
      \fill[\colora, fill opacity=0.5] (v1.center) -- (v2.center) -- (v3.center) -- (v4.center) -- cycle;
    \end{pgfonlayer}
  \end{tikzpicture}
}

\newcommand{\gsTorus}{\begin{tikzpicture}
    \node[vertex, label={45+90*1:$e^0_1$}] at (45+90*1:1.5) (v1) {};
    \node[vertex, label={45+90*2:$e^0_1$}] at (45+90*2:1.5) (v2) {};
    \node[vertex, label={45+90*3:$e^0_1$}] at (45+90*3:1.5) (v3) {};
    \node[vertex, label={45+90*4:$e^0_1$}] at (45+90*4:1.5) (v4) {};

    \draw[-g>-] (v1) -- (v2);
    \node at (90*2:1.6) {$e^1_1$};

    \draw[-g>>-] (v2) -- (v3);
    \node at (90*3:1.6) {$e^1_2$};

    \draw[-g>-] (v4) -- (v3);
    \node at (90*4:1.6) {$e^1_1$};

    \draw[-g>>-] (v1) -- (v4);
    \node at (90*1:1.6) {$e^1_2$};

    \node (f) at (0,0) {$e^2$};
    \node[scale=3, rotate=-45] at (f.center) {$\circlearrowleft$};
    \begin{pgfonlayer}{background}
      \fill[\colora, fill opacity=0.5] (v1.center) -- (v2.center) -- (v3.center) -- (v4.center) -- cycle;
    \end{pgfonlayer}
  \end{tikzpicture}
}

\newcommand{\gsKlein}{\begin{tikzpicture}
    \node[vertex, label={45+90*1:$e^0_1$}] at (45+90*1:1.5) (v1) {};
    \node[vertex, label={45+90*2:$e^0_1$}] at (45+90*2:1.5) (v2) {};
    \node[vertex, label={45+90*3:$e^0_1$}] at (45+90*3:1.5) (v3) {};
    \node[vertex, label={45+90*4:$e^0_1$}] at (45+90*4:1.5) (v4) {};

    \draw[-g>-] (v1) -- (v2);
    \node at (90*2:1.6) {$e^1_1$};

    \draw[-g>>-] (v2) -- (v3);
    \node at (90*3:1.6) {$e^1_2$};

    \draw[-g>-] (v4) -- (v3);
    \node at (90*4:1.6) {$e^1_1$};

    \draw[-g>>-] (v4) -- (v1);
    \node at (90*1:1.6) {$e^1_2$};

    \node (f) at (0,0) {$e^2$};
    \node[scale=3, rotate=-45] at (f.center) {$\circlearrowleft$};
    \begin{pgfonlayer}{background}
      \fill[\colora, fill opacity=0.5] (v1.center) -- (v2.center) -- (v3.center) -- (v4.center) -- cycle;
    \end{pgfonlayer}
  \end{tikzpicture}
}

\newcommand{\gsProj}{\begin{tikzpicture}
    \node[vertex, label={45+90*1:$e^0_1$}] at (45+90*1:1.5) (v1) {};
    \node[vertex, label={45+90*2:$e^0_1$}] at (45+90*2:1.5) (v2) {};
    \node[vertex, label={45+90*3:$e^0_1$}] at (45+90*3:1.5) (v3) {};
    \node[vertex, label={45+90*4:$e^0_1$}] at (45+90*4:1.5) (v4) {};

    \draw[-g>-] (v1) -- (v2);
    \node at (90*2:1.6) {$e^1_1$};

    \draw[-g>>-] (v2) -- (v3);
    \node at (90*3:1.6) {$e^1_2$};

    \draw[-g>-] (v3) -- (v4);
    \node at (90*4:1.6) {$e^1_1$};

    \draw[-g>>-] (v4) -- (v1);
    \node at (90*1:1.6) {$e^1_2$};

    \node (f) at (0,0) {$e^2$};
    \node[scale=3, rotate=-45] at (f.center) {$\circlearrowleft$};
    \begin{pgfonlayer}{background}
      \fill[\colora, fill opacity=0.5] (v1.center) -- (v2.center) -- (v3.center) -- (v4.center) -- cycle;
    \end{pgfonlayer}
  \end{tikzpicture}
}

\newcommand*{\gate}[3]{\node[box] at (m-#1-#2) {$#3$};}
\newcommand*{\cgate}[4]{
  \node[vertex] at (m-#1-#3) {};
  \node[box] at (m-#2-#3) {$#4$};
  \begin{pgfonlayer}{background}
    \draw (m-#1-#3) -- (m-#2-#3);
  \end{pgfonlayer}
}
\newenvironment{qcirc}[4][1]{
  \begin{tikzpicture}[scale=#1, every node/.style={scale=#1}]
    \tikzstyle{box}=[fill=white, draw=black, minimum size=4pt, inner sep=2pt]
    \tikzstyle{measure}=[fill=white, draw=black, shape=semicircle, minimum size=0.8em, inner sep=0pt]
    \tikzstyle{textonly}=[fill=white, draw=none]

    \newcommand*{\maxx}{#2}
    \newcommand*{\maxy}{#3}
    \newcommand*{\phys}{#4}

    \foreach \x in {0,...,\maxx}{
      \foreach \y in {0,...,\maxy}{
        \coordinate (m-\y-\x) at (0.8*\x, -0.7*\y);
      }}

    \foreach \y [evaluate={\i=int(\y-\phys+1); \m=\maxx-1}] in {\phys,...,\maxy}{
      \node[textonly] at ($(m-\y-1)-(4mm,0)$) {$\ket 0_\i$};
      \node[measure] (meas\y) at (m-\y-\m) {};
      \draw[-{Latex[scale=0.7]}] (meas\y.south) -- ($(meas\y)+(60:3.5mm)$);
      \node[textonly] at ($(m-\y-\maxx)+(4mm,0)$) {$s_\i$};
      \begin{pgfonlayer}{background}
        \draw (m-\y-0) -- (m-\y-\m);
        \draw[double] (m-\y-\m) -- (m-\y-\maxx);
      \end{pgfonlayer}
    }
  }{
    \begin{pgfonlayer}{background}
      \pgfmathparse{int(\phys-1)}
      \foreach \y in {0,...,\pgfmathresult}{
        \draw (m-\y-0) -- ($(m-\y-\maxx)+(0.7,0)$);
      }
    \end{pgfonlayer}
  \end{tikzpicture}
}

\usepackage{xstring}

\newcommand{\drawRPHIH}[2][2cm]{
  \begin{tikzpicture}[]
    \IfSubStr{#2}{A}{\tikzstyle{v1}=[fill=cyan]}{\tikzstyle{v1}=[]}
    \IfSubStr{#2}{B}{\tikzstyle{v2}=[fill=cyan]}{\tikzstyle{v2}=[]}
    \IfSubStr{#2}{C}{\tikzstyle{v3}=[fill=cyan]}{\tikzstyle{v3}=[]}
    \IfSubStr{#2}{D}{\tikzstyle{v4}=[fill=cyan]}{\tikzstyle{v4}=[]}
    \IfSubStr{#2}{E}{\tikzstyle{v5}=[fill=cyan]}{\tikzstyle{v5}=[]}
    \IfSubStr{#2}{F}{\tikzstyle{v6}=[fill=cyan]}{\tikzstyle{v6}=[]}

    \foreach \v [evaluate={\n=int(mod(\v,5)+1); \a=int(\n+1)}] in {1,...,5}{
      \node[vertex, v\v] at (162+2*72*\v:#1) (v\v-1) {};
    }
    \node[vertex, v6] at (0,0) (v6) {};

    \node[vertex, v1] at ($(v2-1) +(90+72:#1)$) (v1-2) {};
    \node[vertex, v4] at ($(v5-1) +(90+2*72:#1)$) (v4-2) {};
    \node[vertex, v2] at ($(v3-1) +(90+3*72:#1)$) (v2-2) {};
    \node[vertex, v5] at ($(v1-1) +(90+4*72:#1)$) (v5-2) {};
    \node[vertex, v3] at ($(v4-1) +(90:#1)$) (v3-2) {};

    \node[draw=none] at ($(v1-2) +(90+3*72:0.5*#1)$) (f1) {};
    \node[draw=none] at ($(v2-2) +(90:0.5*#1)$) (f2) {};
    \node[draw=none] at ($(v3-2) +(90+2*72:0.5*#1)$) (f3) {};
    \node[draw=none] at ($(v4-2) +(90+4*72:0.5*#1)$) (f4) {};
    \node[draw=none] at ($(v5-2) +(90+72:0.5*#1)$) (f5) {};

    \node[draw=none] at ($(v6) +(126:0.5*#1)$) (f6) {};
    \node[draw=none] at ($(v6) +(126+4*72:0.5*#1)$) (f7) {};
    \node[draw=none] at ($(v6) +(126+3*72:0.5*#1)$) (f8) {};
    \node[draw=none] at ($(v6) +(126+2*72:0.5*#1)$) (f9) {};
    \node[draw=none] at ($(v6) +(126+72:0.5*#1)$) (f10) {};

    \foreach \f in {1,...,10}{
      \node[scale=1.6] at (f\f.center) {$\circlearrowright$};
    }

    \IfSubStr{#2}{01}{\tikzstyle{e01}=[double,red]}{\tikzstyle{e01}=[]}
    \IfSubStr{#2}{02}{\tikzstyle{e02}=[double,red]}{\tikzstyle{e02}=[]}
    \IfSubStr{#2}{03}{\tikzstyle{e03}=[double,red]}{\tikzstyle{e03}=[]}
    \IfSubStr{#2}{04}{\tikzstyle{e04}=[double,red]}{\tikzstyle{e04}=[]}
    \IfSubStr{#2}{05}{\tikzstyle{e05}=[double,red]}{\tikzstyle{e05}=[]}
    \IfSubStr{#2}{06}{\tikzstyle{e06}=[double,red]}{\tikzstyle{e06}=[]}
    \IfSubStr{#2}{07}{\tikzstyle{e07}=[double,red]}{\tikzstyle{e07}=[]}
    \IfSubStr{#2}{08}{\tikzstyle{e08}=[double,red]}{\tikzstyle{e08}=[]}
    \IfSubStr{#2}{09}{\tikzstyle{e09}=[double,red]}{\tikzstyle{e09}=[]}
    \IfSubStr{#2}{10}{\tikzstyle{e10}=[double,red]}{\tikzstyle{e10}=[]}
    \IfSubStr{#2}{11}{\tikzstyle{e11}=[double,red]}{\tikzstyle{e11}=[]}
    \IfSubStr{#2}{12}{\tikzstyle{e12}=[double,red]}{\tikzstyle{e12}=[]}
    \IfSubStr{#2}{13}{\tikzstyle{e13}=[double,red]}{\tikzstyle{e13}=[]}
    \IfSubStr{#2}{14}{\tikzstyle{e14}=[double,red]}{\tikzstyle{e14}=[]}
    \IfSubStr{#2}{15}{\tikzstyle{e15}=[double,red]}{\tikzstyle{e15}=[]}

    \draw[->-,e11] (v2-1) -- node[midway,left] {} (v6);
    \draw[->-,e12] (v4-1) -- node[midway,above] {} (v6);
    \draw[->-,e13] (v1-1) -- node[midway,right] {} (v6);
    \draw[->-,e14] (v3-1) -- node[midway,right] {} (v6);
    \draw[->-,e15] (v5-1) -- node[midway,below] {} (v6);

    \draw[->-,e06] (v2-1) -- node[midway,above] {} (v4-1);
    \draw[->-,e07] (v1-1) -- node[midway,right] {} (v4-1);
    \draw[->-,e08] (v1-1) -- node[midway,below] {} (v3-1);
    \draw[->-,e09] (v3-1) -- node[midway,left] {} (v5-1);
    \draw[->-,e10] (v2-1) -- node[midway,above left] {} (v5-1);

    \draw[-g>>-,e01] (v1-2) -- node[midway,above] {$\mathbf e_1$} (v2-1);
    \draw[-g>-,e02] (v2-1) -- node[midway,above] {$\mathbf e_2$} (v3-2);

    \draw[-g>-,e03] (v3-2) -- node[midway,right] {$\mathbf e_3$} (v4-1);
    \draw[-g>-,e04] (v4-1) -- node[midway,above right] {$\mathbf e_4$} (v5-2);

    \draw[-g>-,e05] (v5-2) -- node[midway,below] {$\mathbf e_5$} (v1-1);
    \draw[-g>>-,e01] (v1-1) -- node[midway,below right] {$\mathbf e_1$} (v2-2);

    \draw[-g>-,e02] (v2-2) -- node[midway,below left] {$\mathbf e_2$} (v3-1);
    \draw[-g>-,e03] (v3-1) -- node[midway,below] {$\mathbf e_3$} (v4-2);

    \draw[-g>-,e04] (v4-2) -- node[midway,above left] {$\mathbf e_4$} (v5-1);
    \draw[-g>-,e05] (v5-1) -- node[midway,left] {$\mathbf e_5$} (v1-2);

    \begin{pgfonlayer}{background}
      \IfSubStr{#2}{alpha}{
        \draw[fill=\colora, fill opacity=0.7, draw=none] (v1-2.center) to (v2-1.center) -- (v5-1.center);
      }{}
      \IfSubStr{#2}{beta}{
        \draw[fill=\colora, fill opacity=0.7, draw=none] (v2-2.center) to (v3-1.center) -- (v1-1.center);
      }{}
      \IfSubStr{#2}{gamma}{
        \draw[fill=\colora, fill opacity=0.7, draw=none] (v3-2.center) to (v2-1.center) -- (v4-1.center);
      }{}
      \IfSubStr{#2}{delta}{
        \draw[fill=\colora, fill opacity=0.7, draw=none] (v4-2.center) to (v3-1.center) -- (v5-1.center);
      }{}
      \IfSubStr{#2}{epsilon}{
        \draw[fill=\colora, fill opacity=0.7, draw=none] (v5-2.center) to (v4-1.center) -- (v1-1.center);
      }{}
      \IfSubStr{#2}{etta}{
        \draw[fill=\colora, fill opacity=0.7, draw=none] (v6.center) to (v2-1.center) -- (v5-1.center);
      }{}
      \IfSubStr{#2}{theta}{
        \draw[fill=\colora, fill opacity=0.7, draw=none] (v6.center) to (v2-1.center) -- (v4-1.center);
      }{}
      \IfSubStr{#2}{iota}{
        \draw[fill=\colora, fill opacity=0.7, draw=none] (v6.center) to (v1-1.center) -- (v4-1.center);
      }{}
      \IfSubStr{#2}{kappa}{
        \draw[fill=\colora, fill opacity=0.7, draw=none] (v6.center) to (v1-1.center) -- (v3-1.center);
      }{}
      \IfSubStr{#2}{lambda}{
        \draw[fill=\colora, fill opacity=0.7, draw=none] (v6.center) to (v5-1.center) -- (v3-1.center);
      }{}
    \end{pgfonlayer}
  \end{tikzpicture}
}

\newcommand{\drawRPHIHnolab}[2][2cm]{
  \begin{tikzpicture}
    \IfSubStr{#2}{A}{\tikzstyle{v1}=[fill=cyan]}{\tikzstyle{v1}=[]}
    \IfSubStr{#2}{B}{\tikzstyle{v2}=[fill=cyan]}{\tikzstyle{v2}=[]}
    \IfSubStr{#2}{C}{\tikzstyle{v3}=[fill=cyan]}{\tikzstyle{v3}=[]}
    \IfSubStr{#2}{D}{\tikzstyle{v4}=[fill=cyan]}{\tikzstyle{v4}=[]}
    \IfSubStr{#2}{E}{\tikzstyle{v5}=[fill=cyan]}{\tikzstyle{v5}=[]}
    \IfSubStr{#2}{F}{\tikzstyle{v6}=[fill=cyan]}{\tikzstyle{v6}=[]}

    \foreach \v [evaluate={\n=int(mod(\v,5)+1); \a=int(\n+1)}] in {1,...,5}{
      \node[vertex, v\v] at (162+2*72*\v:#1) (v\v-1) {};
    }
    \node[vertex, v6] at (0,0) (v6) {};

    \node[vertex, v1] at ($(v2-1) +(90+72:#1)$) (v1-2) {};
    \node[vertex, v4] at ($(v5-1) +(90+2*72:#1)$) (v4-2) {};
    \node[vertex, v2] at ($(v3-1) +(90+3*72:#1)$) (v2-2) {};
    \node[vertex, v5] at ($(v1-1) +(90+4*72:#1)$) (v5-2) {};
    \node[vertex, v3] at ($(v4-1) +(90:#1)$) (v3-2) {};

    \node[draw=none] at ($(v1-2) +(90+3*72:0.4*#1)$) (f1) {};
    \node[draw=none] at ($(v2-2) +(90:0.4*#1)$) (f2) {};
    \node[draw=none] at ($(v3-2) +(90+2*72:0.4*#1)$) (f3) {};
    \node[draw=none] at ($(v4-2) +(90+4*72:0.4*#1)$) (f4) {};
    \node[draw=none] at ($(v5-2) +(90+72:0.4*#1)$) (f5) {};

    \node[draw=none] at ($(v6) +(126:0.6*#1)$) (f6) {};
    \node[draw=none] at ($(v6) +(126+4*72:0.6*#1)$) (f7) {};
    \node[draw=none] at ($(v6) +(126+3*72:0.6*#1)$) (f8) {};
    \node[draw=none] at ($(v6) +(126+2*72:0.6*#1)$) (f9) {};
    \node[draw=none] at ($(v6) +(126+72:0.6*#1)$) (f10) {};

    \foreach \f in {1,...,10}{
      \node[scale=1.4] at (f\f.center) {$\circlearrowright$};
    }

    \IfSubStr{#2}{01}{\tikzstyle{e01}=[double,red]}{\tikzstyle{e01}=[]}
    \IfSubStr{#2}{02}{\tikzstyle{e02}=[double,red]}{\tikzstyle{e02}=[]}
    \IfSubStr{#2}{03}{\tikzstyle{e03}=[double,red]}{\tikzstyle{e03}=[]}
    \IfSubStr{#2}{04}{\tikzstyle{e04}=[double,red]}{\tikzstyle{e04}=[]}
    \IfSubStr{#2}{05}{\tikzstyle{e05}=[double,red]}{\tikzstyle{e05}=[]}
    \IfSubStr{#2}{06}{\tikzstyle{e06}=[double,red]}{\tikzstyle{e06}=[]}
    \IfSubStr{#2}{07}{\tikzstyle{e07}=[double,red]}{\tikzstyle{e07}=[]}
    \IfSubStr{#2}{08}{\tikzstyle{e08}=[double,red]}{\tikzstyle{e08}=[]}
    \IfSubStr{#2}{09}{\tikzstyle{e09}=[double,red]}{\tikzstyle{e09}=[]}
    \IfSubStr{#2}{10}{\tikzstyle{e10}=[double,red]}{\tikzstyle{e10}=[]}
    \IfSubStr{#2}{11}{\tikzstyle{e11}=[double,red]}{\tikzstyle{e11}=[]}
    \IfSubStr{#2}{12}{\tikzstyle{e12}=[double,red]}{\tikzstyle{e12}=[]}
    \IfSubStr{#2}{13}{\tikzstyle{e13}=[double,red]}{\tikzstyle{e13}=[]}
    \IfSubStr{#2}{14}{\tikzstyle{e14}=[double,red]}{\tikzstyle{e14}=[]}
    \IfSubStr{#2}{15}{\tikzstyle{e15}=[double,red]}{\tikzstyle{e15}=[]}

    \draw[->-,e11] (v2-1) -- node[midway,left] {} (v6);
    \draw[->-,e12] (v4-1) -- node[midway,above] {} (v6);
    \draw[->-,e13] (v1-1) -- node[midway,right] {} (v6);
    \draw[->-,e14] (v3-1) -- node[midway,right] {} (v6);
    \draw[->-,e15] (v5-1) -- node[midway,below] {} (v6);

    \draw[->-,e06] (v2-1) -- node[midway,above] {} (v4-1);
    \draw[->-,e07] (v1-1) -- node[midway,right] {} (v4-1);
    \draw[->-,e08] (v1-1) -- node[midway,below] {} (v3-1);
    \draw[->-,e09] (v3-1) -- node[midway,left] {} (v5-1);
    \draw[->-,e10] (v2-1) -- node[midway,above left] {} (v5-1);

    \draw[-g>>-,e01] (v1-2) -- node[midway,above] {} (v2-1);
    \draw[-g>-,e02] (v2-1) -- node[midway,above] {} (v3-2);

    \draw[-g>-,e03] (v3-2) -- node[midway,right] {} (v4-1);
    \draw[-g>-,e04] (v4-1) -- node[midway,above right] {} (v5-2);

    \draw[-g>-,e05] (v5-2) -- node[midway,below] {} (v1-1);
    \draw[-g>>-,e01] (v1-1) -- node[midway,below right] {} (v2-2);

    \draw[-g>-,e02] (v2-2) -- node[midway,below left] {} (v3-1);
    \draw[-g>-,e03] (v3-1) -- node[midway,below] {} (v4-2);

    \draw[-g>-,e04] (v4-2) -- node[midway,above left] {} (v5-1);
    \draw[-g>-,e05] (v5-1) -- node[midway,left] {} (v1-2);

    \begin{pgfonlayer}{background}
      \IfSubStr{#2}{alpha}{
        \draw[fill=\colora, draw=none] (v1-2.center) to (v2-1.center) -- (v5-1.center);
      }{}
      \IfSubStr{#2}{beta}{
        \draw[fill=\colora, draw=none] (v2-2.center) to (v3-1.center) -- (v1-1.center);
      }{}
      \IfSubStr{#2}{gamma}{
        \draw[fill=\colora, draw=none] (v3-2.center) to (v2-1.center) -- (v4-1.center);
      }{}
      \IfSubStr{#2}{delta}{
        \draw[fill=\colora, draw=none] (v4-2.center) to (v3-1.center) -- (v5-1.center);
      }{}
      \IfSubStr{#2}{epsilon}{
        \draw[fill=\colora, draw=none] (v5-2.center) to (v4-1.center) -- (v1-1.center);
      }{}
      \IfSubStr{#2}{etta}{
        \draw[fill=\colora, draw=none] (v6.center) to (v2-1.center) -- (v5-1.center);
      }{}
      \IfSubStr{#2}{theta}{
        \draw[fill=\colora, draw=none] (v6.center) to (v2-1.center) -- (v4-1.center);
      }{}
      \IfSubStr{#2}{iota}{
        \draw[fill=\colora, draw=none] (v6.center) to (v1-1.center) -- (v4-1.center);
      }{}
      \IfSubStr{#2}{kappa}{
        \draw[fill=\colora, draw=none] (v6.center) to (v1-1.center) -- (v3-1.center);
      }{}
      \IfSubStr{#2}{lambda}{
        \draw[fill=\colora, draw=none] (v6.center) to (v5-1.center) -- (v3-1.center);
      }{}
    \end{pgfonlayer}
  \end{tikzpicture}
}

\newenvironment{RPhalfplane}[4]{%
  \begin{tikzpicture}
    \newcommand*{\crad}{#1}
    \newcommand*{\arrowoffset}{#2}
    \newcommand*{\arrowstart}{#3}
    \newcommand*{\arrowsec}{#4}

    \tikzstyle{vertex}=[fill=black,draw=black,circle,minimum size=4pt,inner sep=0]
    \tikzstyle{compedge}=[
    postaction={decorate},
    line width=0.7pt,
    draw]
    \tikzstyle{->-}=[compedge,
    decoration={
      markings,
      mark=at position 0.55 with {\arrow{Triangle[scale=1]}}}]
    \tikzstyle{-g>-}=[compedge,
    decoration={
      markings,
      mark=at position 0.55 with {\arrow{Triangle[scale=1.5, fill=white]}}}]

    \pgfdeclarelayer{background}
    \pgfdeclarelayer{halfplane}
    \pgfdeclarelayer{foreground}
    \pgfsetlayers{background,halfplane,main,foreground}

    \coordinate (o) at (0,0);

    \begin{pgfonlayer}{halfplane}
      \draw[dashed, line width=0.7pt] (o) circle (\crad);
      \draw[-g>-] (\arrowstart:\crad+\arrowoffset)
      arc (\arrowstart:\arrowstart+\arrowsec:\crad+\arrowoffset);
      \draw[-g>-] (\arrowstart+180:\crad+\arrowoffset)
      arc (\arrowstart+180:\arrowstart+\arrowsec+180:\crad+\arrowoffset);
    \end{pgfonlayer}
  }{
  \end{tikzpicture}
}

\setlist{nosep,
  listparindent=\parindent,
  parsep=0pt
}

\counterwithin{equation}{section}

\declaretheoremstyle[
  headfont=\normalfont\bfseries,
  numberwithin=section,
  bodyfont=\normalfont,
  headformat={$\vartriangleright$ \NAME\ \NUMBER \NOTE},
  qed={$\vartriangleleft$},
]{mydefstyle}

\declaretheorem{definition}[style=mydefstyle, title={Definition}, refname={definition,definitions}, Refname={Definition,Definitions}]
\declaretheorem{notation}[style=mydefstyle, title={Notation}, refname={notation,notations}, Refname={Notation,Notations}, sibling=definition]
\declaretheorem{convention}[style=mydefstyle, title={Convention}, refname={convention,conventions}, Refname={Convention,Conventions}, sibling=definition]

\declaretheoremstyle[
  headfont=\normalfont\bfseries,
  sibling=definition,
  bodyfont=\normalfont,
  headformat={$\blacktriangleright$ \NAME\ \NUMBER \NOTE},
  qed={$\blacktriangleleft$},
]{mythmstyle}

\declaretheorem{theorem}[style=mythmstyle, title=Theorem, refname={theorem,theorems}, Refname={Theorem,Theorems}]
\declaretheorem{lemma}[style=mythmstyle, title=Lemma, refname={lemma,lemmas}, Refname={Lemma,Lemmas}]
\declaretheorem{corollary}[style=mythmstyle, title=Corollary, refname={corollary,corollaries}, Refname={Corollary,Corollaries}]
[style=mythmstyle, title=Proposition, refname={proposition,propositions}, Refname={Proposition,Propositions}]

\declaretheoremstyle[headfont=\normalfont\bfseries,sibling=definition,bodyfont=\normalfont]{mynotestyle}

\declaretheorem{note}[style=mynotestyle, title=Note, refname={note,notes}, Refname={Note,Notes}]
\declaretheorem{example}[style=mynotestyle, title=Example, refname={example,examples}, Refname={Example,Examples}]
[style=mynotestyle, title={Non-example}, refname={{non-example},{non-examples}}, Refname={{Non-example},{Non-examples}}]

\counterwithout{footnote}{chapter}

\newlength{\currentparindent}
\newcommand{\diagramonright}[4]{%
  \setlength{\currentparindent}{\parindent}
  \noindent
  \begin{minipage}{#1\linewidth}
    \baselineskip=18pt plus1pt
    \setlength{\parindent}{\currentparindent}
    #3
  \end{minipage}%
  \hfill %
  \begin{minipage}{#2\linewidth}
    #4
  \end{minipage}\\
}
\newcommand{\hiddenparbreak}{{\parfillskip=0pt\parskip=0pt\par}}

\newenvironment{caution}{%
  \setlength{\currentparindent}{\parindent}
  \vspace{1ex}
  \noindent
  \begin{minipage}{0.9\linewidth}
    \baselineskip=18pt plus1pt
    \setlength{\parindent}{\currentparindent}
  }{
  \end{minipage}%
  \hfill %
  \begin{minipage}{0.09\linewidth}
    \centering
    \begin{tikzpicture}
      \node at (90:0.5mm) {\Huge !};
      \node (v0) at (90:0.7cm) {};
      \node (v1) at (90+120:0.7cm) {};
      \node (v2) at (90+2*120:0.7cm) {};
      \begin{pgfonlayer}{background}
        \draw[line width=1mm, red, fill=yellow, rounded corners] (v0.center) -- (v1.center) -- (v2.center) -- cycle;
      \end{pgfonlayer}
    \end{tikzpicture}
  \end{minipage}
  \\
}

\renewcommand{\vec}[1]{\underline{#1}}
\newcommand{\cvert}{\mathbf{v}}
\newcommand{\cedge}{\mathbf{e}}
\newcommand{\cface}{\mathbf{f}}
\newcommand{\ket}[1]{\left|#1\right\rangle}
\newcommand{\neket}[1]{|#1\rangle}

\newcommand{\nebra}[1]{\langle#1|}
\newcommand{\braket}[2]{\left\langle#1\middle|#2\right\rangle}
\newcommand{\nebraket}[2]{\langle#1|#2\rangle}
\newcommand{\deq}{\coloneqq}
\newcommand{\eqd}{\eqqcolon}

\newcommand{\idealeq}{\trianglelefteq}
\newcommand{\epito}{\twoheadrightarrow}
\newcommand{\monoto}{\rightarrowtail}

\makeatletter
\newcommand{\oset}[3][0ex]{%
  \mathrel{\mathop{#3}\limits^{
    \vbox to#1{\kern-2\ex@
    \hbox{$\scriptstyle#2$}\vss}}}}
\makeatother
\newcommand{\isoto}{\oset[-0.6ex]{\!\!\sim}{\to}}

\DeclareMathOperator{\End}{End}

\DeclareMathOperator{\Stab}{Stab}

\DeclareMathOperator{\Tor}{Tor}
\DeclareMathOperator{\id}{id}
\DeclareMathOperator{\dom}{dom}

\DeclareMathOperator{\im}{im}

\DeclareMathOperator{\rk}{rk}

\newcommand{\R}{\mathbb{R}}
\newcommand{\C}{\mathbb{C}}
\newcommand{\N}{\mathbb{N}}
\newcommand{\Z}{\mathbb{Z}}

\newcommand{\Pauli}{\mathfrak{P}}

\newcommand{\RMod}[1]{\ensuremath{\prescript{}{#1}{\mathbf{Mod}}}}

\newcommand{\Ch}[1]{\ensuremath{\mathbf{Ch}_\bullet(#1)}}

\newcommand{\genR}[2]{\left\langle #2 \right\rangle_{#1}}

\makeatletter
\newsavebox{\@brx}
\newcommand{\llangle}[1][]{\savebox{\@brx}{\(\m@th{#1\langle}\)}%
  \mathopen{\copy\@brx\kern-0.5\wd\@brx\usebox{\@brx}}}
\newcommand{\rrangle}[1][]{\savebox{\@brx}{\(\m@th{#1\rangle}\)}%
  \mathclose{\copy\@brx\kern-0.5\wd\@brx\usebox{\@brx}}}
\makeatother

\makeatletter
\DeclareFontEncoding{LS2}{}{\@noaccents}
\makeatother
\DeclareFontSubstitution{LS2}{stix2}{m}{n}
\DeclareSymbolFont{largesymbolsstix}{LS2}{stixex}{m}{n}
\DeclareMathDelimiter{\lBrace}{\mathopen}{largesymbolsstix}{"E8}{largesymbolsstix}{"0E}
\DeclareMathDelimiter{\rBrace}{\mathclose}{largesymbolsstix}{"E9}{largesymbolsstix}{"0F}
\DeclareMathDelimiter{\lParen}{\mathopen}{largesymbolsstix}{"DE}{largesymbols}{"02}
\DeclareMathDelimiter{\rParen}{\mathclose}{largesymbolsstix}{"DF}{largesymbols}{"03}

\title{Homological\\[0.5ex] Quantum Error Correction\\[1ex] with Torsion}

\author{Samo Nov\'ak}
\college{Mathematical Institute and Mansfield College}

\renewcommand{\submittedtext}{A dissertation submitted for the degree of}
\degree{MSc Mathematics and Foundations of Computer Science}
\degreedate{Trinity \& Summer 2023}

\begin{document}

\baselineskip=18pt plus1pt

\setcounter{secnumdepth}{3}
\setcounter{tocdepth}{3}

\maketitle                  
\begin{abstract}
  Homological quantum error correction uses tools of algebraic topology and homological algebra to derive Calderbank-Shor-Steane quantum error correcting codes from cellulations of topological spaces. This work is an exploration of the relevant topics, a journey from classical error correction, through homology theory, to CSS codes acting on qudit systems. Qudit codes have torsion in their logical spaces. This is interesting to study because it gives us extra logical qudits, of possibly different dimension.

  Apart from examples and comments on the topic, we prove an original result, the \emph{Structure Theorem for the Qudit Logical Space}, an application of the Universal Coefficient Theorem from homological algebra, which gives us information about the logical space when torsion is involved, and that improves on a previous result in the literature. Furthermore, this work introduces our own abstracted and restricted version of the general notion of a cell complex, suited exactly to our needs.
\end{abstract}
\begin{acknowledgements}
  First and foremost, I would like to thank my supervisor, Alex Cowtan, without whom this project would not have been possible. Thank you for your support and much advice during this work. I would also like to thank Cole Comfort for his helpful input at the beginning of the project.

  This year would never have been as great as it was without my friends in MFoCS: Ari, Ben, Clarice, Edwin, Isi, Jean, Mira, No\'e, Pali, Satya, Seb, and Zoe. I will always remember our times together, lunches and pub nights, Tins \& Talks, running and board games, and especially our trip to Cornwall. Thank you for being my home away from home. In~addition, I would like to acknowledge the MFoCS Counting Society, current number: 6004.

  I would like to thank the members of my college rowing club, with whom we shared early mornings of splashy rowing, as well as days of racing: Lorenzo, James, Tomas, Karim, Conor, Ky, Jude, Albi, Sam, Amir, Fergus, Kat, Bruno, James, Johanna, Maritsa, and Emma.

  I would also like to thank my friends outside Oxford who always have my back: Mart\!'us, Mišo, Debča, Kika, Kaca, Klára, Katka, Katka, Tia, Táňa, Marek, Mat\!'a, Pet\!'o, Šlojma, Lucretia, Paula, Carina, Orges, Pablo, Alex, and Alex. A particular thank you goes to Štefka, who set me off on my academic path.

  Na záver sa chcem pod\!'akovat\!' tým, ktorí ma poznajú najdlhšie -- svojej rodine. Ďakujem mamke, ockovi, babke, Ferovi, Martine, Tomášovi a Lucii za Vašu neustálu lásku a podporu počas môjho roka na Oxforde i~mimo neho.
\end{acknowledgements}  
{
  \thispagestyle{empty}
  \vspace*{\stretch{1}}
  \itshape             
  \centering          
  Babke.
  \vspace{\stretch{3}} 
  \clearpage           
}

\begin{romanpages}          
  \tableofcontents            
  \listoffigures              
\end{romanpages}            

\chapter{Introduction}

Computation in the real world is an inherently physical process. Our computers, be they classical or quantum, are devices made of physical components, and it is possible for them to go wrong due to outside interference. An example of this may be an electrically charged cosmic ray impacting a transistor in a computer, flipping a bit from $0$ to $1$ or vice-versa. In the setting of quantum computation, many architectures are so fragile that they have to be cooled down to extremely low temperatures and be very well shielded from the environment because they might get unintentionally entangled with the outside world. What is more, even if it were physically possible, we cannot completely cut ties between the environment and the device, because we need a way to control it and extract results.
In short, the influence of the environment is inevitable, and errors are a fact of life. In this work, we turn to detecting, and if possible correcting the errors that do happen.

We focus on a class of error detection and correction protocols, called \emph{codes}, acting on quantum computing systems. These are \emph{Calderbank-Shor-Steane (CSS)} codes, a subclass of \emph{stabilizer codes} which are defined using an abelian group of operators, called \emph{stabilizers}, acting on the quantum system in such a way that \mbox{error-free} states are fixed. A key characteristic of CSS codes is that the generators of the stabilizer group can be split into two types corresponding to the complementary bases of a quantum system.\cite{aaronson_intro_to_qis2} The fact that these have to commute allows us to define a \emph{chain complex} to represent the code. This is a structure from homological algebra, and we call codes that can be constructed this way \emph{homological}.\cite{bravyi2013homological}
One way of constructing homological codes is to look at various topological spaces and divide them into cells. Such \emph{cellulations} also correspond to chain complexes, and this forms a connection between algebraic topology, homological algebra, and quantum error correction that allows us to graphically construct CSS codes with good properties.\cite{10.1063/1.1499754}

Furthermore, we are interested not only in qubit-based computers, where the basic building blocks can be in a superposition of two distinct basis states but more generally in qudit systems, with an arbitrary $d \ge 2$.\cite{sarkar2023qudit} One reason to think about qudits is that our ability to build large quantum computers is limited, so we might as well stuff as much information into the small computers we can build.
The case of qudits is, however, particularly challenging when $d$ is not prime, because this breaks many nice properties we are used to. However, this introduces \emph{torsion} into the logical space of the code. This means, essentially, that we may get more encoded logical qudits than we normally would. However, these extra qudits are different, in that their dimension is lower than $d$.\cite{vuillot2023homological}

Our goal is to study homological quantum CSS codes in the qudit systems, with arbitrary $d\ge2$, such that torsion is possible. As mentioned, this requires tools from several different areas of mathematics, particularly the theory of modules, homological algebra, and algebraic topology. These fields were new to the author, so this is also an exploration of those topics, with the aim of using the material learned to study homological codes. We invite the reader on this journey with us.

Throughout the work, we explore and introduce the respective topics. We provide original examples and commentary, and as a useful tool, we define our own abstract topological construction to help reason about cellulations -- this is the \emph{abstract $2$-dimensional cell complex} from \Cref{def:abstract-2-cell-complex}. This is an abstraction and restriction of the general concept of a cell complex in the literature.

Furthermore, we prove an original result in \Cref{thm:structure-theorem-for-qudit-logical-space} which we call the \emph{Structure Theorem for the Qudit Logical Space} in analogy to the Structure Theorem of Modules over a principal ideal domain (PID). We choose this name because our theorem is essentially the same as the latter transported to the case of a qudit logical space, which is a $\Z_d$-module. Note that $\Z_d$ is not necessarily a PID, and this result is only possible due to the structure of a homological CSS code derived from a cellulation of a topological space. The proof requires connecting several non-trivial lemmas and theorems that we collect along the way. Our structure theorem extends upon a result from ref. \cite{vuillot2023homological} about the decomposition of the logical space, transporting it from rotor codes to qudits. and it allows us to precisely find out how many logical qudits a code has, and what their dimensions are.

\section{Structure of the dissertation}

In \Cref{part:intr-error-corr}, we lay down the basics of error correction. First, we explore classical codes, specifically those defined using the language of linear algebra, in \Cref{chap:classical-linear-codes}. Then we move on to the quantum setting, and we introduce stabilizer and CSS codes in \Cref{chap:quantum-errors}.

Then in \Cref{part:alg-top-hom-alg}, we introduce the mathematical tools needed to reason about homological CSS codes on qudits. In \Cref{chap:rings-and-modules}, we review some theory of rings and modules. We follow with an introduction to homological algebra in \Cref{chap:homol-algebra}, where we study chain complexes and homology in the abstract, without any topological meaning. Finally, in \Cref{cha:cell-compl}, we define our original notion of the \emph{abstract $2$-dimensional cell complex}, and explore how it relates to other topological complexes in the literature, and to homology.

In \Cref{part:homol-quant-error}, we use those tools to construct homological quantum error correcting codes from cellulations of spaces. In \Cref{chap:CSS-codes-from-cellulations}, we explain this first in the case of qubits by constructing a simple example of a toric code. Then we prove that generalizing to qudits still allows us to use the same methodology, and this leads into \Cref{chap:torsion}, where we cellulate the real projective plane to obtain codes with torsion in their logical space. We provide a practical guide on how to construct a cellulation. This is where we also prove our original \Cref{thm:structure-theorem-for-qudit-logical-space}.

Finally, in \Cref{chap:conclusion}, we conclude and reflect on related work that one might study in the future. Outside of the main body, we give a brief introduction to merging two homological codes in \Cref{chap:joining-codes}.

As a supplement to one of the computations in \Cref{chap:torsion}, we provide an excerpt from a Jupyter\cite{Kluyver:2016aa} notebook, using Sage\cite{sagemath}, a Computer Algebra System, to help us compute properties of matrices in a chain complex. This can be found in \Cref{cha:jupyter-sage}.

\part[Introduction to Error Correction]{Introduction to Error Correction \\[4cm] \begin{tikzpicture}
  \tikzstyle{point}=[draw=black, line width=0.5mm, circle, minimum size=4mm, inner sep=0]
  \tikzstyle{bullet}=[point, fill=white]
  \tikzstyle{circ}=[point, fill=black]
  \renewcommand*{\bullet}{\node[bullet] {};}
  \renewcommand*{\circ}{\node[circ] {};}

  \pgfdeclarelayer{background}
  \pgfdeclarelayer{foreground}
  \pgfsetlayers{background,main,foreground}

  \begin{pgfonlayer}{foreground}
    \matrix [matrix of nodes, column sep=3mm, row sep=3mm] (mat)
    {
      \bullet & \bullet & \circ   & \bullet &[3mm] \bullet & \circ   & \circ   \\
      \bullet & \circ   & \bullet & \bullet & \circ   & \bullet & \circ   \\
      \circ   & \bullet & \bullet & \bullet & \circ   & \circ   & \bullet \\
    };
  \end{pgfonlayer}
  \foreach \i in {0,1,2}{
    \ifnum \i=0 \tikzstyle{cl}=[fill=magenta]
    \else \ifnum \i=1 \tikzstyle{cl}=[fill=\colora]
      \else \tikzstyle{cl}=[fill=\colorb]
      \fi \fi
    \begin{pgfonlayer}{background}
      \fill [cl, fill opacity=0.5] ($(mat.center)+(90+120*\i:1cm)$) circle[radius=2cm];
    \end{pgfonlayer}
    \begin{pgfonlayer}{main}
      \draw ($(mat.center)+(90+120*\i:1cm)$) circle[radius=2cm];
    \end{pgfonlayer}
  }
\end{tikzpicture}}
\label{part:intr-error-corr}
\chapter[Classical Linear Codes]{Classical Linear Codes \hspace*{\fill} \begin{tikzpicture}[baseline=-0.6ex]
  \tikzstyle{point}=[draw=black, line width=0.7pt, circle, minimum size=1.5mm, inner sep=0]
  \tikzstyle{circ}=[point, fill=white]
  \tikzstyle{bullet}=[point, fill=black]
  \renewcommand*{\bullet}{\node[bullet] {};}
  \renewcommand*{\circ}{\node[circ] {};}

  \matrix (mat)
  [
  column sep=0.7mm,
  row sep=0.7mm]
  {
    \node[circ] (m1) {}; & \circ & \bullet \\
    \circ & \bullet & \node[circ] (m2) {}; \\
  };
  \useasboundingbox (m1.north west) rectangle (m2.south east);
\end{tikzpicture}}
\label{chap:classical-linear-codes}

The transfer and processing of information is a physical process, and as such it is subject to unwanted outside influence. This means that random errors may occur during transfer of some signal, or during some computation. Ideally, we want to detect when this happens, and also correct the error, so that we may read the message, or carry on with the computation.

A protocol for encoding information in such a way that error can be detected, and perhaps corrected, is called an \emph{error correcting code}. There are multiple different classes of codes: in this work, we focus on \emph{linear (block) codes}, protocols defined using linear algebra. The word \emph{block} refers to partitioning a sequence of information bits into blocks of fixed size and encoding each separately -- we will not care much for this, because our aim will be to encode a fixed number of (qu)bits, so we refer to them as just \emph{linear codes}.\footnote{The partition into blocks may, in fact, be relevant if we want to encode a larger system: in this case, we just split the system into disjoint partitions, and encode each separately (see \Cref{sec:gluing-complexes} and \Cref{sec:connected-sum-of-codes}).}

We start in the \emph{classical} (i.e. not quantum) setting, and then move to quantum codes in \Cref{chap:quantum-errors}. To describe classical codes, we follow ref. \cite{ryan_lin_2009} with some minor changes to conventions, in particular using column instead of row vectors, and largely omitting generator matrices.

\section{General idea}

We work with binary systems, where the unit of information is one \emph{bit}. The natural setting is then the finite field of two elements $\Z_2 \deq \Z/2\Z = \{ 0, 1 \}$.\footnote{Formally, these should be $[0]$ and $[1]$ as they are equivalence classes. We abuse notation to unburden ourselves.} We represent words of information bits by vectors with entries in $\Z_2$.

\begin{definition}[words]
  We call \emph{logical} the bits of actual information, which form \emph{logical words} of length $k \in \N$. They are vectors from $\Z_2^k$ which we call the \emph{logical space}. We call \emph{physical} the bits and words used to encode the logical words. These have length $n \ge k$, and they live in the \emph{physical space} $\Z_2^n$. We encode logical words into physical words. The physical words obtained this way are called \emph{codewords} and live in a \emph{codespace} $\mathfrak{C} \subseteq \Z_2^n$.
\end{definition}

\begin{note}[linearity]
  We have defined the words as elements of linear spaces. This is no accident, and the linearity is a desirable feature. In particular, it means that a linear combination of codewords is again a codeword.
\end{note}

We require that the encoding is invertible: we can always uniquely decode a codeword to the corresponding logical word. This implies that there is an isomorphism $\mathfrak{C} \cong \Z_2^k$, i.e. $\mathfrak{C}$ has dimension $k$. This further means that there exists a monomorphism $\Z_2^k \monoto \Z_2^n$. We represent this by a matrix, called the \emph{generator matrix} of the code. This is what encodes logical words into codewords. Generator matrices are important in the classical error correction literature, however, we have no use for them, so we only briefly mention them.

\subsection{Error detection}
\label{sec:error-detection}

We have a monomorphism $\Z_2^k \monoto \Z_2^n$ that encodes logical words in physical words. We can always decode a codeword. However, there are possibly many more physical words than just codewords. These are word where bits have been randomly flipped -- errors have occurred. Our aim is to find out that this is the case, and ideally correct the error, thus successfully recovering logical information.

\begin{definition}[parity check, syndrome]
  We define a morphism $P : \Z_2^n \to \Z_2^m$, for $m \in \N$, that takes a physical word $\vec v$ and tells us whether this contains an error. We call is the \emph{parity check matrix}. By definition, $P\vec v = \vec 0$ if and only if $\vec v$ is a codeword. Otherwise, the nonzero vector $P\vec v$ tells us about the kind of error that happened. We call $P\vec v$ the \emph{syndrome} of $\vec v$. We use this term in both cases, if $\vec v$ is a codeword, or if there is an error.
\end{definition}

If no error has occurred, and only then, we want the syndrome to be the zero vector. This means that we require the kernel to be exactly the codespace: $\ker P = \mathfrak{C}$.
Using the rank-nullity theorem, we see that $m$, the dimension of the codomain of $P$, must satisfy:
\begin{equation}
  \label{eq:dimension-of-syndrome-space}
  m \ge \dim \im P = \dim \dom P - \dim \ker P = n - k,
\end{equation}
where $\dom P = \Z_2^k$ is the domain of $P$, and $\im P$ is the image of $P$.

\subsection{Interpretation of the parity check matrix}
\label{sec:interpretation-of-the-parity-matrix}

The parity check matrix $P$ sends a physical word $\vec v \in \Z_2^n$ to its \emph{syndrome} $\vec s \deq P \vec v$ which is a vector of $m$ components. We interpret these as the outcomes of individual parity checks, or \emph{measurements} of the physical word. Concretely, the component $s_i$, for $i\in\{1, \dots, m\}$, is
\begin{equation}
  \label{eq:parity-check-row}
  s_i = \sum_{\alpha=1}^n P_{i,\alpha} v_\alpha = \langle \vec{P}_{i}, \vec v\rangle,
\end{equation}
where $\vec P_i^\top$ is the $i^{\text{th}}$ row of $P$, seen as a row vector, $\vec P_i$ its transpose (column vector), and $\langle -, - \rangle$ is the usual inner product in $\Z_2^n$.

The parity check matrix describes a procedure for performing error detection. Each row $\vec{P}_i^\top$ defines a single measurement of the physical word, essentially listing which individual bits of the physical word to measure. The interpretation as measurements is perhaps strange in the context of classical codes. However, it is important, because in \Cref{sec:css-codes}, when dealing with quantum codes, this is what will tell us how to measure a quantum system.

\subsection{Characterizing codes}
\label{sec:code-metrics}

We now introduce a handy way to quickly characterize an error correcting code. First, we need the following definitions:
\begin{definition}[Hamming weight]
  Let $m \in \N$, and let $\vec v = (v_1, \dots, v_m) \in \Z_2^m$. Define the \emph{Hamming weight} (or just \emph{weight}) of $\vec v$, denoted $\mathfrak{w}(\vec v)$, to be the number of non-zero components of $\vec v$. Note that defining the vectors as elements of $\Z_2^m$ and reasoning in terms of tuples, we implicitly use the standard basis.
\end{definition}

\begin{definition}[Hamming distance]
  Let $\vec v = (v_i)_i, \vec w = (w_i)_i \in \Z_2^m$. Define the \emph{Hamming distance} (or just \emph{distance}) of $\vec v$ and $\vec w$, denoted $\mathfrak{d}(\vec v, \vec w)$, to be the number of components where $\vec v$ and $\vec w$ differ.
\end{definition}

Note that Hamming distance is a \emph{metric} on $\Z_2^m$. This means that $\mathfrak{d}(\vec v, \vec w) = 0$ for $\vec v = \vec w$, and $\mathfrak{d}(\vec v, \vec w) > 0$ otherwise; it is symmetric ($\mathfrak{d}(\vec v, \vec w) = \mathfrak{d}(\vec w, \vec v)$), and it satisfies the \emph{triangle inequality}
\[ \mathfrak{d}(\vec u, \vec v) + \mathfrak{d}(\vec v, \vec w) \ge \mathfrak{d}(\vec u, \vec w) \]
for all $\vec u, \vec v, \vec w \in \Z_2^n$.

Notice also that Hamming weight is a kind of \emph{norm} on $\Z_2^m$.\footnote{It is only a norm in the case of vector spaces over $\Z_2$, where the only possible non-zero value of components is $1$.} Similarly to other pairs of metric and norm, the distance between two vectors is the weight of their difference:
\begin{equation}
  \label{eq:hamming-distance-vs-hamming-weight}
  \mathfrak{d}(\vec v, \vec w) = \mathfrak{w}(\vec w - \vec v).
\end{equation}

\begin{definition}[code distance]
  The minimum Hamming distance between two different codewords in $\mathfrak{C} \le \Z_2^n$ is called the \emph{code distance} (name from \cite{cowtan2023css}), and we denote it
  \[
    \mathfrak{d} = \mathfrak{d}(\mathfrak{C})
    \deq \min \big\{ \mathfrak{d}(\vec v, \vec w) : \vec v, \vec w \in \mathfrak{C}, \vec v \ne \vec w \big\}.
    \qedhere
  \]
\end{definition}
Using \cref{eq:hamming-distance-vs-hamming-weight}, we can rewrite the definition above as
\[
  \mathfrak{d}
  = \min \big\{ \mathfrak{w}(\vec w - \vec v) : \vec v, \vec w \in \mathfrak{C}, \vec v \ne \vec w \big\}
  = \min \big\{ \mathfrak{w}(\vec x) : \vec x \in \mathfrak{C}, \vec x \ne \vec 0 \big\},
\]
so the code distance is, in fact, the minimum Hamming weight of a non-zero vector in the codespace.

The code distance is a property of the code, and together with the dimensions of logical and physical space, they describe the code:
\begin{definition}[code metrics, code rate]
  Suppose we have a code that encodes words of length $k$ into words of length $n$, and its code distance is $\mathfrak{d}$. These numbers are called the \emph{code metrics}, and we write that this is an $[n,k,\mathfrak{d}]$ code. Sometimes, we may omit the distance.
  Additionally, define the \emph{code rate} $k/n$, interpreted as the average amount of information carried by a physical bit.
\end{definition}
This characterization is not unique, and there may be multiple codes that have the same metrics. For example, if $G : \Z_2^k \monoto \Z_2^n$ is a generating matrix of a code, then there are in general many composites of the form
\[
  \begin{tikzcd}
    {\Z_2^k} & {\Z_2^k} & {\Z_2^n} & {\Z_2^n,}
    \arrow["G", tail, from=1-2, to=1-3]
    \arrow["\varphi", "\sim"', from=1-1, to=1-2]
    \arrow["\psi", "\sim"', from=1-3, to=1-4]
  \end{tikzcd}
\]
where $\varphi$ and $\psi$ are isomorphisms (i.e. changes of basis). There are corresponding changes of basis for the parity check matrix $P$.

There is a decision to be made whether to consider such composites separate codes. This is clear if the code distance is changed by the change of basis. However, suppose that the code distance is unchanged, and $\psi = \id_{\Z_2^n}$. Then the codespaces $\im(\varepsilon \circ \varphi) \subseteq \Z_2^n$ are the same for any $\varphi$. As such, the codes arising could be considered \emph{essentially the same}. A similar argument could be made for the case when $\psi$ is not necessarily the identity -- the codespaces may be different, but isomorphic. On the other hand, in either case, the encoding and error detection are actually different for different choices of $\varphi, \psi$.

\begin{convention}
  \label{con:different-codes}
  We follow the convention of ref. \cite{cowtan2023css}: we declare that codes are only the same if their encoding, and the parity check matrices, are equal.
\end{convention}

\subsection{Detection versus decoding}

We now have all the machinery needed for error detection, and perhaps correction. What is the difference, though?

If the syndrome is non-zero, that means the code \emph{detected} some kind of error. There exist, of course, \emph{undetectable errors}: if an error takes a codeword to another codeword, the code has no way to notice. However, if an error is detectable, we will always know when it occurs.

Knowing whether an error occurred, however, may not always be enough information. If the syndrome uniquely determines the kind of error that happened, then we may undo, or \emph{correct} it. This is also called \emph{decoding}, because it gives us back the logical information.

On the other hand, some syndromes may correspond to more than one kind of error. Then we can either assume one of the possibilities is much more probable, and take a leap of faith; or we can declare that an error has been detected, and leave it there. The choice always depends on the specific code. We will see an example of this in the following section.

\section{Examples}
\subsection{A [3,1,3] repetition code}
\label{sec:3-1-3-code}

We show a $[3,1,3]$ repetition code that copies the single logical bit three times. The code rate is $k/n = 1/3$. The only non-zero codeword is $(1,1,1)^\top$, which gives the code distance $\mathfrak{d} = \mathfrak{w}(1,1,1)^\top = 3$. The codespace is one-dimensional, and it is the span of $(1,1,1)^\top$.

We will now derive the parity check matrix $P$. By definition, the codespace is equal to the kernel of $P$, which means $P(1,1,1)^\top = \vec 0$. This parity matrix must always measure exactly two bits (otherwise the syndrome is nonzero even for $(1,1,1)^\top$). We require at least two measurements to compare all three bits this way. This gives us the following matrix $P$:
\begin{equation}
  \label{eq:3-1-3-parity-check-matrix}
  P =
  \begin{pmatrix}
    1 & 1 & 0 \\
    1 & 0 & 1
  \end{pmatrix}
\end{equation}
Note that we make a choice here. The rows could be swapped, and the columns could be also swapped. Parity matrices with such swaps would formally correspond to different codes under \Cref{con:different-codes}, though they would be \emph{essentially the same}.

The physical space is $\Z_2^3$, meaning it has $8$ elements. Thus it is convenient enough to enumerate all possible syndromes. We display them in the table in~\cref{fig:table-of-syndromes-3-1-3}.

\begin{figure}[h]
  \centering
  \begin{tabular}[h!]{|c|c|c|}
    \hline
    \textbf{word} $\vec v \in \Z_2^n$ & \textbf{syndrome} $P\vec v$ & \textbf{error} \\
    \hline
    $(0,0,0)^\top$ & \multirow{2}{*}{$(0,0)^\top$} & \multirow{2}{*}{no error, or three errors} \\
    $(1,1,1)^\top$ & & \\
    \hline
    $(0,0,1)^\top$ & \multirow{2}{*}{$(0,1)^\top$} & \multirow{2}{*}{error on bit $3$, or two errors on bits $1,2$} \\
    $(1,1,0)^\top$ & & \\
    \hline
    $(0,1,0)^\top$ & \multirow{2}{*}{$(1,0)^\top$} & \multirow{2}{*}{error on bit $2$, or two errors on bits $1,3$} \\
    $(1,0,1)^\top$ & & \\
    \hline
    $(1,0,0)^\top$ & \multirow{2}{*}{$(1,1)^\top$} & \multirow{2}{*}{error on bit $1$, or two errors on bits $2,3$} \\
    $(0,1,1)^\top$ & & \\
    \hline
  \end{tabular}
  \caption[Syndromes of the $\lbrack 3,1,3 \rbrack$ repetition code]{Syndromes of the $[3,1,3]$ repetition code.}
  \label{fig:table-of-syndromes-3-1-3}
\end{figure}

The table displays all the possible physical words, and all the possible ways to interpret their syndromes. We follow with a few observations. First, we focus on knowing about errors with certainty:

\begin{itemize}
\item The extreme case: three simultaneous errors make $\vec v$ look like a codeword. Such errors are called \emph{undetectable}, because the code has no way to find them. Undetectable errors are all those that take one codeword to another -- in other codes, this may not require every bit to get an error.
\item Secondly, the present $[3,1,3]$ code cannot tell apart whether one or two errors occurred. This means also that it cannot tell with certainty what the logical word is meant to be, or in other words, it cannot decode the physical word.

  For example, suppose the logical word is $\vec u = \vec 0$, and a single bit error occurs, giving us the physical word $\vec v = (0,0,1)^\top$. The physical word $(1,1,0)^\top$ has the same syndrome. We detect that an error occurred, either on the third bit only, or on the other two bits only. This is not enough to decide whether the logical word is supposed to be $0$ or $1$.
\end{itemize}

However, probability is an important part of error correction. Usually, we assume that a bit can randomly change value (\emph{random bit-flip}) with probability~$p$, independently of others. Then no error occurs with probability $(1-p)^3$, a single error with probability $p(1-p)^2$, two simultaneous errors with $p^2(1-p)$, and finally a three-bit error with probability $p^3$.

If $p$ is low enough so that $p^3$ is small, then we may declare three-bit errors unlikely and assume they never occur. Then we decode a codeword as the corresponding logical word. Furthermore, if $p$ is low enough to make $p^2(1-p)$ small, then we may assume the only errors that happen are single bit-flips. If we only get single errors, then we can also successfully correct them.

\subsection{The [7,4,3] Hamming codes}
\label{sec:7-4-3-hamming}

In this section, we show a different way to construct a code. We introduce a family of codes called the \emph{$[7,4,3]$ Hamming codes}. They encode $4$ logical bits using $7$ physical bits, meaning their code rate is $k/n = 4/7$, which is higher than the code rate of $1/3$ of the $[3,1,3]$ repetition code from \cref{sec:3-1-3-code}.
The members of this class are all essentially the same, but considered different under \Cref{con:different-codes}. For convenience, we choose a representative and refer to it as \emph{the $[7,4,3]$ Hamming code}.

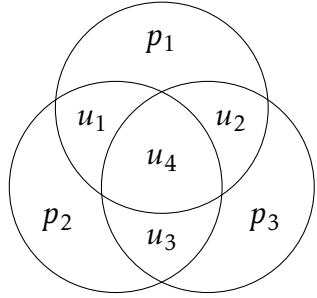
\begin{figure}[t]
  \centering
  \begin{tikzpicture}
  \newcommand*{\centralradius}{0.7cm}
  \node at (0,0) {$u_4$};
  \foreach \i [evaluate={\a=int(\i+1)}] in {0,1,2}{
    \draw (90+120*\i:\centralradius) circle[radius=2*\centralradius];
    \node at (150-120*\i:1.5*\centralradius) {$u_\a$};
    \node at (90+120*\i:2.25*\centralradius) {$p_\a$};
  }
\end{tikzpicture}
  \caption[Design of the $\lbrack 7,4,3 \rbrack$ Hamming code]{Design of the $[7,4,3]$ Hamming code.}
  \label{fig:design-of-7-4-3-Hamming}
\end{figure}
The code is based on the Venn diagram in \cref{fig:design-of-7-4-3-Hamming}. Each region corresponds to a bit in the physical word $\vec v = (u_1, \dots, u_4, p_1, \dots, p_3)^\top$, where $\vec u = (u_1, \dots, u_4)^\top$ is a logical word, and $p_1, \dots, p_3$ are called \emph{parity bits}.
Values of the parity bits are chosen to make the sum of each circle in the diagram even. In other words, the three circles give us a system of equations that have to be satisfied, and this determines the parity bits:
\begin{equation}
  \label{eq:7-4-3-parity-equations}
  \left\{
    \begin{matrix}
      u_1 & + u_2 &       & + u_4 & + p_1 &&& = 0 \\
      u_1 &       & + u_3 & + u_4 && + p_2 && = 0 \\
          & \hphantom+  u_2 & + u_3 & + u_4 &&& + p_3 & = 0
    \end{matrix}
  \right.
\end{equation}

In the system of equations, each $p_i$ is completely determined by the values of $\{u_j\}_j$ which allows us to determine the codewords: for a logical word $\vec u = (u_1, \dots, u_4)^\top$, the corresponding codeword is $\vec v^\top = (\vec u^\top | \vec p^\top)$, where we write the transpose $\vec v^\top$ just for convenience of notation, and where
\[ \vec p =
  \begin{pmatrix}
    u_1+u_2+u_4 \\ u_1+u_3+u_4 \\ u_2+u_3+u_4
  \end{pmatrix}.
\]

Furthermore, the system of equations (\ref{eq:7-4-3-parity-equations}) gives us the parity matrix. Each row corresponds to one of the equations. There are ones in columns corresponding to the variables present in the particular equation:
\[
  P = \begin{pmatrix}
        1 & 1 & 0 & 1 & 1 & 0 & 0 \\
        1 & 0 & 1 & 1 & 0 & 1 & 0 \\
        0 & 1 & 1 & 1 & 0 & 0 & 1
      \end{pmatrix}.
\]

\subsubsection*{Single-bit error syndromes}

The $[7,4,3]$ Hamming code can detect \emph{and correct} single-bit errors. Here, we show what the syndromes of single-bit errors look like.

Suppose $\vec u$ is a logical word, and $\vec v^\top = (\vec u^\top | \vec p^\top)$ is its codeword, as above. We add a single-bit error on bit $i \in \{ 1, \dots, 7\}$, giving us the physical word $\vec x = \vec v + \vec e_i$, where $\vec e_i$ is a vector of the standard basis.
Then we compute the syndrome:
\[ P\vec x = \cancel{P \vec v} + P \vec e_i = P \vec e_i, \]
where we note that, by definition, $P \vec v = 0$. This means that the syndrome of a single-bit error is $P \vec e_i$, the $i^{\text{th}}$ column of the matrix $P$. For example, if $i=2$, then $P \vec e_2 = (1,0,1)^\top$; if $i=7$, then $P \vec e_7 = (0,0,1)^\top$.

All columns of $P$ are nonzero, hence $P \vec e_i \ne \vec 0$. This means that all single-bit errors are detectable. Furthermore, all columns are different, so we can always determine where the error occurred, and correct it.

\subsubsection*{Two-bit error syndromes}

The code at hand can also detect two-bit errors, but not correct them. Here we show why both of those are the case.
Take the physical word $\vec x = \vec v + \vec e_i + \vec e_j$, where again $\vec v$ is a codeword, and $i \ne j$.\footnote{Depending on what we choose to model, we might also allow two errors to occur in the same place. In the case of bits, this undoes the error. However, from \Cref{sec:qudits-and-ops} onward, we work with (qu)dits, with possibly $d > 2$, meaning two errors in the same place may still be an error.} As before, the codeword part cancels, and we are left with the error terms:
\[ P\vec x = \cancel{P \vec v} + P(\vec e_i + \vec e_j) =  P(\vec e_i + \vec e_j). \]
The error syndrome is a sum of two columns of $P$. Under the assumption that $i\ne j$, i.e. we add two different columns, this is always nonzero. This means that there are no undetectable 2-bit errors.

However, if we allow two-bit errors, we lose the ability to perform correction. For instance, take the two-bit error $\vec e_1 + \vec e_2$ and a single-bit error $\vec e_3$. Their syndromes are the same:
\[
  P(\vec e_1 + \vec e_2)
  = \begin{pmatrix} 1 \\ 1 \\ 0 \end{pmatrix}
  + \begin{pmatrix} 1 \\ 0 \\ 1 \end{pmatrix}
  = \begin{pmatrix} 0 \\ 1 \\ 1 \end{pmatrix}
  = P \vec e_3.
\]
We conclude that decoding (with certainty) is not possible in this case, and in fact, any sum of two columns of $P$ is equal to another column of $P$.

\subsubsection*{More simultaneous errors}

Finally, we briefly note that simultaneous errors on more than two bits may be undetectable. Inspired by the previous example of a 2-bit error, suppose we have a physical word $\vec x = \vec v + \vec e_1 + \vec e_2 + \vec e_3$:
\[
  P\vec x = \cancel{P\vec v} + P(\vec e_1 + \vec e_2 + \vec e_3)
  = \begin{pmatrix} 1 \\ 1 \\ 0 \end{pmatrix}
  + \begin{pmatrix} 1 \\ 0 \\ 1 \end{pmatrix}
  + \begin{pmatrix} 0 \\ 1 \\ 1 \end{pmatrix}
  = \vec 0.
\]
An error like the above takes a codeword to another codeword, without the possibility of detection that this happened.

\section{Low-Density Parity-Check (LDPC) Codes}
\label{sec:LDPC}

In this section, we briefly mention an interesting class of linear codes, though we do not go into too much detail. The class is called \emph{Low-Density Parity-Check}, or \emph{LDPC}, codes. The parity check part is what we have already seen in the previous sections: these are linear codes, and we specifically describe them by using parity matrices.

The \emph{low-density} is the new part. It means the parity matrix $P$ has a low density of non-zero values. This can be formulated in terms of Hamming weights of rows and columns of $P$:
\begin{definition}[LDPC]
  A code is called a \emph{low-density parity-check (LDPC) code} if the density of ones in the corresponding parity matrix, defined using Hamming weights of each row and each column, is bounded by a constant.
\end{definition}

The above definition can be interpreted two different ways. It can be a statement about the weights of rows and columns of individual parity matrices. In this case, one would choose a small constant and say that the density of ones in the matrix is, e.g. 1 \%.\cite{ryan_lin_2009} The other view is that if we have a family of codes, for example, if we have a scheme to generate a code given $k$ or $n$, then the boundedness of the weights can be understood as the density of ones not increasing as we scale the code. In particular, LDPC codes have a constant code rate, as we scale them.\cite{cowtan2023css}

The rows of the parity matrix correspond to measurements of bits of the physical words. If each of them has a bounded Hamming weight, this means that there is a bound on the number of bits that are involved in any measurement.
Similarly, the columns of the parity matrix correspond to physical bits. If the rows have bounded weight, this means that each bit is measured only a bounded number of times.

In the classical world, measuring many bits, or having a bit be a part of many measurements, is not really a problem. However, we aim to do quantum error correction, where every operation and measurement is precious. This means that the constraint posed by the LDPC condition is a useful one. In this work, we will design quantum LDPC codes using cellulations of spaces.

\chapter[Quantum Error Correction]{Quantum Errors Correction \hspace*{\fill}\begin{tikzpicture}
  \tikzstyle{vertex}=[fill=black, draw=black, circle, minimum size=3pt, inner sep=0]
  \tikzstyle{box}=[fill=white, draw=black, minimum size=4pt, inner sep=0pt]
  \tikzstyle{tri}=[fill=white, draw=black, shape=isosceles triangle, rotate=180, minimum size=3pt, inner sep=0pt]
  \tikzstyle{measure}=[fill=white, draw=black, shape=semicircle, minimum size=3pt, inner sep=0pt]
  \foreach \x in {0,...,4}{
    \foreach \y in {0,1}{
      \coordinate (m-\y-\x) at (3mm*\x, -3mm*\y);
    }}
  \tikzstyle{ground}=[]
  \node[vertex] at (m-0-1) {};
  \node[box] at (m-0-3) {};

  \node[tri] at (m-1-0) {};
  \node[box] at (m-1-1) {};
  \node[measure] (meas) at (m-1-2) {};
  \draw[-{Latex[scale=0.5]}] (meas.south) -- ($(meas)+(60:1.8mm)$);
  \node[vertex] at (m-1-3) {};
  \node[tri, rotate=180] at (m-1-4) {};

  \begin{pgfonlayer}{background}
    \draw ($(m-0-0)-(1mm,0)$) -- ($(m-0-4)+(1mm,0)$);
    \draw (m-1-0) -- (m-1-2);
    \draw[double] (m-1-2) -- (m-1-4);
    \draw (m-0-1) -- (m-1-1);
    \draw[double] (m-1-3) -- (m-0-3);
  \end{pgfonlayer}
\end{tikzpicture}}
\label{chap:quantum-errors}

Quantum computation (and transmission of information) works in a fundamentally different way to classical computation. While in classical error correction, we measure the state of the message (which is a fancy way of saying we look at it), this cannot be done in the quantum setting without also destroying the computation state. Any superposition we may have had would collapse, and any entanglement would be destroyed.
This means we have to find a clever way to measure an error syndrome. The way to do this will be to perform partial measurements on a state enlarged by ancillae.

In this chapter, we assume knowledge of basic quantum mechanics and of quantum computation. An unfamiliar reader may find out more in \cite{berera_del-debbio_2021,MR1796805}. In introducing quantum error correction, we follow \cite{MR1796805,aaronson_intro_to_qis1,aaronson_intro_to_qis2,Devitt_2013}.

\section{Two-state system}
\label{sec:two-state-system}

In classical (non-quantum) computation, the basic unit that we can use is, usually, a two-state system called a \emph{bit}. Analogously, in quantum computation, at least on the basic level, works with \emph{qubits}, a quantum analogy of a bit. This is our starting point. Note, however, that in the quantum computing world, building blocks that are systems of more than two possible states are quite common, more so than in the classical case. We reach these in \Cref{sec:qudits-and-ops}.

\begin{definition}[spin\index{spin}]
  A two-state quantum system lives in a two-dimensional complex Hilbert space $\mathcal{H} = \C^2$. It is conventionally called a \emph{spin}. In the context of quantum computation, it is called a \emph{qubit}.
\end{definition}

\begin{definition}[standard basis\index{standard basis}]
  We choose the \emph{standard basis} of a single-qubit state space to be the \emph{eigenbasis} of the $Z$ Pauli operator (see \Cref{def:Pauli-matrices}). We denote $\ket 0$ the \emph{spin up} eigenstate, and we denote $\ket 1$ the \emph{spin down} eigenstate. The basis is written in ordered form $(\ket 0, \ket 1)$, so that we can write matrices.
\end{definition}

\begin{definition}[Pauli matrices\index{Pauli matrices}]
  \label{def:Pauli-matrices}
  If $\mathcal{H}$ is the state space of a two-state system, written in the standard basis, then we represent the Pauli operators as the following \emph{Pauli spin matrices}:
  \begin{align*}
    X = \sigma_1 & \deq \begin{pmatrix} 0 & 1 \\ 1 & 0 \end{pmatrix},
    & Y = \sigma_2 & \deq \begin{pmatrix} 0 & -i \\ i & 0 \end{pmatrix},
    & Z = \sigma_3 & \deq \begin{pmatrix} 1 & 0 \\ 0 & -1 \end{pmatrix}.
                     \qedhere
  \end{align*}
\end{definition}
\begin{definition}[Pauli group]
  \label{def:1-qubit-Pauli-group}
  The above matrices generate a group called the \emph{Pauli group} $\Pauli_1$, where the subscript means these act on a single qubit. The elements of $\Pauli_1$ are the Pauli matrices along with the identity $\bm 1_2$, and their multiples by $\pm 1$ and $\pm i$.
\end{definition}

\begin{lemma}[product, commutation]
  \label{lem:pauli-product-commutation}
  For Pauli matrices $\sigma_j$ and $\sigma_k$, their product is
  \[ \sigma_j \sigma_k = \delta_{jk} \bm1_2 + i \sum_{l} \varepsilon_{jkl} \sigma_l, \]
  where $\delta_{jk}$ is the \emph{Kronecker delta}, with value $1$ if $j=k$ and $0$ otherwise. The \emph{Levi-Civita symbol} $\varepsilon_{jkl}$ has values equal for cyclic permutations of indices, $\varepsilon_{123} = 1$, $\varepsilon_{213} = -1$, and is zero if indices repeat. Consequently, the commutation relation of Pauli matrices is
  \begin{equation}
    \label{eq:commutation-of-Paulis}
    \sigma_j \sigma_k = -\sigma_k \sigma_j
  \end{equation}
  if $j \ne k$. Obviously they commute if $j=k$.
\end{lemma}
This is a standard fact, so for the interest of brevity, we omit the proof. We prove a generalization of this later, in \Cref{lem:qudit-pauli-commutation}. An interested reader may verify the present Lemma by explicitly multiplying the matrices -- this is more efficient if one writes:
\[ \sigma_j =
  \begin{pmatrix}
    \delta_{j3} & \delta_{j1} - i\delta_{j2} \\
    \delta_{j1} + i\delta_{j2} & - \delta_{j3}
  \end{pmatrix}.
\]

\begin{definition}[$X$-basis]
  There is another common basis. It is the one made of eigenstates of Pauli $X$, with eigenstates $\ket + = \frac 1{\sqrt 2}(\ket 0 + \ket 1)$ and $\ket - = \frac 1{\sqrt 2}(\ket 0 - \ket 1)$, in the order $(\ket +, \ket -)$.
\end{definition}

\begin{definition}[Hadamard]
  The change of basis between $Z$- and $X$-eigenbases is called the \emph{Hadamard} operator $H$, written in the standard basis as
  \[ H = \frac{1}{\sqrt 2} \begin{pmatrix} 1 & 1 \\ 1 & -1 \end{pmatrix}. \]
  We can write the change of basis explicitly as $\ket + = H \ket 0$ and $\ket - = H \ket 1$.
\end{definition}

\subsection{Multiple qubits}

Obviously, a single qubit cannot do that much computation.
If we have $n \in \N$ qubits, then these live in a tensor product space $\mathcal{H} = (\C^2)^{\otimes n}$, which is isomorphic to $\C^{2^n}$. Similarly to the 1-qubit case, it will be useful to us to have a group of Pauli operators that act on multiple qubits.
\begin{definition}[Pauli group on $n$ qubits]
  \label{def:n-qubit-Pauli-group}
  Recall that $\Pauli_1$ is a group of operators acting on one qubit. We use the tensor product to define operators acting on a system of $n$ qubits. We denote $\Pauli_n$ the group of all $n$-fold tensor products of single-qubit Pauli operators. Explicitly,
  \(
    \Pauli_n \deq \left\{ P_1 \otimes \cdots \otimes P_n : P_i \in \Pauli_1  \right\}.
  \)
\end{definition}

Multiple qubits, without a way to let them interact, would be only as useful as classical bits. We need another ingredient, an that is an \emph{entangling operation}, one that acts on multiple qubits, previously perhaps in a \emph{product state} (not entangled), and creates entanglement between them. The canonical one is the \emph{controlled $X$ gate}:
\begin{definition}[CNOT]
  Let $\mathcal{H}$ be the state space of a 2-qubit system. A \emph{controlled~$X$}, also called a \emph{controlled NOT (CNOT)}, is an operator that sends
  \begin{align*}
    \ket 0 \otimes \neket \psi & \mapsto \ket 0 \otimes \neket \psi
    & \text{and} &
    & \ket 1 \otimes \neket \psi & \mapsto \ket 1 \otimes X\neket \psi,
  \end{align*}
  where $\neket\psi$ is a qubit.
\end{definition}

Apart from superposition, entanglement is one of the main advantages of quantum computing. The non-locality exhibited by joint systems is of great use to computation, and specifically in our case, to error correction.

\section{Stabilizer states}

In this section, we define the model of quantum computation, and the corresponding errors that we allow. We restrict our computers to use \emph{Clifford gates}, a specific group of quantum operators on multi-qubit systems. As such, the allowed states are \emph{stabilizer states}. We define these notions below.

\begin{definition}[stabilizer subgroup]
  Let $\mathcal{H}$ be a state space of a quantum system, and let $\mathfrak{G}$ be a group of operators $\mathcal{H} \to \mathcal{H}$ that act on this space. For a state $\neket\psi \in \mathcal{H}$, define the set of \emph{stabilizers} of $\neket\psi$ with respect to $\mathfrak{G}$ as
  \[
    \Stab_{\mathfrak{G}}\neket\psi \deq \big\{
    A \in \mathfrak{G} : A\neket\psi = \neket\psi
    \big\}.
  \]
  Equivalently, it contains all operators for which $\neket\psi$ is an eigenstate with eigenvalue~$+1$. We extend this to a subspace $\mathcal{K} \le \mathcal{H}$ as
  \begin{equation}
    \label{eq:stabilizer-of-subspace}
    \Stab_{\mathfrak{G}} \mathcal{K} \deq \big\{
    A \in \mathfrak{G} : \forall \neket\psi\in\mathcal{K},\; A\neket\psi = \neket\psi
    \big\},
  \end{equation}
  that is, $\Stab_{\mathfrak{G}}\mathcal{K}$ contains operators that stabilize every state in $\mathcal{K}$. The above sets are clearly subgroups of $\mathfrak{G}$, and we call them the \emph{stabilizer subgroups}.
\end{definition}

\begin{theorem}[Pauli stabilizers are abelian]
  \label{thm:Pauli-stab-abelian}
  Let $\mathcal{K} \le \mathcal{H} = (\C^2)^{\otimes n}$ be some nontrivial subspace of a state space of $n$ qubits. Then the group $\Stab_{\Pauli_n}\mathcal{K}$ is abelian.
\end{theorem}
\begin{proof}
  Following \cite[\S3.4.2]{aaronson_intro_to_qis2}.
  From \Cref{lem:pauli-product-commutation}, we know that single-qubit Pauli matrices always either commute or anticommute (i.e. $\sigma_j\sigma_k = -\sigma_k\sigma_j$). This extends to tensor products of Pauli matrices, because each factor acts only on its corresponding qubit.
  Suppose $A,B$ are stabilizers from $\Stab_{\Pauli_n}\mathcal{K}$, and assume that they do not commute -- then they anticommute ($AB = - BA$). As $\Stab_{\Pauli_m}\mathcal{K}$ is a group, it contains both $AB$ and $BA$, meaning that these products must also stabilize every state $\neket\psi$ in $\mathcal{K}$. Then
  \[ \neket\psi = AB\neket\psi = -BA\neket\psi = -\neket\psi, \]
  which is only possible if $\neket\psi = 0$. The state $\neket\psi \in\mathcal{K}$ was chosen arbitrarily, so this implies $\mathcal{K} = \bm 0$. This contradicts the hypothesis that $\mathcal{K}$ is non-trivial. We conclude that $A$ and $B$ must commute, and that $\Stab_{\Pauli_n}\mathcal{K}$ is abelian.
\end{proof}

\begin{definition}[Clifford group]
  Define the \emph{Clifford group} on $n$ qubits as the normalizer of $\Pauli_n$. This is a group that consists of unitary operators that commute with all of~$\Pauli_n$. The Clifford group is generated, for example, by the Hadamard $H$, the CNOT, and the phase gate
  \[
    S = \begin{pmatrix} 1 & 0 \\ 0 & i \end{pmatrix}.
  \]
\end{definition}

An element $A$ of the Clifford group is called a \emph{Clifford gate} if we can perform $A$ atomically, as a gate that our computer can do (e.g. the $Z$ gate, or the CNOT). In general, including composites, we call $A$ a \emph{Clifford circuit}.

\begin{definition}[stabilizer state]
  A quantum state of $n$ qubits $\neket\psi \in (\C^2)^{\otimes n}$ is called a \emph{stabilizer state} if it can be obtained from $\ket0^{\otimes n}$ by applying a Clifford circuit.\footnote{Originally, the definition of a stabilizer state requires that the cardinality of its stabilizer subgroup with respect to $\Pauli_n$ is $2^n$. We define them in terms of Clifford circuits, secretly applying the \emph{Gottesman-Knill theorem} which states that these are equivalent.\cite{aaronson_intro_to_qis2}}
\end{definition}

\subsection{Error model}

\begin{convention}[computation model]
  The quantum computers we consider are only allowed to perform Clifford gates, and thus every state is a stabilizer state.
\end{convention}

This is a standard convention often taken when reasoning about quantum error correction. Specifically, it allows us to use stabilizers to describe error correcting codes. We remark that this model of computation is not universal, but in this work we do not need it to be.

Similarly to the classical (non-quantum) case, we have random errors appearing in our quantum computation. These are operators on the state space that get randomly applied without our intent and knowledge. In this work, we assume that only the following errors are possible:
\begin{convention}[quantum errors]
  \label{con:quantum-errors}
  We assume that an error is a random application of an element of the group $\Pauli_n$. In particular, this implies that errors are \emph{discrete}, and they act on qubits \emph{independently}, without any new entanglement being created.
\end{convention}

The above is justified by the \emph{digitization} of quantum errors. This is a result stating that considering only the above errors is in fact viable, and by doing this, we can correct any error.\cite{digitization}

\section{CSS codes}
\label{sec:css-codes}

Unlike in the classical world, where we can look at the state of a system and figure out whether an error has occurred, we cannot measure a quantum system during computation. That would destroy any superposition and entanglement we may have had, and effectively halt the computation.

However, only qubits that are measured are destroyed. We are able to perform \emph{partial measurements}, which means we only measure some of the qubits. Of course, this still impacts the whole system, but unmeasured qubits can be used even after the partial measurement.

The way to detect errors in a quantum computer is to \emph{entangle} the system with new helper qubits, called \emph{ancillae}, initialized with known states, usually $\ket 0$ or $\ket + = H\ket 0$. This way, we enlarge the state space, but we also spread  information about the state to the ancillae. Then we measure the ancillae, collapsing the state space back to original. This process gives us information about the state without destroying it completely. It may not be full information, but if we entangled the system with the ancillae in a clever way, we learn enough to detect errors of certain kinds.

The aim of the theory of quantum error correction is to figure out concrete protocols to perform the above. Similarly to \Cref{chap:classical-linear-codes}, instead of covering quantum error correction in general, we immediately specify to the class of codes of our interest. These are called \emph{stabilizer codes}, and they use the stabilizer framework to describe error correcting protocols. 
\begin{definition}[stabilizer code]
  Let $\mathcal{H} = (\C^2)^{\otimes n}$ be the state space of $n$ physical qubits. A quantum error correcting code that describes the codespace $\mathcal{C} \le \mathcal{H}$,\footnote{We use the different notation $\mathcal{C} \le (\C^2)^{\otimes n}$ to distinguish the Hilbert space corresponding to correct physical words from the abstract representation of a codespace $\mathfrak{C} \le \Z_2^{n}$. The elements of $\mathfrak{C}$ correspond to the basis vectors in $\mathcal{C}$.} as well as the error detection measurements, using the stabilizer subgroup $\Stab_{\Pauli_n} \mathcal{C}$ is called a \emph{stabilizer code}. Concretely, a particular code is a choice of \emph{generators} of the stabilizer group.
\end{definition}

Following from  \Cref{lem:pauli-product-commutation}, we know that all elements of $\Pauli_n$ either commute or anticommute. In our model, errors are also members of that group. Let $A \in \Pauli_n$ be some operator. If $A$ anticommutes with at least one stabilizer in $\Stab_{\Pauli_n}\mathcal{C}$, then it is a detectable error. On the other hand, if $A$ commutes with all stabilizers, then it preserves the codespace, and it is a stabilizer itself. This could mean that is it a logical operator, but it could also be an undetectable error. In this sense, logical operators and undetectable errors are the same, and the difference is whether we intended to apply them

We are interested in a subclass of stabilizer codes, the \emph{Calderbank-Shor-Steane (CSS) codes}. They are defined by a partition on the generating set of their stabilizer group to two types of Pauli operators:

\begin{definition}[X-type, Z-type Pauli operators]
  Suppose $A = \bigotimes_{i=1}^n A_i$ is a Pauli operator from $\Pauli_n$. If each $A_i$ is either the identity $\bm 1$ or $X$, we call $A$ an \emph{$X$-type} operator. Similarly, if each $A_i$ is either $\bm 1$ or $Z$, we call $A$ a \emph{$Z$-type} operator.
\end{definition}

Note that we do not pay any mind to $Y$-type operators. This is exactly because $Y = iXZ$. Assigning importance specifically to the pair $X$ and $Z$, in contrast to e.g. $X$ and $Y$, is a choice -- one made by convention. In any case, we do need exactly two of these, because they are \emph{complementary}.

\begin{definition}[CSS code]
  Suppose we have a stabilizer code defined by a stabilizer group $\mathfrak{S}$. If we can find sets $S_X$ and $S_Z$ such that $S_X$ contains only $X$-type Pauli operators, and likewise $S_Z$ contains only $Z$-type operators, and such that $\mathfrak{S}$ is generated by $S_X \sqcup S_Z$, then we call this a \emph{Calderbank-Shor-Steane (CSS)~code}.
\end{definition}

The nice idea behind making this definition is that we can take two classical error correcting codes, use one of them to define $X$-type stabilizers, and the other one for $Z$-type stabilizers. That is, we use two classical codes to define the generating sets $S_X$ and $S_Z$, one each. Recall from \Cref{thm:Pauli-stab-abelian} that $\Stab_{\Pauli_n}\mathcal{K}$ is abelian for a nontrivial space $\mathcal{K}$. This means that not every two classical codes can be used for this: the stabilizer generators we obtain from them must all commute with each other, otherwise the codespace is trivial.

The stabilizers describe the encoding; and more concretely, they tell us how to perform syndrome measurements. We illustrate this on an example.

\subsection{Example: phase-flip code}

One of the basic quantum error correcting codes is the \emph{phase-flip} code. It encodes a single logical qubit ($k=1$) as three physical qubits ($n=3$). It can correct a random application of a Pauli $Z$, a phase flip, to a single physical qubit. It is one of the quantum versions of the $[3,1,3]$ repetition code from \Cref{sec:3-1-3-code}, the other being a \emph{bit-flip} code that corrects $X$ errors.

In general, we want to define a code by the stabilizer group. This determines the encoding as well. However, in this example, we start with an encoding, find the stabilizers, and then perform error correction.

\subsubsection*{Encoding}

Due to the \emph{No cloning theorem}, we cannot copy an arbitrary quantum state. Thus we cannot perform encoding as $\neket\varphi \mapsto \neket\varphi^{\otimes 3}$. Instead, we use the circuit in \Cref{fig:quantum-phase-flip-encoding} to create an entangled state. This generates the following codeword physical states:
\begin{align*}
  \ket 0
  & \mapsto \neket{\overline 0} \deq \ket{000},
  & \text{and} &
  & \ket 1
  & \mapsto \neket{\overline 1} \deq \ket{111}.
\end{align*}
Above, we use the overline notation to indicate that $\neket{\overline 0}$ is a physical state corresponding to the logical qubit $\ket 0$. We also use the notation $\ket{000}$ to mean the tensor product $\ket 0 \otimes \ket 0 \otimes \ket 0$, and similarly for $\ket{111}$.
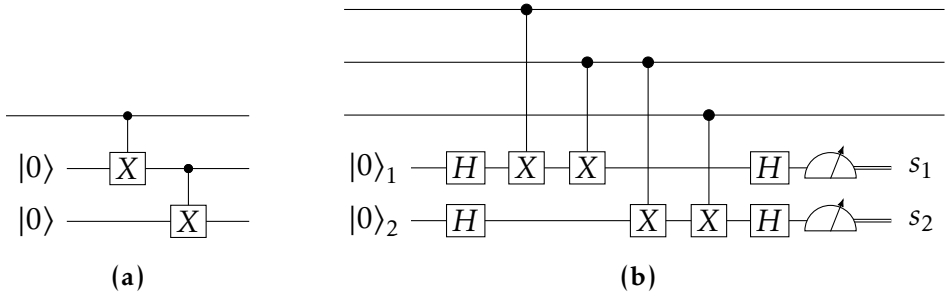
\begin{figure}[h]
  \centering
  \subfloat[\label{fig:quantum-phase-flip-encoding}]{%
    \begin{tikzpicture}
  \tikzstyle{vertex}=[fill=black, draw=black, circle, minimum size=3pt, inner sep=0]
  \tikzstyle{box}=[fill=white, draw=black, minimum size=4pt, inner sep=2pt]
  \tikzstyle{measure}=[fill=white, draw=black, shape=semicircle, minimum size=3pt, inner sep=0pt]

  \foreach \x in {0,...,3}{
    \foreach \y in {0,...,2}{
      \coordinate (m-\y-\x) at (0.8*\x, -0.7*\y);
    }}
  \node[vertex] at (m-0-1) {};
  \node[vertex] at (m-1-2) {};

  \node[fill=white] at ($(m-1-0)-(4mm,0)$) {$\ket 0$};
  \node[fill=white] at ($(m-2-0)-(4mm,0)$) {$\ket 0$};
  \node[box] at (m-1-1) {$X$};
  \node[vertex] (meas) at (m-1-2) {};
  \node[box] at (m-2-2) {$X$};

  \begin{pgfonlayer}{background}
    \draw ($(m-0-0)-(8mm,0)$) -- (m-0-3);
    \draw (m-1-0) -- (m-1-3);
    \draw (m-2-0) -- (m-2-3);

    \draw (m-0-1) -- (m-1-1);
    \draw (m-1-2) -- (m-2-2);
  \end{pgfonlayer}
\end{tikzpicture}
  }
  \qquad
  \subfloat[\label{fig:quantum-phase-flip-detection}]{%
    \begin{qcirc}{9}{4}{3}
      \gate32H \gate42H
      \cgate033X \cgate134X \cgate145X \cgate246X
      \gate37H \gate47H
    \end{qcirc}
  }
  \caption[Quantum phase-flip code]{Quantum phase-flip code. The encoding circuit is shown in~\textbf{(a)}, and the detection circuit in~\textbf{(b)}.}
\end{figure}

\subsubsection*{Stabilizers}

With this encoding, a codeword physical state is a superposition
\( \alpha\ket{000} + \beta\ket{111}. \)
This defines a two-dimensional subspace $\mathcal{C} \le (\C^2)^{\otimes 3}$, which is the codespace. We compute the group $\Stab_{\Pauli_3}\mathcal{C}$. If $A$ is a stabilizer from this group, it is a tensor product of Pauli matrices $A_1 \otimes A_2 \otimes A_3$ such that
\[ A_1 \otimes A_2 \otimes A_3 \big(\alpha\ket{000} + \beta\ket{111}\big) = \alpha\ket{000} + \beta\ket{111}. \]
It must be $Z$-type. If it was $X \otimes X \otimes X$, it would swap the coefficients $\alpha$ and $\beta$. If it was some other tensor product containing an $X$, it would send the state to a completely different subspace of $(\C^2)^{\otimes 2}$, for example
\[X \otimes X \otimes \bm 1 \big(\alpha\ket{000} + \beta\ket{111}\big) = \alpha\ket{110} + \beta\ket{001}. \]
It cannot contain $Y$ for the same reason.

The basis state $\ket{000}$ is stabilized by any $Z$-type operator, because $Z\ket 0 =\ket 0$. However, $Z\ket 1 = -\ket 1$, which means we need $A$ to be a tensor product of two $Z$ operators, so that the minus signs cancel. Conclude that $\Stab_{\Pauli_n}\mathcal{C}$ consists of the following elements:
\[
  \begin{matrix}
    \bm 1 & \otimes & \bm 1 & \otimes & \bm 1, \\
    Z & \otimes & Z & \otimes & \bm 1, \\
    \bm 1 & \otimes & Z & \otimes & Z, \\
    Z & \otimes & \bm 1 & \otimes & Z.
  \end{matrix}
\]
The identity is always contained and gives us no information. Note that $Z \circ Z = \bm 1$, which means
\[ \bm 1 \otimes Z \otimes Z = (Z \otimes Z \otimes \bm 1) \circ (Z \otimes \bm 1 \otimes Z). \]
Thus we can choose a minimal generating set for the stabilizer group consisting of the following generators:
\[
  \begin{matrix}
    Z & \otimes & Z & \otimes & \bm 1 \\
    Z & \otimes & \bm 1 & \otimes & Z
  \end{matrix}
\]
However, note that it is not always necessary to choose a minimal set. If we choose a non-minimal generating set, we have to perform more measurements (see below) than necessary, but this may be desirable in some cases.

\subsubsection*{Syndrome measurement}

These stabilizer generators tell us how to measure the syndromes. Each stabilizer corresponds to a single syndrome measurement. We do not need a syndrome measurement corresponding to $\bm 1 \otimes Z \otimes Z$ --- as it is generated by the two stabilizers above, measuring it would give us no new information.

We show how to measure the syndrome in \Cref{fig:quantum-phase-flip-detection}. Each generator $A_1 \otimes A_2 \otimes A_3$ gets an ancillary qubit. If $A_i = Z$, then we entangle the ancilla with the~$i^{\text{th}}$ qubit of the computation state, using a controlled $Z$, which can detect $X$-type errors. For example, the first ancilla (labeled $\ket0_1$) corresponds to $Z \otimes Z \otimes \bm 1$, so we entangle it with the first two computation qubits. In the figure, we write everything in the standard ($Z$) basis, so we write a controlled $Z$ as a CNOT conjugated by the Hadamard operator.

\subsection{A wild parity check matrix appears}

Continuing with the example, recall the parity check matrix of the classical $[3,1,3]$ repetition code from \cref{eq:parity-check-row}:
\begin{align*}
  P
  & =
    \begin{pmatrix}
      1 & 1 & 0 \\
      1 & 0 & 1
    \end{pmatrix}
  & \longleftrightarrow &
  &  \begin{matrix}
       Z & \otimes & Z & \otimes & \bm 1 \\
       Z & \otimes & \bm 1 & \otimes & Z
     \end{matrix}
\end{align*}
We write it next to the generating set of $\Stab_{\Pauli_3}\mathcal{C}$. The positions of $0$ and $1$ in the matrix correspond to the positions of $\bm 1$ and $Z$, respectively, in the generators (in the order in which we wrote them). This is not a coincidence! It shows that the classical $[3,1,3]$ repetition code from \Cref{sec:3-1-3-code} does in fact correspond to the phase-flip code with this choice of generators.

Recall from \Cref{sec:interpretation-of-the-parity-matrix} that we interpret the rows of the parity matrix as specifying measurements of the physical bits. Stabilizer generators also tell us which qubits to measure. Then this correspondence between the parity matrix and generators of the stabilizer subgroup makes sense, and it means we can represent our quantum codes using vectors and matrices over $\Z_2$, in the same way as we did classical codes.

\begin{definition}[binary representation of stabilizers]
  \label{def:binary-rep-of-stabs}
  Suppose we have a system of $n$ physical qubits. We represent $X$- and $Z$-type Pauli operators from $\Pauli_n$ using vectors from $\Z_2^n$ as follows: Take $\vec v \in \Z_2^n$, then the corresponding $X$-type operator is
  \[ \mathcal{X}(\vec v) \deq \bigotimes_{i=1}^n X^{v_i} = X^{v_1} \otimes X^{v_2} \otimes \cdots \otimes X^{v_n}, \]
  and the $Z$-type operator is
  \[ \mathcal{Z}(\vec v) \deq \bigotimes_{i=1}^n Z^{v_i} = Z^{v_1} \otimes Z^{v_2} \otimes \cdots \otimes Z^{v_n}. \]
  Clearly these are bijections between the set of $X$-type (resp. $Z$-type) operators and~$\Z_2^n$.
\end{definition}

With the above, we can represent a quantum code on $n$ physical qubits that obtains syndromes by performing $m$ measurements by a parity check matrix $P : \Z_2^n \to \Z_2^m$.

\subsection{CSS Code from orthogonal classical codes}
\label{sec:css-code-from-ortho-classical-codes}

We can now study CSS stabilizer codes as a pair of classical codes joined together in a compatible manner. The splitting into two codes comes from the fact that the stabilizer group of a CSS code has generating set split into $X$- and $Z$-type operators. Thus each of those corresponds to one classical code. We represent them using parity check matrices $P_X : \Z_2^n \to \Z_2^{m_X}$ and $P_Z : \Z_2^n \to \Z_2^{m_Z}$. The compatibility between the two codes comes from the fact that the stabilizer group is required to be abelian. This leads to the following relationship:

\begin{theorem}
  \label{thm:stabilizers-commute-iff-vectors-orthogonal}
  Let $\vec x$ and $\vec z \in \Z_2^n$ be two vectors representing stabilizers $\mathcal{X}(\vec x)$ and $\mathcal{Z}(\vec z)$, respectively. The stabilizers commute if and only if $\langle \vec x, \vec z \rangle = 0$, where $\langle -,- \rangle$ is the usual inner product on $\Z_2^n$. In other words, the stabilizers commute exactly when their representing vectors are orthogonal.
\end{theorem}
This is a known fact (see \cite{cowtan2023css}) and we do not prove it here. However, we will generalize this to qudits in \Cref{thm:qudit-stab-commute-iff-rep-vectors-ortho}, which we prove.

It will be incredibly useful to express the theorem in terms of the parity check matrices $P_X$ and $P_Z$, and particularly their composite. Both have $\Z_2^n$ as domain, and they have, in general, different spaces as codomains, so they cannot be composed on the nose. We transpose one of them, say the $Z$-type matrix, to obtain $P_Z^\top : \Z_2^{m_Z} \to \Z_2^n$. Then \Cref{thm:stabilizers-commute-iff-vectors-orthogonal} tells us the following:
\begin{corollary}
  The composite $\Z_2^{m_Z} \xrightarrow{P_Z^\top} \Z_2^n \xrightarrow{P_X} \Z_2^{m_x}$ is zero, or equivalently $\im P_Z^\top \subseteq \ker P_X$.
\end{corollary}
\begin{proof}
  Write the parity matrices as
  \begin{align*}
    P_X
    & =
      \begin{pmatrix}
        \leftarrow & \vec x^1 & \rightarrow \\
                   & \vdots & \\
        \leftarrow & \vec x^{m_X} & \rightarrow
      \end{pmatrix},
    & \text{and} &
    & P_Z
    & =
      \begin{pmatrix}
        \leftarrow & \vec z^1 & \rightarrow \\
                   & \vdots & \\
        \leftarrow & \vec z^{m_Z} & \rightarrow
      \end{pmatrix}.
  \end{align*}
  The generators are represented by $\vec x^i$ and $\vec z^i$, and they are the rows of the matrices. In components, we write $(P_X)_{i,j} = x^i_j$ and $(P_Z)_{i,j} = z^i_j$. Then the composite is
  \[ (P_X \circ P_Z^\top)_{i,j} = \sum_{\alpha=1}^n (P_X)_{i,\alpha} (P_Z)_{j,\alpha} = x^i_\alpha z^j_\alpha = \langle \vec x^i, \vec z^j \rangle = 0, \]
  where we used the fact that the vectors $\vec x^i$ and $\vec z^j$ are orthogonal for all $i,j$. Conclude that $P_X \circ P_Z^\top = 0$.
\end{proof}

This is an important property of a CSS code. The sequence
\[ \Z_2^{m_Z} \xrightarrow{P_Z^\top} \Z_2^n \xrightarrow{P_X} \Z_2^{m_X} \]
such that $\im P_Z^\top \subseteq \ker P_X$ is called a \emph{chain complex}, a term which originated in \emph{homological algebra}. This is deeply connected to algebraic topology, and it will allows us to design quantum error correcting codes by analyzing certain kinds of topological spaces. We will also use these to generalize from CSS codes on qubit systems to codes on systems of qudits. We need some mathematical tools to properly deal with these, and we introduce them in the following chapters.

\part[Tools from Algebraic Topology and Homological Algebra]{Tools from Algebraic Topology and Homological Algebra \\[4cm] \includegraphics{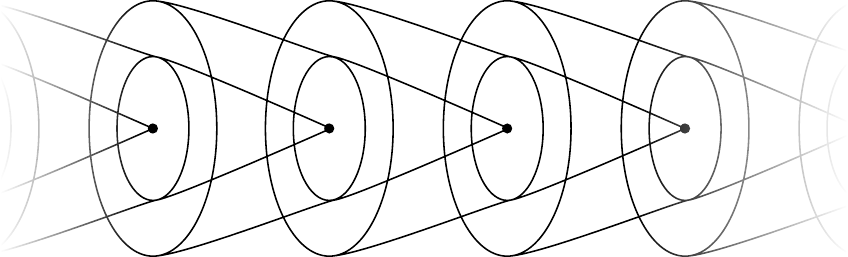}}
\label{part:alg-top-hom-alg}
\chapter[Rings and Modules]{Rings and Modules \hspace*{\fill} \begin{tikzpicture}
  \tikzstyle{point}=[fill=black,draw=black,circle,minimum size=3pt,inner sep=0]
  \tikzstyle{whitepoint}=[fill=white,draw=black,circle,minimum size=3pt,inner sep=0]
  \tikzstyle{->-}=[
  decoration={
    markings,
    mark=at position 0.7 with {\arrow{Triangle[scale=0.7]}}},
  postaction={decorate},
  draw]
  \node[point] at (90 : 0.3) (n0) {};
  \foreach \i in {1,...,5}{
    \node[whitepoint] at (90 - \i*360/6 : 0.3) (n\i) {};
  }
  \begin{pgfonlayer}{background}
    \draw[->-] (n0.center) -- (n1.center);
    \draw[->-] (n1.center) -- (n2.center);
    \draw[->-] (n2.center) -- (n3.center);
    \draw[->-] (n3.center) -- (n4.center);
    \draw[->-] (n4.center) -- (n5.center);
    \draw[->-] (n5.center) -- (n0.center);
  \end{pgfonlayer}
\end{tikzpicture}}
\label{chap:rings-and-modules}

Firstly, we need some basics.
In this chapter, we generalize the notions of fields and vector spaces over them. This will be used to reason about qudit systems.

We assume the reader is familiar with most of these ideas. As such, we only review the basic definitions (to fix notation), and parts that can subtle or are especially useful to us. We follow \cite{JacobsonNathan2009Ba,JacobsonNathan2009Bb}, and an unfamiliar reader can use those to find out more.

\section{Rings}

\begin{definition}[ring]
  A \emph{ring} is a structure $(R, +, \cdot, 0, 1)$, where $R$ is a set, called the \emph{carrier}, we have binary operations $+, \cdot : R \times R \to R$ called, respectively, \emph{addition} and \emph{multiplication}, and $0,1$ are distinguished elements of $R$. The triple $(R,+,0)$ is an abelian group, called the \emph{additive group} of $R$, and the triple $(R, \cdot, 1)$ is a monoid, called the \emph{multiplicative monoid}. We require that multiplication distributes over addition:
  \begin{align*}
    r \cdot (s + t) & = r \cdot s + r \cdot t
    & \text{and} &
    & (r + s) \cdot t & = r \cdot t + s \cdot t.
  \end{align*}
  By convention, every ring in this work is \emph{unital} (multiplication forms a monoid, as above, and not just a semigroup). Often, we write multiplication without the dot. Lastly, whenever this does not create ambiguity, we abuse notation and refer to the whole ring as $R$.
\end{definition}

\begin{definition}[ring homomorphism]
  A function between rings that preserves the ring structure is a \emph{ring homomorphism}.
\end{definition}

\begin{notation}
  As is usual, a homomorphism that is injective is also called a \emph{monomorphism} (denoted $R \monoto S$), surjective an \emph{epimorphism} ($R \epito S$), and bijective an \emph{isomorphism} ($R \isoto S$). If there is an isomorphism, the rings are \emph{isomorphic} ($R \cong S$). This terminology and notation will be used also for homomorphisms between other kinds of objects.
\end{notation}

Now a few relevant examples:
\begin{example}[integers]
  The set of integers $\Z$ forms a ring with the usual addition and multiplication. This ring is \emph{commutative}, which means its multiplicative monoid is commutative.
\end{example}

\begin{example}[endomorphism rings]
  \label{ex:endomorphism-ring}
  A ring may not necessarily contain numbers. For an abelian group $(A, \boxplus, \hat 0)$, denote $\End(A)$ the set of its \emph{endomorphisms} $f : A \to A$, that is homomorphisms from $A$ to itself. This set forms a ring, where we define the ring addition $+$ pointwise using the addition of the group, and ring multiplication as the composition of functions:
  \begin{align*}
    (f + g)(a) & \deq f(a) \boxplus g(a)
    & \text{and} &
    & (f \cdot g)(a) & \deq (f \circ g)(a) = f(g(a))
  \end{align*}
  for $f,g \in \End(A)$ and $a \in A$.
  The neutral element for multiplication is the identity function $\id_A : A \to A$ which maps $a \mapsto a$ for all $a \in A$, and zero, the neural element for addition, is $0 : A \to A$ that maps $a \mapsto \hat 0$. Note that the endomorphism ring is, in general, not commutative.
\end{example}

\begin{caution}
  \begin{convention}
    \label{con:all-rings-commutative}
    Most of our rings will be commutative. In fact, \textbf{from now on, unless otherwise stated, every ring will be commutative}.
  \end{convention}
\end{caution}

\subsection{Quotients}

Throughout the work, we will need quotient rings, especially $\Z/d\Z$. We define these below:
\begin{definition}[ideal]
  Let $R$ be a ring. A subset $I \subseteq R$ is an \emph{ideal} (denoted $I \idealeq R$) if it is a subgroup of the additive group of $R$ and it absorbs multiplication by all of $R$, which means that for any $r \in R$ and $i \in I$, the product $ri \in I$.
\end{definition}

\begin{definition}[quotient]
  Let $R$ be a ring and $I \idealeq R$ an ideal. Define an \emph{equivalence relation} $\sim\ \subseteq R \times R$ such that $r \sim s$ if and only if their difference $r - s$ is in $I$, for all $r,s \in R$. Then the equivalence classes, also called \emph{cosets}, are $[r] = r + I \deq \{ r + i : i \in I \}$ for $r \in R$, and we denote the set of equivalence classes $R/I$. The set $R/I$ is a ring with the following operations:
  \begin{itemize}
  \item addition: $[r] + [s] \deq [r + s] = r + s + I$, with the neutral element $[0] = I$, and
  \item multiplication: $[r][s] \deq [rs] = rs + I$, with the neutral element $[1] = 1 + I$,
  \end{itemize}
  for $r,s,0,1 \in R$.
  The ring $R/I$ is called the \emph{quotient ring} of $R$ by $I$, and we have a canonical epimorphism $R \epito R/I$ that maps $r \mapsto [r]$ for all \mbox{$r \in R$}.
\end{definition}

\begin{example}[modular arithmetic]
  \label{ex:modular-arithmetic}
  In our work, the most important rings are $\Z$ and its quotients by ideals. Take an integer $n \in \Z$; then we have the ideal $n\Z = \{ nm : m\in\Z \}$ containing all integer multiples of $n$. The quotient ring consists of \emph{residue classes} $[k] = k + n\Z = \{ n m + k : m \in \Z \}$ for $k \in \Z$.
  For $k, \ell \in \Z$, the sum is $[k] + [\ell] = [k + \ell \mod n]$, and product is $[k][\ell] = [k \ell \mod n]$, that is the quotient ring $\Z/n\Z$ is the setting for integer arithmetic modulo $n$. We denote it by the shorthand $\Z_n \deq \Z / n\Z$.
\end{example}

This has a few important special cases:
If $n = 0$, then $0\Z = \{ 0 \}$ and all elements $m \in \Z$ have their own singleton equivalence class $[m] = \{ m \}$. The quotient $\Z_0$ is isomorphic to $\Z$.

In the opposite extreme, if $n = 1$, then $1\Z = \Z$ which contains the difference of any $m,p \in \Z$. Thus we have a single equivalence class, and $\Z_1 \cong \bm 0$, the \emph{zero ring}.

If $n$ is prime, then $\Z_n$ contains all multiplicative inverses of non-zero elements, and is thus a \emph{finite field}. However, if $n$ is not prime, this no longer holds, which leads us to the following section.

\subsection{Division}
\label{sec:division}

The existence of multiplicative inverses of elements of a ring $R$, that is \emph{division}, is an important feature of the ring. Special interest belongs to the division of the zero element:
\begin{definition}[zero-divisor]
  Let $R$ be a ring, and $r \in R$. If there exists $s \in R$, $s \ne 0$ such that $rs = 0$, then we call $r$ a \emph{zero-divisor}. If $r$ is also non-zero, then it is a \emph{non-zero zero-divisor}.
\end{definition}

\begin{example}
  \label{ex:Z6-zero-divisors}
  Take $\Z_6 = \Z/6\Z$, and observe that $[2] \cdot [3] = [6 \mod 6] = [0]$, i.e. we can multiply non-zero elements to get zero. The elements $[2]$ and $[3]$ are both zero-divisors.
\end{example}
The division of zero is important enough that we define classes of rings based on whether this is possible:
\begin{definition}[PID]
  \label{def:PID}
  A non-zero ring $R$ that does not contain any non-zero zero-divisors, and where every ideal is generated by a single element, is called a \emph{principal integral domain}.
\end{definition}

Notably, the ring $\Z$ is a PID, and as such, many familiar properties hold here. However, as shown in \Cref{ex:Z6-zero-divisors}, $\Z_d$ with non-prime $d$ has zero-divisors, and hence is not a domain. We care about this distinction because much of the previous work on homological quantum error correction was focused on working with qudits of prime dimension $p$, which are related to the fields $\Z_p$, or on working with $\Z$, which is a PID. We want to extend this work to $\Z_d$, where $d > 1$ is arbitrary.

\section{Modules}
\label{sec:modules}

We now focus our attention to spaces that are linear with respect to a ring. This generalizes the notion of a vector space.

\begin{definition}[module]
  Let $R$ be a ring, and let $M$ be an abelian group. Define a ring action of $R$ on $M$ as the ring homomorphism $\mu : R \to \End M$, where $\End M$ is the endomorphism ring of $M$ (see \cref{ex:endomorphism-ring}). We write this action as left multiplication $r \cdot x = rx \deq \mu(r)(x)$. The structure $(M,\mu)$ is called an \emph{$R$-module}.\footnote{With multiplication on the left, it is a \emph{left $R$-module}. The other possibility would be a \emph{right $R$-module}. All our modules will be over commutative rings (\Cref{con:all-rings-commutative}), in which case the left vs. right distinction does not matter, and we omit it.}
\end{definition}
We often leave the action implicit, use multiplicative notation, and refer to $M$ as the module.

\begin{definition}[module homomorphism]
  Let $R$ be a ring, and let $(M,\mu)$\hiddenparbreak
  \diagramonright{0.8}{0.2}{
    \vspace{1ex}
    \noindent
    and $(N,\nu)$ be $R$-modules. A \emph{homomorphism of $R$-modules}, also called an \emph{$R$-linear map}, is a group homomorphism $f : M \to N$ that is also compatible with the action of $R$, meaning it makes the diagram on the right commute. This can be also written as\hiddenparbreak
    \vspace{1.1ex}
  }{
    \centering
    \hspace{1ex}
    \begin{tikzcd}[ampersand replacement=\&]
      M \& M \\
      N \& N
      \arrow["{\mu(r)}", from=1-1, to=1-2]
      \arrow["{\nu(r)}"', from=2-1, to=2-2]
      \arrow["f"', from=1-1, to=2-1]
      \arrow["f", from=1-2, to=2-2]
    \end{tikzcd}
  }
  $f(r x) = r f(x)$ for each $r \in R$ and $x \in M$.
\end{definition}

It is important here that both modules are of the same kind, i.e. both are $R$-modules. The collection of all $R$-modules, together with all $R$-linear maps, forms the category $\RMod R$.

We generalize the following objects from the theory of vector spaces to $R$-modules. They are exactly what one would think.
\begin{definition}[kernel, image]
  Let $M,N$ be $R$-modules, and $f:M \to N$ a~homomorphism of \mbox{$R$-modules}. As usual, denote its \emph{kernel} as $\ker f \deq \{ x \in M : f(x) = 0 \}$, and its image as $\im f \deq \{ f(x) : x \in M \}$.
\end{definition}
\begin{theorem}
  If $f$ is a homomorphism as above, then $\ker f$ is a submodule of $M$, and $\im f$ is a submodule of $N$
\end{theorem}
\noindent
This is a standard fact, so we omit the proof. See \cite{HatcherAllen2002At} for details.

\begin{definition}[submodule]
  Let $M$ be an $R$-module. A submodule $N$ of $M$, written as $N \le M$, is a subgroup of the additive group of $M$ that is closed under the action of $R$, that is $R \cdot N = \{ rn : r \in R, n \in N \} \subseteq N$.
\end{definition}

\subsection{Examples of Modules}

Here, we present a selection of examples of modules. In particular, the first two will be important later.

\begin{example}[ring]
  \label{ex:ring-as-module-over-itself}
  Any ring $R$ can be seen as a  module over itself, where the action is just the ring multiplication: $\mu(r)(s) = rs$ for $r,s \in R$.
  Observe that if $I$ is a submodule of $R$ as module over itself, then $I$ is an ideal of the ring $R$.
\end{example}

\begin{example}[abelian group]
  Let $(A, +, 0)$ be an abelian group. Define the multiplication by $n \in \Z$ as
  \[ n \cdot a \deq \underbrace{a + a + \cdots + a}_{n\ \text{times}}\]
  for all $a \in A$. This makes $A$ a $\Z$-module. Every abelian group is a $\Z$-module this way.
  Conversely, any $\Z$-module is already an abelian group, we just forget the multiplication. This means we can reason about abelian groups and $\Z$-modules interchangeably.
\end{example}

As previously mentioned, modules are meant to generalize vector spaces, so it is nice to see that a vector space is, indeed, a module:
\begin{example}[vector space]
  \label{ex:K-module-is-vector-space}
  Recall that a field $\mathbb{K}$ is a PID where division by all non-zero elements is defined. A $\mathbb{K}$-module is, by definition, a \emph{vector space}. The action in this case is exactly \emph{scalar multiplication}. A submodule is now a vector subspace.
\end{example}

\subsection{Free and Finitely Generated Modules}
\label{sec:free-and-finitely-generated-modules}

Modules over a ring generalize vector spaces (see \Cref{ex:K-module-is-vector-space}), in that they are linear spaces, except over a ring, not necessarily a field. Many familiar properties of vector spaces no longer hold in this setting, depending on how far the ring is from a field. In this section, we define a class of modules that are like vector spaces, in that they are spanned by a linearly independent basis.

\begin{definition}[generating set]
  \label{def:generating-set}
  Let $M$ be an $R$-module, and $S \subseteq M$ be a subset of $M$. Denote $\genR R S$ the \emph{submodule generated by $S$}, defined by
  \[
    \genR RS \deq \left\{
      \sum_{x_i \in S} r_i x_i : r_i \in R
    \right\}
    = \sum_{x_i \in S} R x_i,
  \]
  i.e. $\genR RS$ contains $R$-linear combinations of elements of $S$. If $\genR RS = M$, we say $M$ is \emph{generated} by $S$, and we call $S$ the \emph{generating set}.
  If $S$ is a smallest such set, then we call its cardinality the \emph{rank} of $M$, denoted $\rk M$.

  If there exists a finite $S$ such that $M = \genR RS$ (so $\rk M$ is finite) then we say that $M$ is \emph{finitely generated}.
\end{definition}

Note that so far, there is no mention of linear independence. In particular, $M$ generates itself: $M = \genR R M$.

\begin{definition}[free module]
  \label{def:free-module-base}
  An $R$-module $M$ is \emph{free} if there exists a set $S = \{ e_i \}_i$ of elements of $M$, such that $M = \genR R S$ and all its elements are linearly independent: for $r_i \in R$, if
  \[ \sum_{i} r_i e_i = 0, \]
  then each $r_i = 0$. The set $S$ is called a \emph{basis} of $M$.\footnote{Module theory literature prefers the term \emph{base} for modules, and \emph{basis} only for vector spaces.\cite{JacobsonNathan2009Ba} There are the same concept, and we call them the same name for clarity.}
\end{definition}

A free module is the most vector space-like. The rank of a free module is analogous to the dimension of a vector space. In fact if $\mathbb{K}$ is a field, then a $\mathbb{K}$-module is a vector space and the rank and dimension are the same thing.

There is a useful equivalent characterization using direct sums. This is a generalization of the fact that a vector space $V$ over a field $\mathbb{K}$ of dimension $n$ is a direct sum of $n$ copies of $\mathbb{K}$.

\begin{definition}[finitary direct sum of modules]
  Let $R$ be a ring, and for $n \in \N$, let $M_1, \dots, M_n$ be $R$-modules. Recall that the Cartesian product of these is
  \[
    \prod_{i=1}^n M_i \deq \{ (x_1, \dots, x_n) : x_j \in M_j \; \forall 1 \le j \le n \},
  \]
  i.e. all ordered tuples, where the $j^{\mathrm{th}}$ entry is from module $M_j$. Define the addition and $R$-action on this component-wise:
  \begin{align*}
    (x_1, \dots, x_n) + (y_1, \dots, y_n) & \deq (x_1 + y_1, \dots, x_n + y_n),
    & \text{and} &
    & r(x_1, \dots, x_n) & \deq (rx_1, \dots, rx_n),
  \end{align*}
  where $(x_1, \dots, x_n)$ and $(y_1, \dots, y_n) \in \prod_{i=1}^nM_i$, and $r \in R$. The zero element is $(0, \dots, 0)$, an $n$-tuple of all zeros. The module defined this way is called the \emph{direct sum} of $M_1, \dots, M_n$, denoted $\bigoplus_{i=1}^n M_i$.
\end{definition}
Note that the above definition explicitly requires a finite number of modules. This can be generalized to a direct sum any family of modules: then the elements of the direct sum are further required to have only finitely many non-zero components. In our work, all direct sums will be finitary, so we do not worry about this further.

\begin{theorem}[free module]
  An $R$-module is free if and only if it is isomorphic to $\bigoplus_{i=1}^nR$, a direct sum of $n$ copies of $R$, denoted $R^{\oplus n}$, for some $n \in \N$. Here, we see $R$ as a module over itself (see \Cref{ex:ring-as-module-over-itself}).
\end{theorem}
This is a standard fact, and a proof can be found in \cite[\S3.4]{JacobsonNathan2009Ba}.

\begin{notation}
  When we want to specify a generating set that is also a basis, i.e. we have $S = \{ e_1, \dots, e_n \}$ and $M = \bigoplus_{i} R e_i$, we write it as $\genR R {e_1, \dots, e_n}^\oplus$, where we mark it with the $\oplus$ symbol to indicate that the generating elements are linearly independent.
\end{notation}

\subsection{Matrices}

Throughout the work, we will need to perform explicit computation using modules and homomorphisms between them. In the context of finite-dimensional vector spaces, we can encode their homomorphisms using matrices. This works exactly because they have bases, and it generalizes to free modules:

\begin{definition}
  Let $M,N$ be free and finitely-generated $R$-modules. Suppose $M$ has rank $m$ and basis $\{ m_j \}_{j=1}^m$, and suppose $N$ has rank $n$ and basis $\{ n_i \}_{i=1}^n$. If $f : M \to N$ is an $R$-linear map, we can write $f$ using its action on the basis element $m_j$ as follows:
  \[ f(m_j) \deq \sum_{i=1}^n f_{i,j} n_i, \]
  where $f_{i,j} \in R$, and we linearly extend this to all of $M$.
  Having a choice of bases, we can identify $f$ with the \emph{matrix} $(f_{i,j})_{i,j} \in R^{n \times m}$ that represents it.
\end{definition}

In particular, we will often need to compute kernels and images of homomorphisms of $\Z$-modules in order to compute homologies of chain complexes; see \Cref{sec:chain-complex-abstract}. This can be done by eyeballing the morphism or its representing matrix, but there is also a systematic way to do it.

The algorithm we use is \emph{Gaussian elimination} which transforms a matrix into a \emph{row echelon form} by a sequence of \emph{elementary row operations}. We use this technique to obtain kernels and images of matrices. We assume this to be known; an unfamiliar reader may learn more in ref. \cite{MR1391966}.

Note that we do not necessarily get the \emph{reduced} row echelon form, which is a unique normal form. The reduced part means that all leading coefficients are equal to $1$, and this can in general only be obtained if the ring of scalars contains divisors of all non-zero elements. This is not the case for integers. For example, if a matrix is in a row echelon form and some leading coefficient is $2$, this cannot be reduced, because that would require scalar multiplication by $\frac12$, which is not contained in $\Z$. As we will see later, this is exactly where \emph{torsion} comes in (see \Cref{sec:torsion}).

We have used Sage,\cite{sagemath} a Computer Algebra System, to do many of these computations for us. An example of such computation is in \Cref{cha:jupyter-sage}.\footnote{While looking for a method to wrangle the huge matrices that sometimes came out of cell complexes (see \Cref{sec:chain-complexes-from-cell-complexes}), we have also explored other algorithms. In particular, there is a method to compute kernels and images of matrices over \emph{principal ideal rings}, which include not just $\Z$, but also $\Z_d$; see ref. \cite{Buchmann1996AlgorithmsFL}. However, due to \Cref{thm:structure-theorem-for-qudit-logical-space}, we did not need this in the end.}

\subsection{Quotient}

Similarly to rings, we now construct a \emph{quotient module}. The construction is, basically, a group quotient with induced action.

\begin{definition}[quotient]
  Let $(M, \mu)$ be an $R$-module, and $N \le M$ a submodule. Recall that these are abelian groups, and as such, the subgroup $N$ is normal in $M$.\footnote{A normal subgroup $N \idealeq G$ is one closed under conjugation, i.e. $g N g^{-1} \subseteq N$ for each $g \in G$. If $G$ is abelian, this is true for all subgroups.} Then we have the quotient group $M/N$. Let $x, y \in M$. The elements of $M/N$ are cosets $[x] = x + N$, equivalence classes of the relation defined by $x \sim y$ if $x - y \in N$. The addition is defined as $[x] + [y] = [x + y]$, the zero element is $[0] = N$, and the additive inverse is $-[x] = [-x]$.

  \diagramonright{0.76}{0.24}{
    \vspace{1.2ex}
    We define the action $\overline{\mu} : R \to \End (M/N)$ by requiring that the diagram on the right commutes for every \mbox{$r \in R$}.
    Here, $q : M \to M/N$ is the quotient map.
    Equivalently, this means that $r[x] \deq [rx]$ for every $r \in R$ and $x \in M$. This induced \hiddenparbreak
    \vspace{1.2ex}
  }{
    \centering
    \begin{tikzcd}[ampersand replacement=\&]
      M \& M \\
      {\faktor{M}{N}} \& {\faktor{M}{N}}
      \arrow["{\mu(r)}", from=1-1, to=1-2]
      \arrow["{\overline{\mu}(r)}"', dashed, from=2-1, to=2-2]
      \arrow["q"', two heads, from=1-1, to=2-1]
      \arrow["q", two heads, from=1-2, to=2-2]
    \end{tikzcd}
  }
   action clearly also satisfies everything it needs to, i.e. it is compatible with addition and multiplication of the ring.
  The quotient group $M/N$, together with the induced action $\overline{\mu}$, forms a \emph{quotient module}.
\end{definition}

\subsection{Torsion}
\label{sec:torsion}

With regard to zero-division from \cref{sec:division}, we now define the conceptual opposite to freeness in terms of modules. We use the definition from ref. \cite{LamT.Y2006Eima}, because this is more general than other definitions in the literature which require the ring to be a PID.

\begin{definition}[torsion]
  Let $M$ be an $R$-module. An element $x \in M$ is called a \emph{torsion element} if there exists an $r \in R$ that is neither zero, nor a zero-divisor, such that $rx = 0$. The set of torsion elements is denoted $\Tor_RM$, and it is a submodule of~$M$.
\end{definition}

In the context of modules over a PID, we have the following fundamental result that allows us to split the module into a free part and a torsion part.
\begin{theorem}[Structure Theorem for Finitely Generated Modules over PID]
  \label{thm:structure-thm-PID-mod}
  Let $R$ be a PID, and let $M$ be a nonzero finitely-generated $R$-module. Then there exist nonzero and non-invertible elements $d_1,\dots,d_s \in R$ ($s \in \N$), such that $d_1$ divides $d_2$, that divides $d_3$, etc., and there exists an integer $k \in \N$, such that
  \begin{equation}
    \label{eq:structure-thm-PID-mod}
    M \cong R^{\oplus k} \oplus \bigoplus_{i=1}^s \faktor R {d_i R}.
  \end{equation}
  The generators $d_i$ are unique up to multiplication by an invertible element, and are called \emph{invariant factors} of $M$.
\end{theorem}
We do not prove this theorem, because it is standard and can be looked up for example in ref. \cite{JacobsonNathan2009Ba}. It can be summarized also as
\begin{equation}
  \label{eq:structure-thm-PID-mod-tor}
  M \cong R^{\oplus k} \oplus \Tor_RM,
\end{equation}
because the summands $R/d_iR$ are torsion $R$-modules. This decomposition will be important to us later. It will tell us what the logical objects are in an error correcting code if we include torsion.\cite{vuillot2023homological} More on this in \Cref{chap:torsion}.

\section{Tensor Product of Modules}
\label{sec:tens-prod-modul}

In this section, we introduce the concept of \emph{tensor product}. This is a way to combine two objects into one in a way that carries a lot of coherence. In particular, this is different than, for example, a direct sum. Here, we define the tensor product for modules, but we will meet related tensor products for other kinds of objects as we go along, and the relationships between these will be important to us.\footnote{To foreshadow slightly, this will lead us to products of certain topological complexes, their associated chain complexes of modules, and it will be interesting to see how homology behaves with respect to this. See \Cref{sec:tensor-product-codes}} This is why we dedicate it its own section.
We follow the development in \cite[\S 3.7]{JacobsonNathan2009Bb}, but we simplify the definition to suit our purposes.\footnote{The tensor product is conventionally defined as the \emph{universal balanced product}. We absorb the definition of  balanced product into the definition of tensor product.}

\begin{definition}[tensor product]
  Let $R$ be a ring, and let $M, N$ be $R$-modules. Define the \emph{tensor product} of $M$ and $N$ over $R$ as the abelian group $M \otimes_R N$ together with a biadditive map $-\otimes- : M \times N \to M \otimes_R N$. The biadditivity means that $\otimes$ acts on each component as a homomorphism of the additive group:
  \begin{align}
    \label{eq:balanced-product-biadditive}
    (x + x') \otimes y & = (x \otimes y) + (x' \otimes y)
    & \text{and} &
    & x \otimes (y + y') & = (x \otimes y) + (x \otimes y')
  \end{align}
  for $x,x' \in M$ and $y, y' \in N$. Furthermore, we require that $\otimes$ is compatible with the action of $R$:
  \begin{equation}
    \label{eq:balanced-product-ring-action}
    (rx) \otimes y = x \otimes (ry)
  \end{equation}
  for $x \in M$, $y \in N$ and $r \in R$. Note that these conditions imply the following:
  \begin{align*}
    0 \otimes y & = 0 = x \otimes 0
    & \text{and} &
    & (-x) \otimes y & = -(x \otimes y) = x \otimes (-y).
  \end{align*}

  \diagramonright{0.7}{0.3}{
    \vspace{0.5ex}
    We require that the tensor product is \emph{universal}: this means that if there is another object $P$ together with a biadditive map $f$ that satisfies all of the above, then there exists a unique homomorphism of abelian groups\hiddenparbreak
    \vspace{1.2ex}
  }{
    \centering
    \begin{tikzcd}[ampersand replacement=\&]
      {M \times N} \& {M \otimes_R N} \\
      \& P
      \arrow["\otimes", from=1-1, to=1-2]
      \arrow["f"', from=1-1, to=2-2]
      \arrow["{\exists !}", dashed, from=1-2, to=2-2]
    \end{tikzcd}
  }
  \noindent
  $\eta : M \otimes_R N \to P$ that maps $x \otimes y \mapsto f(x,y)$ for $x \in M$ and $y \in N$, such that the diagram above right commutes.
\end{definition}

We explicitly construct $M \otimes_R N$ and describe its elements. Let \mbox{$F \deq \genR \Z {M \times N}^\oplus$} be the free abelian group with basis $M \times N$ (see \cref{sec:free-and-finitely-generated-modules}). Let $G \le F$ be the subgroup generated by
\begin{align*}
  (x + x', y) & - (x, y) - (x', y),
  & (x, y+y') & - (x,y) - (x,y'),
  & \text{and} &
  & (rx, y) & - (x, ry)
\end{align*}
for all $x, x' \in M$, $y, y' \in N$, and $r \in R$. These generators correspond to equations \eqref{eq:balanced-product-biadditive} and  \eqref{eq:balanced-product-ring-action}. Then we define $M \otimes_R N \deq F/G$, and the map $\otimes : M \times N \to M \otimes_R N$ as $x \otimes y \deq (x,y) + G$ for $x \in M$ and $y \in N$. This construction is indeed a tensor product; the proof is standard and we omit it here, but a reader may look it up in \cite[\S3.7]{JacobsonNathan2009Bb}.

Observe that the universal property makes the tensor product unique up to a unique isomorphism, and hence any tensor product is isomorphic to the above. From now on, we call $M \otimes_R N$ \emph{the} tensor product of modules. We now construct the corresponding tensor product for morphisms:

\begin{definition}[tensor product of $R$-linear maps]
  Let $R$ be a ring, $M,M', N,N'$\hiddenparbreak
  \diagramonright{0.65}{0.35}{
    \vspace{1ex}
    \noindent
    be $R$-modules, and let $f: M \to M'$ and $N \to N'$ be $R$-linear maps. We define the \emph{tensor product} of $f$ and $g$ as the unique map $f \otimes g$ that makes the diagram on the right commute. It maps $x \otimes y \mapsto f(x) \otimes g(y)$ for $x \in M$\hiddenparbreak
    \vspace{1ex}
  }{
    \centering
    \begin{tikzcd}[ampersand replacement=\&]
      {M \times N} \& {M \otimes_R N} \\
      {M' \times N'} \& {M' \otimes_R N'}
      \arrow["\otimes", from=1-1, to=1-2]
      \arrow["{f \otimes g}", dashed, from=1-2, to=2-2]
      \arrow["\otimes"', from=2-1, to=2-2]
      \arrow["{f \times g}"', from=1-1, to=2-1]
    \end{tikzcd}
  }
  and $y \in N$.
\end{definition}

We do not go into much explicit detail, but the map $f \otimes g$ is indeed unique, and it satisfies all the nice properties we might want. We summarize it compactly by noting that the tensor product is a bifunctor \[\otimes : \RMod R \times \RMod R \to \mathbf{Ab} = \RMod \Z. \]
Notice that its results are, so far, abelian groups and their homomorphisms. However, we want a module over the original ring $R$, and we obtain this as follows:
\begin{definition}[tensor product module]
  \label{def:tensor-product-module}
  Let $R$ be a ring, and let $M,N$ be $R$-modules. We make $M \otimes_R N$ an $R$-module by defining the action $R \to \End(M \otimes_R N)$ such that $r(x \otimes y) \deq (rx) \otimes y$.
\end{definition}

This is clearly compatible with the other structure. In particular, recall $(rx) \otimes y = x \otimes (ry)$, and so the action as defined above can be equivalently written $r(x \otimes y) = x \otimes (ry)$, as we would expect.

\subsection{Properties of Tensor Product Modules}

Here, we list some basic properties of the tensor product $R$-module. These are standard, and we do not prove them. We only list the ones that we need at some point. In the following, let $M,N,P$ be $R$-modules.

The tensor product behaves like a \emph{monoid}: it is \emph{associative}, and has a \emph{unit} which is the ring $R$ as module over itself:
\begin{align}
  \label{eq:tensor-monoid}
  M \otimes_R (N \otimes_R P) & \cong (M \otimes_R N) \otimes_R P,
  & \text{and} &
  & M \otimes_R R & \cong M \cong R \otimes_R M.
\end{align}

In $\RMod R$, we have a zero object, the module $\bm 0 = \{ 0 \}$. The tensor product of this with any other module is again zero:
\begin{equation}
  \label{eq:tensor-zero}
  \bm 0 \otimes_R M = \bm 0.
\end{equation}

The tensor product is compatible with the direct sum, in a way that mimics the distributivity of multiplication over addition of numbers. If $\{N_i\}_i$ is a family of $R$-modules, then
\begin{equation}
  \label{eq:tensor-distributes}
  M \otimes_R \bigoplus_i N_i \cong \bigoplus_i M \otimes_R N_i.
\end{equation}

In the context of integer rings as modules over themselves, we have the following equality. Let $a,b \in \N$, then
\begin{equation}
  \label{eq:tensor-Za-Zb-gcd}
  \Z_a \otimes \Z_b \cong \Z_{\gcd(a,b)},
\end{equation}
where $\gcd(a,b)$ is the greatest common divisor of $a$ and $b$.

\subsection{Extension of Scalars}

In \Cref{def:tensor-product-module}, we make the tensor product into a functor $\RMod R \times \RMod R \to \RMod R$. However, the $R$-module structure was added on top of the tensor product by defining a new action of $R$. A similar technique, but using a different ring, can be used to change the ring of scalars of a module. This will be an important step when describing quantum error correcting codes acting on qudit systems in \Cref{chap:torsion}. More details can be found in \cite{MR2286236}.

\begin{definition}[extension of scalars]
  \label{def:extension-of-scalars}
  Let $R, S$ be rings, let $f : R \to S$ be a ring homomorphism, and let $M$ be an $R$-module. Define an action $\sigma : R \to \End S$, such that $\sigma(r)(s) \deq f(r) \cdot_S s$, for $r \in R$ and $s \in S$, with the multiplication done in $S$. This makes $(S, \sigma)$ into an $R$-module, and hence there is a tensor product $S \otimes_R M$.

  The ring $S$ is also canonically an $S$-module, with action being the ring multiplication (see \Cref{ex:ring-as-module-over-itself}). We use this to define the action $\mu_S : S\to \End (S \otimes_R M)$ where $\mu_S(s)(s' \otimes x) \deq (s s') \otimes x$ for $s, s' \in S$ and $x \in M$. This makes $S \otimes_R M$ an $S$-module.

  The procedure we described is a way to turn an $R$-module into an $S$-module, and it is called the \emph{extension of scalars of $M$ by $S$ along $f$}.
\end{definition}

\begin{note}
  \label{note:scalar-extension-is-bimodule}
  Note that the scalar-extended module $S \otimes_R M$ can still be viewed as an $R$ module. Thus $S \otimes_R M$ has actions of two possibly different rings $R$ and $S$, these are compatible, and as such $S \otimes_R M$ is a \emph{bimodule}.
\end{note}

Since the tensor product is a functor, it follows that the extension of scalars is also a functor. Concretely, a ring homomorphism $f: R \to S$ induces the functor \mbox{$f_! : \RMod R \to \RMod S$}. The action on morphisms is as follows: If $\alpha : M \to N$ is  an $R$-linear map, then $f_!$ maps it to an $S$-linear map $f_!(\alpha) \deq \id_S \otimes \alpha$.

\chapter[Homological Algebra]{Homological Algebra \hspace*{\fill} \begin{tikzpicture}[scale=0.3,baseline=-0.35cm]
  \tikzstyle{vertex}=[fill=black,draw=black,circle,minimum size=2pt,inner sep=0]
  \tikzstyle{kernel}=[draw=none, fill=red, fill opacity=0.5]
  \tikzstyle{image}=[draw=none, fill=white, fill opacity=1]

  \begin{scope}[local bounding box=bbox]
    \node         (E1)  at (-2.3, 0) {};
    \node         (E0)  at (   0, 0) {};
    \node[vertex] (E-1) at ( 2.3, 0) {};

    \draw (E1) ellipse [x radius=0.9, y radius=1.3];
    \node (E-1-0-top) at ($(E1)+(0,1.3)$) {};
    \node (E-1-0-bot) at ($(E1)-(0,1.3)$) {};

    \foreach \i [evaluate={\r=1-0.3*\i; \a=0.9*\r; \b=1.3*\r}] in {0,...,2}{
      \draw (E0) ellipse [x radius=\a, y radius=\b];
      \node (E-0-\i-top) at ($(E0)+(0,\b)$) {};
      \node (E-0-\i-bot) at ($(E0)-(0,\b)$) {};
    }
    \node (E-0-3-top) at (E0) {};
    \node (E-0-3-bot) at (E0) {};

    \foreach \i [evaluate={\r=1-0.3*\i; \a=0.9*\r; \b=1.3*\r}] in {0,2}{
      \draw (E-1) ellipse [x radius=\a, y radius=\b];
      \node (E--1-\i-top) at ($(E-1)+(0,\b)$) {};
      \node (E--1-\i-bot) at ($(E-1)-(0,\b)$) {};
    }
    \node (E--1-3-top) at (E-1) {};
    \node (E--1-3-bot) at (E-1) {};

    \draw[out=-5, in=180-5] (E-1-0-top.center) to (E-0-2-top.center);
    \draw[out=5, in=180+5] (E-1-0-bot.center) to (E-0-2-bot.center);
    \foreach \n [evaluate={\m=int(\n-1)}] in {0}{
      \foreach \i [evaluate={\j=int(\i+2)}] in {0, 1}{
        \draw[out=-5, in=180-5] (E-\n-\i-top.center) to (E-\m-\j-top.center);
        \draw[out=5, in=180+5] (E-\n-\i-bot.center) to (E-\m-\j-bot.center);
      }
    }
    \begin{pgfonlayer}{background}
      \fill[kernel] (E0) ellipse [x radius=0.9*0.7, y radius=1.3*0.7];
      \fill[image] (E0) ellipse [x radius=0.9*0.4, y radius=1.3*0.4];
    \end{pgfonlayer}
  \end{scope}
  \node (bboxnw) at ($(bbox.north west)-(0,1)$) {};
  \node (bboxse) at ($(bbox.south east)+(0,1)$) {};
  \useasboundingbox (bboxnw.center) rectangle (bboxse.center);
\end{tikzpicture}}
\label{chap:homol-algebra}

In this chapter, we introduce the basic notions of Homological Algebra. This is an abstract area that originated in Algebraic Topology. However, we first present it abstractly, only using linear algebra and the notions from \Cref{chap:rings-and-modules}, without any topology involved. We follow the books \cite{HatcherAllen2002At} and \cite{weibel_1994}.

\section{Chain Complexes and Homology}
\label{sec:chain-complex-abstract}

\begin{definition}[chain complex]
  \label{def:chain-complex}
  Let $R$ be a ring, and let $\{ C_n \}_{n\in\N}$ be a family of $R$-modules, called \emph{chain modules} or \emph{components}. For $C_n$, the label $n$ is called the \emph{degree}. For convenience, we define also $C_{-1} \deq \bm 0$, the zero module. Let $\{ \partial_n : C_n \to C_{n-1} \}_{n \in \N}$ be a family of $R$-module homomorphisms, where $\partial_n$ is called the \emph{$n^{\text{th}}$ differential} or \emph{boundary map}. For each $n \in \N$, we require that $\partial_n \circ \partial_{n+1} = 0$ (written $\partial^2 = 0$ for short), or equivalently that $\im \partial_{n+1} \subseteq \ker \partial_{n}$. Such structure is called a \emph{chain complex of $R$-modules}, and is denoted $C_\bullet$. When necessary, the differentials may be annotated $\partial_n^{C_\bullet}$ to indicate their corresponding chain complex.
\end{definition}

There is a visual intuition for a chain complex. We display a special case where $\im \partial_{n+1} = \ker \partial_n$ (discussed later) at the beginning of \Cref*{part:alg-top-hom-alg} on \cpageref{part:alg-top-hom-alg}, and we show a general picture in \cref{fig:visualization-of-chain-complex}.
\begin{figure}[h]
  \centering
  \begin{tikzpicture}
  \tikzstyle{vertex}=[fill=black,draw=black,circle,minimum size=4pt,inner sep=0]
  \tikzstyle{kernel}=[draw=none, pattern=north west lines, pattern color=red]
  \tikzstyle{image}=[draw=none, pattern=north east lines, pattern color=blue]
  \foreach \n [evaluate={\x=-2.3*\n}] in {2,...,-3}{
    \node[vertex] (E\n) at (\x, 2) {};
    \foreach \i [evaluate={\r=1-0.3*\i; \a=0.9*\r; \b=1.5*\r}] in {0,...,2}{
      \draw[line width=\r] (E\n) ellipse [x radius=\a, y radius=\b];
      \node (E-\n-\i-top) at ($(E\n)+(0,\b)$) {};
      \node (E-\n-\i-bot) at ($(E\n)-(0,\b)$) {};
    }
    \node (E-\n-3-top) at (E\n) {};
    \node (E-\n-3-bot) at (E\n) {};
    \node at (\x, 0) (C\n) {$C_{\ifnum \n=0 n \else \ifnum \n<0 n\n \else n+\n \fi \fi}$};
    \begin{pgfonlayer}{background}
      \fill[kernel] (E\n) ellipse [x radius=0.9*0.7, y radius=1.5*0.7];
      \fill[image] (E\n) ellipse [x radius=0.9*0.4, y radius=1.5*0.4];
    \end{pgfonlayer}
  }
  \foreach \n [evaluate={\m=int(\n-1)}] in {2,...,-2}{
    \draw[->] (C\n) -- node[above] {\small $\partial_{\ifnum \n=0 n \else \ifnum \n<0 n\n \else n+\n \fi \fi}$} (C\m);
    \foreach \i [evaluate={\j=int(\i+2)}] in {0, 1}{
      \draw[out=-5, in=180-5] (E-\n-\i-top.center) to (E-\m-\j-top.center);
      \draw[out=5, in=180+5] (E-\n-\i-bot.center) to (E-\m-\j-bot.center);
    }
  }
\end{tikzpicture}
  \caption[Visualization of a chain complex]{Visualization of a chain complex $C_\bullet$. Chain modules $C_n$ are the big ovals, with the zero element distinguished by a small black circle $\bullet$. The differentials $\partial_n$ go from left to right, and are shown as funnels that take the domain (left oval) to their image, which is a subset of the codomain (right oval). We indicate the kernels of differentials by the red hatched area, and the images by the blue hatched area. Observe the chain complex condition that $\im \partial_{m+1} \subseteq \ker \partial_{m}$ for all $m$.}
  \label{fig:visualization-of-chain-complex}
\end{figure}
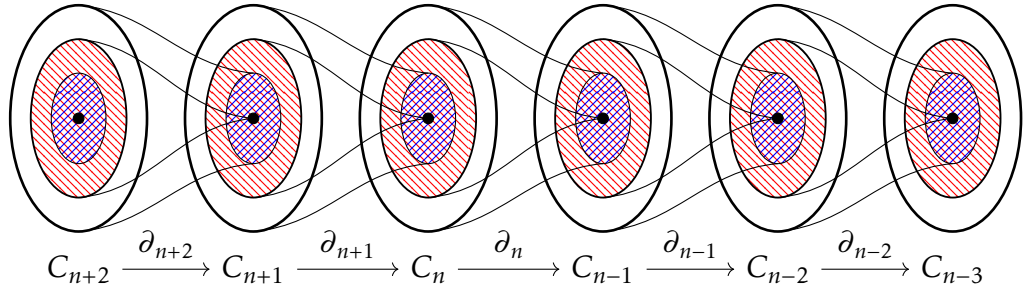

\begin{definition}[length of a chain complex]
  If there exists $n \in \N$, such that for all $m > n$, the chain module $C_m = \bm 0$, then the chain complex is \emph{bounded}. If $n$ is minimal, we call it the \emph{length} of the chain complex.
\end{definition}

There is a corresponding notion of morphisms between chain complexes, called the \emph{chain maps}. We will not need them, so we only briefly mention that these are families of $R$-linear maps, one for each degree, that commute appropriately with the differentials of each complex. Together, these form a category denoted $\Ch{\RMod R}$.

We look more closely at the requirement that $\im \partial_{n+1} \subseteq \ker \partial_n$ for each $n \in \N$. Recall that the images and kernels of $R$-linear maps are modules. This implies a stronger condition, that $\im \partial_{n+1}$ is in fact a submodule $\ker \partial_n$, and this allows us to define the quotient to study their difference.
\begin{definition}[homology]
  \label{def:homology}
  Let $C_\bullet$ be a chain complex of $R$-modules. For $n \in \N$, define the \emph{$n^{\mathrm{th}}$ homology module} as the $R$-module quotient
  \[ H_n(C_\bullet) \deq \faktor{\ker \partial_n}{\im \partial_{n+1}}. \]
  In the literature (e.g. \cite{HatcherAllen2002At}), they are often called the homology \emph{groups}, because often we work over $R = \Z$, however, we want to emphasize that the ring may be arbitrary.
\end{definition}

\begin{notation}[cycles and boundaries]
  \label{nota:cycles-and-boundaries}
  There is some terminology that comes with the above definition. This stems from Algebraic Topology, and we explain what it means in \Cref{sec:homology-of-cell-complex}. For now, these are just names. We call $\ker \partial_n \eqd Z_n(C_\bullet)$ the space of \emph{$n$-cycles}, and $\im \partial_{n+1} \eqd B_n(C_\bullet)$ the space of \emph{$n$-boundaries}. Then $H_1(C_\bullet) = Z_n(C_\bullet) / B_n(C_\bullet)$.
\end{notation}

In the following examples, we very briefly show why we care about this. We will work these out in more detail later.
\begin{example}[classical code]
  \label{ex:classical-code-is-chain-complex}
  Any $R$-linear map, together with its domain and codomain, can be interpreted as a chain complex of length $1$. In particular, in \Cref{sec:error-detection}, we introduce the parity check matrix $P : \Z_2^n \to \Z_2^m$ that defines an error correcting code. This corresponds to the following chain complex of $\Z_2$-modules (i.e. vector spaces over $\Z_2$):
  \[
    \begin{tikzcd}[ampersand replacement=\&, row sep=-4pt]
      \cdots \& {\bm 0} \& {\Z_2^n} \& {\Z_2^m} \& {\bm 0} \\
      \& {C_2} \& {C_1} \& {C_0} \& {C_{-1}}
      \arrow["0", from=1-1, to=1-2]
      \arrow["0", from=1-2, to=1-3]
      \arrow["P", from=1-3, to=1-4]
      \arrow["0", from=1-4, to=1-5]
    \end{tikzcd}
  \]
  Recall that the codespace is $\ker P$, which in this case is isomorphic to $H_1(C_\bullet)$, because $\im \partial_2= \im 0 = \bm 0$. This is to foreshadow that homologies will be very important when analyzing codespaces.
\end{example}

\begin{example}
  \label{ex:css-code-is-chain-complex}
  In \Cref{sec:css-code-from-ortho-classical-codes}, we show how a quantum error correcting code is defined using two parity check matrices $P_X$ and $P_Z$ that satisfy $P_X \circ P_Z^\top = 0$. As mentioned there, this is exactly the chain complex condition, and as such a CSS code corresponds to the following chain complex of length $2$:
  \[
    \begin{tikzcd}[ampersand replacement=\&, row sep=-4pt]
      \cdots \& {\bm 0} \& {\Z_2^{m_Z}} \& {Z_2^n} \& {Z_2^{m_X}} \& {\bm 0} \\
      \& {C_3} \& {C_2} \& {C_1} \& {C_0} \& {C_{-1}}
      \arrow["{P_Z^\top}", from=1-3, to=1-4]
      \arrow["{P_X}", from=1-4, to=1-5]
      \arrow["0", from=1-5, to=1-6]
      \arrow["0", from=1-1, to=1-2]
      \arrow["0", from=1-2, to=1-3]
    \end{tikzcd}
  \]
  We will see in \Cref{sec:logical-space} that the codespace is, in fact, isomorphic to the first homology module~$H_1(C_\bullet)$.
\end{example}

\section{Exact Sequences and Resolutions}

Before we delve into error correction, we need some more tools. These will be required to move from qubits to qudits. Specifically, we need to be able to define the $\Tor$ functors, and this will allow us to compute the codespace of qudit codes.

\begin{definition}[exactness]
  Suppose we have the following sequence of $R$-modules and $R$-linear maps between adjacent modules:
  \[
    \begin{tikzcd}[ampersand replacement=\&]
      \cdots \& {M_{n+1}} \& {M_{n}} \& {M_{n-1}} \& \cdots
      \arrow["{f_{n+2}}", from=1-1, to=1-2]
      \arrow["{f_{n+1}}", from=1-2, to=1-3]
      \arrow["{f_n}", from=1-3, to=1-4]
      \arrow["{f_{n-1}}", from=1-4, to=1-5]
    \end{tikzcd}
  \]
  For some $n$, if $\im f_{n+1} = \ker f_{n}$, then we say the sequence is \emph{exact} at $M_n$. If such a sequence is everywhere exact, we call it an \emph{exact sequence}.
\end{definition}
The above idea is related to that of a chain complex. Note, however, that we do not impose that $\im f_{n+1} \subseteq \ker f_n$ here. If that is the case though, then the sequence can be interpreted as a chain complex $M_\bullet$ with $\partial_n = f_n$, and the exactness at $n$ can be stated as $H_n(M_\bullet) = \bm 0$.

\begin{example}[short exact sequence]
  Ubiquitous in algebra are short sequences of the form
  \[
    \begin{tikzcd}[ampersand replacement=\&]
      {\bm 0} \& A \& B \& C \& {\bm 0.}
      \arrow["0", tail, from=1-1, to=1-2]
      \arrow["f", tail, from=1-2, to=1-3]
      \arrow["g", two heads, from=1-3, to=1-4]
      \arrow["0", two heads, from=1-4, to=1-5]
    \end{tikzcd}
  \]
  The exactness at every object within the sequence implies that $f$ is a monomorphism and $g$ is an epimorphism; we emphasize this notationally, but only here. Such a sequence can be interpreted as $C$ being (isomorphic to) the quotient $B/A$.
\end{example}

\begin{example}[isomorphism]
  An example of an even shorter exact sequence is
  \[
    \begin{tikzcd}[ampersand replacement=\&]
      {\bm 0} \& A \& B \& {\bm 0.}
      \arrow[tail, from=1-1, to=1-2]
      \arrow["f", "\sim"', from=1-2, to=1-3]
      \arrow[two heads, from=1-3, to=1-4]
    \end{tikzcd}
  \]
  The exactness implies that $f$ is an isomorphism.
\end{example}

In Homological Algebra, a special kind of exact sequence is very important to describe modules, and is important for us to define $\Tor$. We will not do so in full generality; instead, we only go as far as we need.
\begin{definition}[free resolution]
  Suppose $M$ is an $R$-module. A \emph{free resolution} of $M$ is a chain complex $F_\bullet$, where each $F_n$ is a free $R$-module, together with an $R$-linear map $\varepsilon : F_0 \to M$ such that the \emph{augmented complex}
  \[
    \begin{tikzcd}[ampersand replacement=\&]
      \cdots \& {F_1} \& {F_0} \& M \& {\bm 0}
      \arrow["{\partial_2^{F_\bullet}}", from=1-1, to=1-2]
      \arrow["{\partial_1^{F_\bullet}}", from=1-2, to=1-3]
      \arrow["\varepsilon", from=1-3, to=1-4]
      \arrow["0", from=1-4, to=1-5]
    \end{tikzcd}
  \]
  is everywhere exact. We write a resolution as $F_\bullet \xrightarrow\varepsilon M$.
\end{definition}

\begin{definition}[Tor]
  Let $A, B$ be $\Z$-modules, and let $F_\bullet \xrightarrow\varepsilon A$ be a free resolution of $A$. We define the \emph{$n^{\text{th}}$ Tor module} of $A$ and $B$ as
  \(
    \Tor_n(A,B) \deq H_n(F_\bullet \otimes B),
  \)
  that is the $n^{\text{th}}$ homology module of the chain complex
  \[
    \begin{tikzcd}[ampersand replacement=\&, column sep=large]
      \cdots \& {F_1 \otimes_\Z B} \& {F_0 \otimes_\Z B} \& {\bm 0.}
      \arrow["{\partial_2^{F_\bullet} \otimes \id_B}", from=1-1, to=1-2]
      \arrow["{\partial_1^{F_\bullet} \otimes \id_B}", from=1-2, to=1-3]
      \arrow["0", from=1-3, to=1-4]
    \end{tikzcd}
  \]
  This is no longer the augmented complex ending with $A$, just the $F_\bullet$ itself with the functor $- \otimes B$ applied to it.
\end{definition}
Note that we use the free resolution $F_\bullet \xrightarrow\varepsilon A$ to define $\Tor_n(A,B)$, but $\Tor$ is actually independent of the choice of resolution.

\begin{lemma}
  \label{lem:Tor-1-over-free}
  Let $A,B$ be $\Z$-modules, and suppose $A$ is free. Then the first torsion module $\Tor_1(A,B) = \bm 0$.
\end{lemma}
\begin{proof}
  If $A$ is free, then $A \cong \Z^{\oplus n}$ for some $n \in \N$. We write this isomorphism as $\varepsilon : \Z^{\oplus n} \isoto A$. A suitable free resolution of $A$ is the following:
  \[
    \begin{tikzcd}[ampersand replacement=\&]
      \cdots \& {\bm 0} \& {\Z^{\oplus n}} \& A \& {\bm 0.}
      \arrow["0", from=1-1, to=1-2]
      \arrow["\varepsilon", "\sim"', from=1-3, to=1-4]
      \arrow["0", from=1-4, to=1-5]
      \arrow["0", from=1-2, to=1-3]
    \end{tikzcd}
  \]
  The chain complex $F_\bullet \otimes B$ has components $F_n = \bm 0$ for $n > 0$, and $F_0 = \Z^{\oplus n} \otimes_\Z B$. Consequently, its differentials are necessarily all zero. Then the first homology is
  \[ H_1(F_\bullet \otimes B) = \faktor{\ker (\partial_1^{F_\bullet} \otimes \id_B)}{\im (\partial_2^{F_\bullet} \otimes \id_B)} = \faktor{\bm 0}{\bm 0} = \bm 0, \]
  and this is the first Tor module $\Tor_1(A,B)$.
\end{proof}

\section{Change of ring for homology}

In \Cref{cha:cell-compl}, we will use chain complexes of $\Z$-modules to study topological spaces. However, to properly study qudit systems, we need to move to complexes of $\Z_d$-modules, where $d > 1$ is arbitrary. We describe here how to do that.

\begin{definition}[change of ring for chain complex]
  \label{def:change-of-ring-for-chain-complex}
  Let $C_\bullet^R$ be a chain complex of $R$-modules, and let $f : R \to S$ be a ring homomorphism. We use the \emph{extension of scalars} from \Cref{def:extension-of-scalars} to obtain the corresponding chain complex of \mbox{$S$-modules} $C^S_\bullet \deq f_!(C_\bullet^R)$. This is the tensor product $S \otimes_R C_\bullet^R$ with component $S \otimes_R C_n^R$ and differential $\id_S \otimes \partial_n^{C_\bullet^R}$ for degree $n$.
\end{definition}

Changing the ring of a chain complex is relatively straightforward. This induces change of ring for the homology modules:
\begin{definition}[change of ring for homology]
  Let $C_\bullet$ be a chain complex of $R$-modules. Its $n^\text{th}$ homology $R$-module is $H_n(C_\bullet)$. Using the above change of scalars, we define the \emph{$n^{\text{th}}$ homology of $C_\bullet$ over $S$} as the $n^{\text{th}}$ homology evaluated in the complex~$C_\bullet^S$. That is, $H_n^S(C_\bullet) \deq H_n^S(C_\bullet^S)$. Written explicitly, it is
  \[ H_n^S(C_\bullet) = \faktor{\ker (\id_S \otimes \partial_n^{C_\bullet})}{\im (\id_S \otimes \partial_{n+1}^{C_\bullet})}. \]
  We intend this to be an $S$-module; however, recall from \Cref{note:scalar-extension-is-bimodule} that this is also still an $R$-module.
\end{definition}

We now present an important result from Homological Algebra. We do so without proof, because this would, in full generality, require concepts beyond the scope of this work. There are several version of the following theorem; we use the one directly useful to us.
\begin{theorem}[Universal Coefficient Theorem\cite{HatcherAllen2002At}]
  \label{thm:UCT}
  Let $C_\bullet$ be a chain complex of $\Z$-modules, and let $R$ be a ring. Then for each $n \in \N$, there exists a short exact sequence as follows:
  \[
    \begin{tikzcd}[ampersand replacement=\&, column sep=scriptsize]
      {\bm 0} \& {H_{n+1}(C_\bullet) \otimes_\Z R} \& {H_{n+1}^R(C_\bullet)} \& {\Tor_1\big( H_{n}(C_\bullet), R \big)} \& {\bm 0.}
      \arrow[from=1-1, to=1-2]
      \arrow[from=1-2, to=1-3]
      \arrow[from=1-3, to=1-4]
      \arrow[from=1-4, to=1-5]
    \end{tikzcd}
    \qedhere
  \]
\end{theorem}
The theorem is usually stated with $R$ being an abelian group ($\Z$-module). Here, we see it as a ring, but recall that every ring is an abelian group. We need the theorem to prove the following lemma for later use.

\begin{lemma}
  \label{lem:H1-change-of-ring}
  Let $C_\bullet$ be a chain complex of $\Z$-modules, and let $R$ be a ring. Suppose that $H_n(C_\bullet)$ is free for some $n \in \N$. Then the $(n+1)^{\text{st}}$ homology over $R$ is
  \[
    H_{n+1}^R (C_\bullet) \cong H_{n+1}(C_\bullet) \otimes_\Z R.
    \qedhere
  \]
\end{lemma}
\begin{proof}
  By hypothesis, $H_n(C_\bullet)$ is a free $\Z$-module. It follows from \Cref{lem:Tor-1-over-free}, that $\Tor_1(H_n(C_\bullet), R) = \bm 0$. Then the short exact sequence from \Cref{thm:UCT} becomes
  \[
    \begin{tikzcd}[ampersand replacement=\&,column sep=scriptsize]
      {\bm 0} \& {H_{n+1}(C_\bullet) \otimes_\Z R} \& {H_{n+1}^R(C_\bullet)} \& {\bm 0.}
      \arrow[from=1-1, to=1-2]
      \arrow[from=1-2, to=1-3]
      \arrow[from=1-3, to=1-4]
    \end{tikzcd}
  \]
  The exactness of the sequence implies an isomorphism, as required.
\end{proof}

\Cref{lem:H1-change-of-ring} seems perhaps insufficiently motivated: why would we care specifically about what happens to the $(n+1)^{\text{st}}$ homology when the $n^{\text{th}}$ homology is free? In a general chain complex, this may be completely useless. As shown in \Cref{ex:classical-code-is-chain-complex,ex:css-code-is-chain-complex}, we intend to build error correcting codes using chain complexes. However, we will specifically use chain complexes arising from topological spaces, where \Cref{lem:H1-change-of-ring} will have great relevance. First, however, we need to define those spaces; we do so in the next chapter.

\chapter[Cell Complexes]{Cell Complexes \hspace*{\fill} \begin{tikzpicture}
  \tikzstyle{point}=[fill=black,draw=none,circle,minimum size=3pt,inner sep=0]
  \tikzstyle{gluon}=[draw=blue]
  \begin{scope}[local bounding box=X0, rotate={15}]
    \node[point] (v0) at (0,0) {};
    \node[point] (v1) at (0.2,0.3) {};
    \node[point] (v2) at (0.5,0) {};
  \end{scope}
  \begin{scope}[local bounding box=X1, shift={($(X0.south)+(0.7,0.05)$)}, rotate={-15}]
    \node (v0e) at (0,0) {};
    \node (v1e) at (0.2,0.3) {};
    \node (v2e) at (0.5,0) {};
    \draw (v0e.center) -- node[midway] (e01) {} (v1e.center) -- node[midway] (e12) {} (v2e.center) -- node[midway] (e02) {} (v0e.center);
  \end{scope}
  \begin{scope}[local bounding box=X2, shift={($(X1.south)+(1,0.45)$)}, rotate={-30}]
    \node (v0f) at (0,0) {};
    \node (v1f) at (0.2,0.3) {};
    \node (v2f) at (0.5,0) {};
    \draw[fill=gray, fill opacity=0.5] (v0f.center) -- node[midway] (e01f) {} (v1f.center) -- node[midway] (e12f) {} (v2f.center) -- node[midway] (e02f) {} (v0f.center);
  \end{scope}
  \draw[gluon] (v0.center) .. controls ($(v0.center)+(0,-0.1)$) and ($(v0e.center)+(-0.5,-0.1)$) .. (v0e.center);
  \draw[gluon] (v1.center) -- (v1e.center);
  \draw[gluon] (v2.center) .. controls ($(v2.center)+(0,-0.2)$) and ($(v2e.center)+(0,-0.2)$) .. (v2e.center);

  \draw[gluon] (e01.center) -- (e01f.center);
  \draw[gluon] (e12.center) -- (e12f.center);
  \draw[gluon] (e02.center) -- (e02f.center);
\end{tikzpicture}}
\label{cha:cell-compl}

This chapter introduces a way to split topological spaces into simple parts that are glued together. This allows us to establish algebraic structures (chain complexes from \Cref{sec:chain-complex-abstract}) that describe these spaces, leading us into algebraic topology.
We take a combinatorial approach, where we define these parts and their relationships abstractly, without worrying too much about the actual underlying topology or geometry.

We assume knowledge of basic topology, and we will not talk about open sets and similar concepts here, because they are not needed for us. Some knowledge of cell complexes (also called CW complexes) would be perhaps beneficial to see deeper into the background, but it is not strictly necessary.

\section{Abstract $2$-dimensional cell complexes}

What we need is a \emph{complex} of simple topological objects glued together in a useful way. There are many possibilities; the one we choose is a \emph{cell complex} (also called a \emph{CW complex}), because this is general enough for the kinds of spaces we need. There are other kinds of topological complexes, and we comment on our choice of cell complexes instead of the other options in \Cref{sec:other-kinds-abstract-top-complexes}. For more details on cell and other kinds of complexes, see the book~\cite{HatcherAllen2002At}, which served as our guide.

First, some intuition of a \emph{cell complex}. The idea is that we build a space inductively by attaching new \emph{cells} of increasing dimension, starting from $0$-dimensional points (called \emph{$0$-cells}). At each step, there is a ``frame'' of dimension $n$, which we call the \emph{$n$-skeleton} $X^n$. This is composed of cells of dimension $n$ (called \emph{$n$-cells}) and lower, glued together in a previous step. We attach $n+1$ dimensional ``faces'' to it to fill the holes in the frame. These faces are copies of the \emph{closed topological $(n+1)$-dimensional ball} $\mathcal{B}^{n+1}$, and the gluing is done along their boundaries, which are copies of the \emph{$n$-dimensional sphere} $\mathcal{S}^n$.

\begin{example}
  \label{ex:gluing-triangles-cell-complex}
  We show an example of this procedure in \cref{fig:gluing-two-triangles-CW}. We take two triangles, which are our $2$-cells $e^2_1$ and $e^2_2$ (where the upper index indicates the dimension), and glue them to a $1$-skeleton $X^1$ constructed from edges and points in a previous step. We always glue by boundaries, as shown in the figure. The result is a $2$-dimensional space, which we call the \emph{$2$-skeleton}. In this case, this is the final step, and we say the space $X = X^2$.
  \begin{figure}[ht]
    \centering
    \begin{tikzpicture}
  \begin{scope}[local bounding box=D1]
    \draw[dashed, fill=\colora, fill opacity=0.5] (0,2) -- node[midway] (D11) {} (0,0) -- node[midway] (D12) {} (2,0) -- node[midway] (D13) {} cycle;
    \node at (0.5, 0.5) {$e^2_1$};
  \end{scope}
  \begin{scope}[local bounding box=skeleton, shift={($(D1.east)+(1,-1)$)}]
    \draw[dashed] (2,0) -- node[midway] (s1) {} (0,0) -- node[midway] (s2) {} (0,2) -- node[midway] (s3) {} (2,0) -- node[midway] (s4) {} (2,2) -- node[midway] (s5) {} (0,2);
    \node at (0.5, 0.5) {$X^1$};
  \end{scope}
  \begin{scope}[local bounding box=D2, shift={($(skeleton.east)+(1,-1)$)}]
    \draw[dashed, fill=\colorb, fill opacity=0.5] (0,2) -- node[midway] (D21) {} (2,2) -- node[midway] (D22) {} (2,0) -- node[midway] (D23) {} cycle;
    \node at (1.5, 1.5) {$e^2_2$};
  \end{scope}

  \tikzstyle{glueline}=[]

  \draw[glueline]
  (D12.center)
  .. controls ($(D12.center)+(0,-0.5)$) and ($(s1.center)+(0,-0.5)$)
  .. (s1.center);

  \draw[glueline]
  (D11.center)
  .. controls ($(D11.center)+(-1,0.7)$) and ($(s2.north)+(-1,0.5)$)
  .. (s2.north);

  \draw[glueline]
  (D13.south east)
  .. controls ($(D13.south east)+(0.5,0.5)$) and ($(s3.north west)+(-0.5,-1)$)
  .. (s3.north west);

  \draw[glueline]
  (D22.center)
  .. controls ($(D22.center)+(1,-0.7)$) and ($(s4.south)+(1,-0.5)$)
  .. (s4.south);

  \draw[glueline]
  (D21.center)
  .. controls ($(D21.center)+(0,0.5)$) and ($(s5.center)+(0,0.5)$)
  .. (s5.center);

  \draw[glueline]
  (D23.north west)
  .. controls ($(D23.north west)+(-0.5,-0.7)$) and ($(s3.south east)+(0.7,1)$)
  .. (s3.south east);

  \node[scale=1.5] at ($(D2.east)+(1,0)$) (mapsto) {$\rightsquigarrow$};
  \begin{scope}[shift={($(mapsto.east)+(1,-1)$)}]
    \draw (0,0) -- (0,2) -- (2,2) -- (2,0) -- cycle;
    \draw (2,0) -- (0,2);
    \node at (0.5,0.5) {$X^2$};
    \begin{pgfonlayer}{background}
      \path[fill=\colora, fill opacity=0.5] (0,2) -- (0,0) -- (2,0);
      \path[fill=\colorb, fill opacity=0.5] (0,2) -- (2,2) -- (2,0);
    \end{pgfonlayer}
  \end{scope}
\end{tikzpicture}
    \caption[Gluing two triangles to a frame made of edges]{Gluing two triangles to a frame made of edges. The dashed edges are the boundaries of the triangles ($2$-cells $e^2_1$ and $e^2_2$), and the $1$-skeleton $X^1$. The solid bendy lines show how the triangle edges are glued to the skeleton. We do not show the vertices here, but they are at the endpoints of the edges, and they are glued as appropriate.}
    \label{fig:gluing-two-triangles-CW}
  \end{figure}
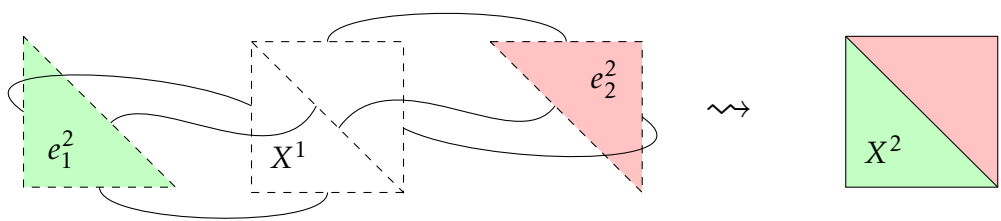
\end{example}

\subsection{Definition}

The above \cref{ex:gluing-triangles-cell-complex} is quite simple.
However, our goal is \emph{homology of cell complexes}, and this is tricky to define in full generality.\footnote{As we shall see in \Cref{sec:chain-complexes-from-cell-complexes}, we attach a chain complex to a cell complex, and the differentials correspond to boundaries where cells are attached. In order to define the boundary maps in full generality, we need to define the \emph{degree} of a continuous map, and this in turn is defined using \emph{relative homology}. Thus, we need one homology theory to define another, more complicated homology theory. A standard way to do this (see \cite{HatcherAllen2002At}) is by first defining \emph{simplicial}, and then \emph{singular homology}, and then using this to define the aforementioned boundary maps for cellular homology.} Thus, we take an alternative approach.
For our purposes, we restrict the maximum dimension to two, because we only work with surfaces. We also abstract away much of the topological structure. What remains is an abstract representation of a \emph{cell complex}, which is in analogy to what an \emph{abstract simplicial complex} is to a \emph{topological simplicial complex}.
The concrete definitions are original, though following the general definition of a cell complex given by \cite{HatcherAllen2002At}.
It took nontrivial effort to get the definition right for our purposes. Because of this, we show it in steps, motivating each design choice.

\subsubsection*{$0$-cells}

When building a \emph{cell complex}, we first need to define a set of discrete points or vertices. This will be the initial step for us too. Denote $X_0 \deq \{ e^0_\alpha \}_\alpha$ a family of $0$-cells, which are abstract symbols, and are all distinct. The $0$-skeleton is the same: $X^0 \deq X_0$.

\subsubsection*{$1$-cells}

Next is a family of $1$-cells, edges between points. Denote $X_1 \deq \{ e^1_\alpha \}_\alpha$ a family of these; they are again just abstract symbols, distinct from each other and from $0$-cells. For an edge $e^1_\alpha \in X_1$, we need to specify its endpoint vertices: the \emph{source} $s(e^1_\alpha) \in X_0$ and \emph{target} $t(e^1_\alpha) \in X_0$. We define the function $s,t : X_1 \to X_0$. We allow the source and target to be the same vertex, or equivalently, an edge may be a self-loop.

In cell complex style, we combine both the source and target functions into a single \emph{attaching map} $\varphi_1 : X_1 \to X_0 \times X_0$ that sends $e^1_\alpha \mapsto (s(e^1_\alpha), t(e^1_\alpha))$.
Note that so far, we have essentially defined a \emph{directed graph}, with vertices $X_0$ and edges~$X_1$.\footnote{In \Cref{sec:larger-torsion}, we will see that it is interesting if a $2$-cell touches the same $1$-cell several times. It seems that this is not possible for $1$-cells and $0$-cells; however, we could generalize the $1$-cells to \emph{oriented hyperedges} to get a similar effect.} This forms the $1$-skeleton $X^1 \deq X^0 \cup X_1$.

Unlike general cell complexes which are unoriented, we are building the orientation of $1$-cells into the definition. This is a choice that will make things more convenient in the definition of $2$-cells in the next step, and the chain complexes that we will define in \Cref{sec:chain-complexes-from-cell-complexes} are uniquely, and easily, determined from this. Furthermore, in \Cref{sec:gluing-complexes}, we will glue cells together to form interesting spaces, and sometimes we want to glue them in opposite orientation.

This means also that while a $1$-cell has a defined orientation, we will need to reason about a copy of the same cell with the opposite orientation. We set up the following machinery:

\begin{definition}[orientation closure]
  \label{def:orientation-closure}
  For a set $A$, define $+A \deq \{ +a : a \in A \}$, and similarly  $-A \deq \{ -a : a \in A \}$. By this we mean that we attach a symbol $+$ or $-$ to the elements. Define the \emph{orientation closure} of $A$ to be the set $\overrightarrow{A} \deq (+A) \cup (-A)$. Define also an action of $e^{i\pi\Z} \deq (\{+1, -1\}, \cdot, +1)$, the cyclic group of order two in multiplicative notation, such that $-1$ switches the signs. We write this as $-(\pm a) = \mp a$ for $a \in A$. By definition of an action, $+1$ acts as the identity.

  We may use the usual multiplicative convention of omitting the symbol $+$, and this way we notationally identify $A \cong +A$. Then $\overrightarrow{A}$ is the set of elements of $A$ and their formal additive inverses.
\end{definition}

We use the above construction to extend the attaching map $\varphi_1$, by requiring that for each $e^1 \in X_1$, the cell $-e^1 \in \overrightarrow{X_1}$ is the same edge, but with direction reversed:
\begin{align*}
  s(-e^1) &\deq t(e^1)
  & \text{and} &
  & t(-e^1) & \deq s(e^1).
\end{align*}
Formally, we define this extension of $\varphi_1$ to be $\overrightarrow{\varphi_1} : \overrightarrow{X_1} \to X_0 \times X_0$, and we require%
\hiddenparbreak
\diagramonright{0.7}{0.3}{%
  \vspace{0.8ex}
  \noindent
   that it is compatible with the action of $e^{i \pi \Z}$ on $X_0 \times X_0$, defined such that $-1$ swaps the entries: $(e^0_1, e^0_2) \xmapsto{-1} (e^0_2, e^0_1)$. The compatibility then means that the diagram on the right commutes.}
{%
  \centering
  \begin{tikzcd}[ampersand replacement=\&, column sep=1.5em]
    {\overrightarrow{X_1}} \& {\overrightarrow{X_1}} \\
    {X_0 \times X_0} \& {X_0 \times X_0}
    \arrow["{\pm 1}"', from=2-1, to=2-2]
    \arrow["{\overrightarrow{\varphi_1}}"', from=1-1, to=2-1]
    \arrow["{\overrightarrow{\varphi_1}}", from=1-2, to=2-2]
    \arrow["{\pm 1}", from=1-1, to=1-2]
  \end{tikzcd}
}

\subsubsection*{$2$-cells}

The highest dimension we reach is two, i.e. faces. Denote $X_2 \deq \{ e^2_\alpha \}_\alpha$ a family of $2$-cells, again abstract symbols, distinct from $0$- and $1$-cells. In the usual definition of a cell complex, the boundary is a $1$-sphere $\mathcal{S}^1$, a circle, so we need an abstract analogy for this. As the $1$-skeleton $X^1$ is essentially a directed graph, the obvious way to define a circle is to use a cycle of edges in the graph. However, we need to take the orientation of the edges into account.

We define an abstract circle as a sequence $(\pm_i e^1_i)_{i=1}^\ell$ of length $\ell \in \N$ which consists of elements of $\overrightarrow{X_1}$. We denote their signs with $\pm_i$, where the $i$ means this is the sign of $i^{\text{th}}$ cell in the sequence. In our context, some boundary has to always exist, so we exclude the empty cycle ($\ell = 0$). The sequence must form a cycle in the graph, so the endpoints of adjacent edges have to match:
\begin{align*}
  t(\pm_i e^1_i) & = s(\pm_{i+1} e^1_{i+1}) \quad \text{for every $i=1,\dots,\ell-1$; and}
  & t(\pm_\ell e^1_\ell) & = s(\pm_1 e^1_1).
\end{align*}
Our $2$-cells are oriented as well, so their boundary circles have a definite direction by which they traverse the edges. We show an example of this in \cref{fig:abstract-1-sphere-examples:plus-minus}.
\begin{figure}[ht]
  \centering
  \subfloat[\label{fig:abstract-1-sphere-examples:plus-minus}]{
    \begin{tikzpicture}
  \tikzstyle{vertex}=[fill=black,draw=black,circle,minimum size=4pt,inner sep=0]
  \tikzstyle{->-}=[
  decoration={
    markings,
    mark=at position 0.6 with {\arrow{Triangle[scale=1]}}},
  postaction={decorate},
  draw]
  \node[vertex] (v1) at (90:1) {};
  \node (v2) at (30:1) {};
  \node (v3) at (-30:1) {};
  \node (v4) at (-90:1) {};
  \node (vl) at (150:1) {};
  \node at (0,0) {\Huge $\circlearrowright$};
  \draw[->-] (v1.center) -- node[midway, above right] {$e^1_1$} (v2.center);
  \draw[->-] (v3.center) -- node[midway, right] {$e^1_2$} (v2.center);
  \draw[->-] (v4.center) -- node[midway, below right] {$e^1_3$} (v3.center);
  \draw[->-] (vl.center) -- node[midway, above left] {$e^1_\ell$} (v1.center);
  \node (a1) at ($(v4.center)+(-180:0.8)$) {};
  \node (a2) at ($(vl.center)+(240:0.8)$) {};
  \draw[dashed] (v4.center) .. controls (a1) and (a2) .. (vl.center);
  \node at (-90:2) {boundary $(+e^1_1, -e^1_2, -e^1_3, \dots, +e^1_\ell)$};
  \begin{pgfonlayer}{background}
    \path[fill=gray, fill opacity=0.3]
    (v1.center) -- (v2.center) -- (v3.center) -- (v4.center) .. controls (a1) and (a2) .. (vl.center) -- cycle;
  \end{pgfonlayer}
\end{tikzpicture}
  }
  \qquad
  \subfloat[\label{fig:abstract-1-sphere-examples:0-cell}]{
    \begin{tikzpicture}
  \tikzstyle{vertex}=[fill=black,draw=black,circle,minimum size=4pt,inner sep=0]
  \tikzstyle{->-}=[
  decoration={
    markings,
    mark=at position 0.5 with {\arrow{Triangle[scale=1]}}},
  postaction={decorate},
  draw]
  \node[vertex, label={-30:$e^0_1$}] (v1) at (-30:1.2) {};
  \node at (-90:0.5) {\Huge $\circlearrowright$};
  \node at (-90:2) {boundary $e^0_1$};
  \begin{pgfonlayer}{background}
    \shade[ball color = gray!40, opacity = 0.4] (0,0) circle (1.5);
  \end{pgfonlayer}
\end{tikzpicture}
  }
  \caption[Examples of boundaries of a $2$-cell]{Examples of boundaries of a $2$-cell (shaded). \textbf{(a)} A boundary that is a circle ($1$-sphere) oriented clockwise. It is represented by a sequence of oriented edges $(+e^1_1, -e^1_2, -e^1_3, \dots, +e^1_\ell)$. We indicate the start and end point of the cycle by a black circle. \textbf{(b)} A $2$-cell bounded by a single point $e^0_1$. It is a sphere $\mathcal{S}^2$ with a distinguished point.}
  \label{fig:abstract-1-sphere-examples}
\end{figure}

So far, the sequences have a distinguished base point: this is the vertex $s(\pm_1 e^1_1)$. However, we do not need it -- we just need to know which edges the circle traverses, and in which direction. Thus we choose to forget the base point, and identify all sequences of the same length that are cyclic permutations of one another.
We summarize this as follows:
\begin{definition}[abstract circle]
  \label{def:abstract-circle}
  Let $X_0, X_1$ be sets of $0$- and $1$-cells respectively. An \emph{abstract circle ($1$-sphere)} is an equivalence class of cycles, up to cyclic permutations, represented by tuples in $\overrightarrow{X_1}$ as described above. We denote such equivalence classes as $\lParen \pm_1e^1_1, \dots, \pm_\ell e^1_\ell \rParen$, and we denote the set of all abstract circles $\mathcal{S}\overrightarrow{X_1}$.
\end{definition}

Similarly to $1$-cells, we need to consider reversing the orientation. An abstract circle is already oriented, so the only remaining thing needed is the action of $e^{i\pi\Z}$ on $\mathcal{S}\overrightarrow{X_1}$. We define
\[ - \lParen \pm_1 e^1_1,\ \pm_2 e^1_2, \dots, \pm_{\ell-1}e^1_{\ell-1},\ \pm_\ell e^1_\ell \rParen
  \deq \lParen \mp_\ell e^1_\ell,\ \mp_{\ell-1} e^1_{\ell-1}, \dots, \mp_2 e^1_2,\ \mp_1 e^1_1  \rParen \]
for all abstract circles. This means that traversing a circle backward, we see the edges in the opposite orientation, and in the opposite order, as expected.

\textbf{Special case:} A general cell complex allows a $2$-cell to be bounded by a single $0$-cell, without any $1$-cells in between. This corresponds to a $2$-sphere $\mathcal{S}^2$ with a distinguished point, as shown in \Cref{fig:abstract-1-sphere-examples:0-cell}. We have no use for this, so in order to keep the definitions simpler, we exclude this case. This means that our definition is even more restrictive, but that is fine.

Finally, we define the \emph{attaching map} $\varphi_2 : X_2 \to \mathcal{S}\overrightarrow{X_1}$ that maps a $2$-cell to its boundary which is an (oriented) abstract circle.
Similarly to the case of\mbox{ $1$-cells}, we\hiddenparbreak
\diagramonright{0.7}{0.3}{
  \vspace{-0.25ex}
  \noindent
  extend $\varphi_2$ to the orientation closure of $2$-cells, and define $\overrightarrow{\varphi_2} : \overrightarrow{X_2} \to \mathcal{S}\overrightarrow{X_1}$. This extension must be compatible with the action of $e^{i \pi \Z}$, meaning the diagram to the right must commute.
}{%
  \centering
  \begin{tikzcd}[ampersand replacement=\&]
    {\overrightarrow{X_2}} \& {\overrightarrow{X_2}} \\
    {\mathcal{S} \overrightarrow{X_1}} \& {\mathcal{S} \overrightarrow{X_1}}
    \arrow["{\overrightarrow{\varphi_2}}"', from=1-1, to=2-1]
    \arrow["{\overrightarrow{\varphi_2}}", from=1-2, to=2-2]
    \arrow["{\pm 1}", from=1-1, to=1-2]
    \arrow["{\pm 1}"', from=2-1, to=2-2]
  \end{tikzcd}
}

\subsubsection*{Full definition}

We bound the dimension to two, so there will be no $n$-cells of dimension $n>2$. Hence, we now have all the ingredients required for the whole definition. Note that for reasons of practicality, we require our complexes to be \emph{finite}.

\begin{definition}[abstract $2$-dimensional cell complex]
  \label{def:abstract-2-cell-complex}
  Let \mbox{$X_n \deq \{ e^n_\alpha \}_\alpha$} be a \emph{finite} family of \emph{$n$-cells}, where $n \in \{ 0, 1, 2 \}$ is called the \emph{dimension} of the cells. The $e^n_\alpha$ are abstract symbols, and all are distinct. Let $\varphi_1 : X_1 \to X_0 \times X_0$ and $\varphi_2 : X_2 \to \mathcal{S}\overrightarrow{X_1}$ be \emph{attaching maps} as described above. We call the tuple $X \deq (X_0, X_1, X_2, \varphi_1, \varphi_2)$, also denoted $X = (X_\bullet, \varphi_\bullet)$, an \emph{abstract $2$-dimensional cell complex}.
\end{definition}

\begin{example}
  \label{ex:acc2-square}
  \begin{figure}[h]
    \centering
    \begin{tikzpicture}
  \tikzstyle{vertex}=[fill=black,draw=black,circle,minimum size=4pt,inner sep=0]
  \tikzstyle{->-}=[
  decoration={
    markings,
    mark=at position 0.55 with {\arrow{Triangle[scale=1]}}},
  postaction={decorate},
  draw]
  \node[vertex, label={-135:$e^0_1$}] at (0,0) (v1) {};
  \node[vertex, label={135:$e^0_2$}] at (0,3) (v2) {};
  \node[vertex, label={45:$e^0_3$}] at (3,3) (v3) {};
  \node[vertex, label={-45:$e^0_4$}] at (3,0) (v4) {};
  \draw[->-] (v1.center) -- node[midway, left] {$e^1_1$} (v2.center);
  \draw[->-] (v2.center) -- (v4.center);
  \node at (1.4,2.2) {$e^1_2$};
  \draw[->-] (v4.center) -- node[midway, below] {$e^1_3$} (v1.center);
  \draw[->-] (v2.center) -- node[midway, above] {$e^1_4$} (v3.center);
  \draw[->-] (v3.center) -- node[midway, right] {$e^1_5$} (v4.center);
  \node at (1,0.7) (f1) {$e^2_1$};
  \node[scale=3] at (f1.center) {$\circlearrowright$};
  \node at (2.3,2) (f2) {$e^2_2$};
  \node[scale=3] at (f2.center) {$\circlearrowright$};
  \begin{pgfonlayer}{background}
    \path[fill=\colora, fill opacity=0.5] (0,3) -- (0,0) -- (3,0);
    \path[fill=\colorb, fill opacity=0.5] (0,3) -- (3,3) -- (3,0);
  \end{pgfonlayer}
\end{tikzpicture}
    \caption[Example of an abstract $2$-dimensional cell complex]{Example of an abstract $2$-dimensional cell complex.}
    \label{fig:square-abs-cell-complex}
  \end{figure}
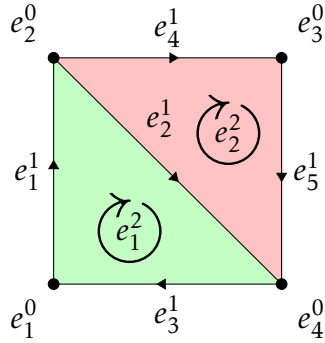
  Take the space from \cref{ex:gluing-triangles-cell-complex}. In that example, we showed the gluing of $2$-cells (triangles) to the $1$-skeleton to make a $2$-skeleton (square). Now, we construct the same space using our new definition. Note that this requires that we decide the orientation of the faces and edges; we indicate this by arrows in \cref{fig:square-abs-cell-complex}.

  The $0$-cells are $X_0 = \{ e^0_1, \dots, e^0_4 \}$. The $1$-cells are $X_1 = \{ e^1_1, \dots, e^1_5 \}$, and they are attached to the $0$-cells using $\varphi_1$ as follows:
  \begin{align*}
    \varphi_1(e^1_1) & = (e^0_1, e^0_2),
    & \varphi_1(e^1_4) & = (e^0_2, e^0_3), \\
    \varphi_1(e^1_2) & = (e^0_2, e^0_4),
    & \varphi_1(e^1_5) & = (e^0_3, e^0_4), \\
    \varphi_1(e^1_3) & = (e^0_4, e^0_1),
  \end{align*}
  Together, this makes the $1$-skeleton. The $2$-cells are $X_2 = \{ e^2_1, e^2_2 \}$, and they are attached by $\varphi_2$ as follows:
  \begin{align*}
    \varphi_2(e^2_1) & = \lParen +e^1_1, +e^1_2, +e^1_3 \rParen,
    & \text{and} &
    & \varphi_2(e^2_1) & = \lParen -e^1_2, +e^1_4, +e^1_5 \rParen.
  \end{align*}
\end{example}

\begin{definition}[subcomplex]
  \label{def:subcomplex}
  Let $X = (X_\bullet, \varphi_\bullet)$ be an abstract $2$-dimensional cell complex. Define a \emph{subcomplex} of $X$ to be a complex $Y = (Y_\bullet, \psi_\bullet)$, where $Y_n \subseteq X_n$ and the attaching maps are restrictions on these subsets, i.e. $\psi_n = \varphi_n|_{Y_n}$, for each $n\in\{0,1,2\}$. For this to work, we require that if an $(n+1)$-cell $e^{n+1} \in X_{n+1}$ is in $Y_{n+1}$, then all the $n$-cells forming its boundary are in $Y_n$, for $n \in \{ 0, 1\}$.
\end{definition}

\subsection{Gluing Complexes}
\label{sec:gluing-complexes}

To obtain interesting spaces, we will glue these complexes, or different cells of the same complex, together. To do this, we define the notion of a quotient. First, we describe how to glue a single complex to itself; gluing different complexes will be done by first joining them together disjointly, and then quotienting.

We restrict ourselves to only glue cells of the same dimension. No contractions are allowed, meaning we cannot transform an $n$-cell to a cell of lower dimension. On the other hand, we wish to be able to glue cells in an oriented manner, so that we may, for example, construct a cylinder and M\"obius strip from the same complex by gluing edges in different ways (see \Cref{ex:cyl-Mob}).

\begin{definition}[oriented equivalence]
  Let $A$ be a set. Take its orientation closure $\overrightarrow{A}$, together with the action of $e^{i\pi\Z}$ (see \Cref{def:orientation-closure}). A set-theoretic equivalence relation $\sim$ on $\overrightarrow{A}$ is called \emph{oriented} if it does not relate an element to its inverse, that is $\pm a \nsim \mp a$ for $a \in A$, and if it is compatible with the action of $e^{i\pi\Z}$, meaning that $\pm_a a \sim \pm_b b$ implies $\mp_a a \sim \mp_b b$ for $\pm_a a$ and $\pm_b b \in \overrightarrow{A}$.
\end{definition}

\begin{definition}[coherent equivalence of cells]
  Let $X = (X_\bullet, \varphi_\bullet)$ be an abstract \mbox{$2$-dim}ensional cell complex. Suppose we have an equivalence relation $\sim_0$ on the set of points $X_0$ which have no orientation. Suppose further we have an oriented equivalence relation $\sim_1$ on the orientation closure of edges $\overrightarrow{X_1}$, and similarly $\sim_2$ on $\overrightarrow{X_2}$. These can be joined together as
  $\sim \;\deq\; \sim_0 \cup \sim_1 \cup \sim_2$,
  an equivalence relation on $X_0 \cup \overrightarrow{X_1} \cup \overrightarrow{X_2}$.
  We call $\sim$ \emph{coherent} if the equivalences are compatible with the attaching maps:
  \begin{itemize}
  \item Suppose $e^1_a$ and~$e^1_b \in X_1$ are $1$-cells. If $\pm_a e^1_a \sim \pm_b e^1_b$, then $s(\pm_a e^1_a) \sim s(\pm_b e^1_b)$ and $t(\pm_a e^1_a) \sim t(\pm_b e^1_b)$. That is, the endpoints of equivalent edges must match, and this must respect the orientation.
  \item Suppose $e^2_a$ and $e^2_b \in X_2$ are $2$-cells with boundary circles $\varphi_2(e^2_a)$ and $\varphi_2(e^2_b)$. If $\pm_a e^2_a \sim \pm_b e^2_b$,
    then $\overrightarrow{\varphi}^2(\pm_a e^2_a) \simeq \overrightarrow{\varphi}^2(\pm_b e^2_b)$. By this we mean that their boundary circles must match, respecting the orientations, up to equivalence of $1$-cells, indicated by the symbol $\simeq$. \qedhere
  \end{itemize}
\end{definition}

\begin{definition}[quotient of abstract cell complex]
  Let $X = (X_\bullet, \varphi_\bullet)$ be an abstract $2$-dimensional cell complex, and let $\sim$ be a coherent equivalence of the cells of $X$. Define the \emph{quotient} of $X$ by $\sim$, denoted $X/\sim$, to be the abstract $2$-dimensional cell complex $(Q_\bullet, \varkappa_\bullet)$ obtained by identifying the cells that are equivalent under $\sim$ and inducing the attaching maps on the equivalence classes.
\end{definition}

\subsubsection*{Examples}

We show a few important examples. In the following, let $X=(X_\bullet, \varphi_\bullet)$ be an~abstract $2$-dimensional cell complex corresponding to a square, with four vertices $X_0 = \{ e^0_1, \dots, e^0_4 \}$, four edges $X_1 = \{ e^1_1, \dots, e^1_4 \}$, and a face $X_2 = \{ e^2 \}$. The edges attach as $\varphi_1(e^1_1) = (e^0_1, e^0_2)$, etc.; and the face attaches as $\varphi_2(e^2) = \lParen e^1_1, e^1_2, e^1_3, e^1_4 \rParen$. See \Cref{fig:gluing-square:square}.

\begin{figure}[h]
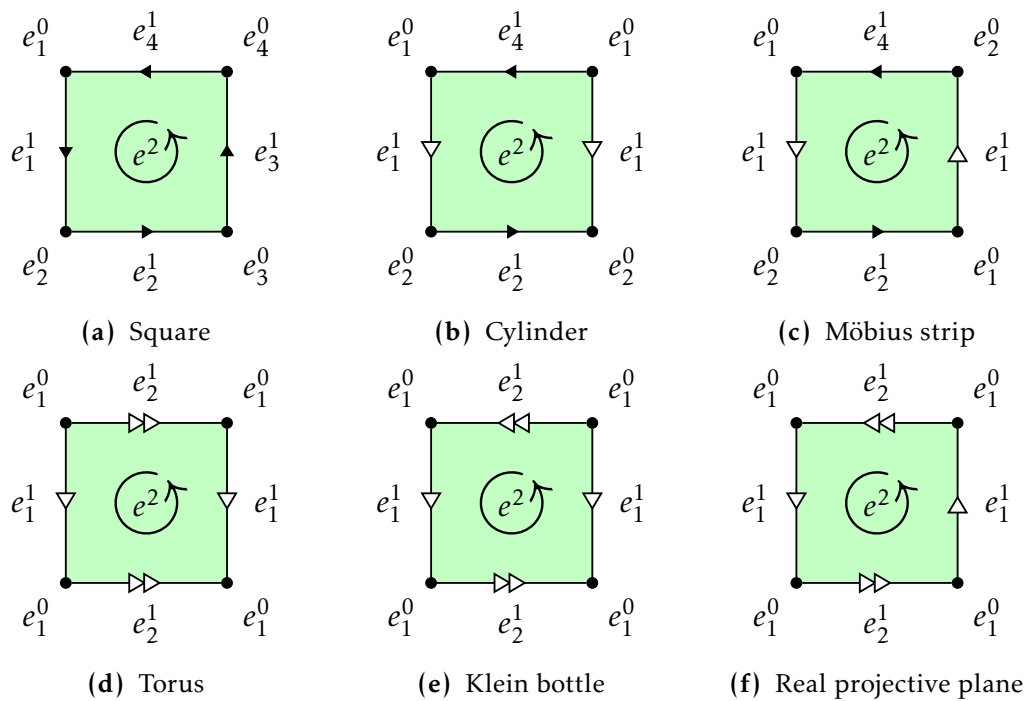

  \centering
  \subfloat[\label{fig:gluing-square:square} Square]{\gsNoglue}
  \qquad
  \subfloat[\label{fig:gluing-square:cylinder} Cylinder]{\gsCylinder}
  \qquad
  \subfloat[\label{fig:gluing-square:mobius} M\"obius strip]{\gsMobius}
  \qquad
  \subfloat[\label{fig:gluing-square:torus} Torus]{\gsTorus}
  \qquad
  \subfloat[\label{fig:gluing-square:klein} Klein bottle]{\gsKlein}
  \qquad
  \subfloat[\label{fig:gluing-square:proj} Real projective plane]{\gsProj}
  \qquad
  \caption[Gluing a square to itself in different ways]{Gluing a square to itself in different ways. Pairs of edges with the same number of white arrows are glued together, in the orientation indicated by those arrows. Edges with black arrows are not glued.}
\end{figure}

\begin{example}[cylinder, M\"obius strip]
  \label{ex:cyl-Mob}
  If we put $e^1_1 \sim -e^1_3$, the quotient we obtain is a \emph{cylinder} in \Cref{fig:gluing-square:cylinder}. Observe that because of the way we constructed the square $X$, we have to glue $e^1_1$ to $e^1_3$ in the opposite orientation. This induces the equivalences of points $e^0_1 \sim e^0_4$ and $e^0_2 \sim e^0_3$. The boundary circle of the face becomes $\lParen e^1_1, e^1_2, -e^1_1, e^1_4 \rParen$. In the figure, we indicate the gluing of $e^1_1 \sim -e^1_3$ by the white triangles. In \Cref{fig:cylinder_3d}, we show the result of the gluing in three dimensions.

  On the other hand, if we put $e^1_1 \sim e^1_3$, the quotient is a \emph{M\"obius strip} in \Cref{fig:gluing-square:mobius}, with the face now bounded by the circle $\lParen e^1_1, e^1_2, e^1_1, e^1_3 \rParen$. Observe that the white triangle on the right is now oriented the other way.
\end{example}
\begin{figure}[h]
  \centering
  \begin{tikzpicture}
  \newcommand{\h}{1.5}
  \newcommand{\s}{40}
  \draw[domain=120:360+120, samples=\s, variable=\t, ->-]
  plot ({cos(-\t)},0,{sin(-\t)});

  \foreach \i [evaluate={\t={\i*360/\s}}] in {0,...,\s}{
    \fill[gray, fill opacity=0.15]
    ({cos(\t)}, 0, {sin(\t)})
    -- ({cos(\t+\s)}, 0, {sin(\t+\s)})
    -- ({cos(\t+\s)}, \h, {sin(\t+\s)})
    -- ({cos(\t)}, \h, {sin(\t)})
    -- cycle;
  }

  \draw[domain=-120:270, samples=\s, variable=\t, ->-]
  plot ({cos(\t)},\h,{sin(\t)});

  \coordinate (b) at ({cos(150)}, 0, {sin(150)});
  \draw[-g>-] ($(b)+(0,\h,0)$) -- node[midway,left] {$e^1_1$} (b);

  \node[vertex, label={-120:$e^0_2$}] at (b) {};
  \node[vertex, label={150:$e^0_1$}] at ($(b)+(0,\h,0)$) {};
  \node at ({1.3*cos(-30)}, 0, {1.3*sin(-30)}) {$e^1_2$};
  \node at ({1.3*cos(-30)}, \h, {1.3*sin(-30)}) {$e^1_4$};

  \node (f) at ({cos(75)}, {0.6*\h}, {sin(75)}) {$e^2$};
  \node[scale=3] at (f.center) {$\circlearrowleft$};
\end{tikzpicture}
  \caption[Cylinder from \protect\Cref*{fig:gluing-square:cylinder} in 3D]{Cylinder from \protect\Cref{fig:gluing-square:cylinder} in 3D.}
  \label{fig:cylinder_3d}
\end{figure}

\begin{example}[torus]
  \label{ex:gs-torus}
  Now, let $e^1_1 \sim -e^1_3$ and $e^1_2 \sim -e^1_4$. This means we glue two pairs of edges, and the quotient is a \emph{torus} in \Cref{fig:gluing-square:torus}. In the picture, we now indicate the pairs of edges: the gluing $e^1_1 \sim -e^1_3$ is indicated by a single white triangle on the edges, and the gluing $e^1_2 \sim -e^1_4$ has two triangles on each glued edge. Observe that all vertices have been identified. The face is bounded by the circle $\lParen e^1_1, e^1_2, -e^1_1, -e^1_2 \rParen$.
\end{example}

\begin{example}[Klein bottle, real projective plane]
  \label{ex:gs-Klein-RP2}
  Finally, we show what happens if we glue pairs of opposite edges as in \Cref{ex:gs-torus}, but we change the orientation in one or both of the gluings. Putting $e^1_1 \sim -e^1_3$ and $e^1_2 \sim e^1_4$ gives the\emph{ Klein bottle} in \Cref{fig:gluing-square:klein}. Putting $e^1_1 \sim e^1_3$ and $e^1_2 \sim e^1_4$ gives the \emph{real projective plane} in \Cref{fig:gluing-square:proj}.
\end{example}

Note the omission of a \emph{sphere} $\mathcal{S}^2$. This does not fit our definition of a cell complex very neatly, but we also do not need to use any spheres, so we do not mind this.

We can obtain those same spaces using different abstract cell complexes. One particular method useful for us later is to work backwards, define a space, and then \emph{cellulate} it, i.e. construct an abstract cell complex that describes it. We use, and describe this procedure, in \Cref{sec:cellulation}.

\subsection{Gluing two complexes}

Now, we shortly describe how to glue two complexes together. Essentially, we first disjointly merge them to a single complex, and then we quotient.
\begin{definition}[direct sum of complexes]
  \label{def:direct-sum-of-cell-complexes}
  Let $X = (X_\bullet, \varphi_\bullet)$ and $Y = (Y_\bullet, \psi_\bullet)$ be abstract cell complexes. Define their direct sum, denoted $X \oplus Y$, as a complex $U = (U_\bullet, \upsilon_\bullet)$ where $U_n \deq X_n \sqcup Y_n$ and $\upsilon_n \deq \varphi_n \sqcup \psi_n$, for $n\in\{0,1,2\}$. These disjoint unions are set-theoretic coproducts, and this is a coproduct of cell complexes.
\end{definition}

Then in order to merge two complexes $X$ and $Y$, we first put them together in the coproduct $X \oplus Y$, and then we identify whatever cells we want to identify. This is a \emph{pushout} of abstract cell complexes, but we omit the details of this, because we have not defined morphisms for abstract cell complexes.\footnote{We do not use morphisms directly, which is why we do not define them. They are also tricky to get right, and several choices are available. To give an idea, the easiest to define are dimension-preserving maps, which are required to send $n$-cell to $n-$cells. These must of course respect the boundaries and the action of $e^{i\pi\Z}$. In fact, the action of $-1 \in e^{i\pi\Z}$ is such a morphism.}

Even though the direct sum of cell complexes is defined using disjoint unions of sets ($\sqcup$), we denote it similarly to a direct sum $\oplus$. This is because we will later see that the chain complex corresponding to a direct sum of cell complexes is a direct sum of chain complexes (see \Cref{lem:direct-sum-of-complexes}).

\subsection{Relationship to general cell complexes}

The construction we defined is designed to give very precise instructions on how to construct a topological (non-abstract) cell complex from an abstract cell complex $X = (X_\bullet, \varphi_\bullet)$. This follows by mapping the $n$-cells to actual topological objects, namely $n$-balls $\mathcal{B}^n$. These are glued by their boundaries $\partial\mathcal{B}^n = \mathcal{S}^{n-1}$, as prescribed by the abstract cell complex. The abstract attaching maps $\varphi_n$ ($n=1,2$) are realized as a family of topological attaching maps $\overline\varphi^n_\alpha : \mathcal{S}^{n-1} \to \overline X^{n-1}$ for each $n$-cell $e^n_\alpha \in X_n$, where $\overline X^{n-1}$ is the $(n-1)$-skeleton of the topological realization. As such, our definition is a useful abstraction to faithfully reason about cell complexes in the setting we have here. It further allows us to reach cellular homology without too much complication.

\subsection{Other kinds of (abstract) topological complexes}
\label{sec:other-kinds-abstract-top-complexes}

We now comment on other kinds of topological complexes that we chose \emph{not} to use. We only briefly and informally describe them.

The simplest kind are \emph{simplicial complexes}.\cite[\S2.1]{HatcherAllen2002At} These are made of triangles, tetrahedra, and their higher-dimensional analogues; though we only care about surfaces, so triangles are where we end. Simplicial complexes are very easy to work with, however, the definition is quite restrictive in what kind of cells are allowed. For example, to construct a space like the torus (see \Cref{ex:gs-torus}), we need at least $19$ triangles, and similarly many edges and vertices. For our purposes, we want control over how many cells are present, and simplicial complexes cannot give us that.

A slight generalization are \emph{$\Delta$-complexes}.\cite[\S2.1]{HatcherAllen2002At} They look like simplicial complexes, but some conditions are relaxed. Describing a torus, say, no longer requires a large amount of cells. However, all the faces still have to be triangles, and this is a deal-breaker for us.

Another generalization are \emph{polyhedral complexes}.\cite[\S2.3]{MR3287221} These allow any polygons to be the faces, which would work more nicely for us. However, they still exclude certain things that we like to have, like faces bounded by self-loops. Also, the way they are defined in ref. \cite{MR3287221}, they are already concrete topological spaces, and we want abstractions.

Thus the choice is to use an extremely general kind of complex, the \emph{cell complex}. Their definition in ref. \cite{HatcherAllen2002At}, while stated in terms of concrete spaces, like $n$-balls and similar, is very amenable to abstraction, and this is what we have done.

Lastly, we mention another kind of \emph{abstract cell complex}, unrelated to our construction. These are special kind of topological spaces which only hold information about dimension of cells and their incidence. They seemed promising for a while, but they were in the end not quite what we needed, so we chose not to use them. They are part of a current line of research in \emph{digital topology}, mainly used for image processing.\cite{Kovalevsky2021}

\section{Chain Complexes from Abstract Cell Complexes}
\label{sec:chain-complexes-from-cell-complexes}

When constructing a cell complex, we build it up by saying how higher-dimensional cells attach, by their boundaries, to lower-dimensional ones. Now, we break the space back into its parts to study it algebraically. We define chain complexes (see \Cref{def:chain-complex}) that correspond to cell complexes.

In the following definitions, let $X = (X_\bullet, \varphi_\bullet)$ be an abstract $2$-dimensional cell complex. We define a \emph{cellular chain complex} for this, in analogy to a (non-abstract, and not bounded to dimension $2$) cellular chain complex from ref.~\cite{HatcherAllen2002At}.

\begin{definition}[cellular chain modules]
  \label{def:cellular-chain-modules}
  Define a family of free $\Z$-modules $\{ C_n \}_{n \in \N}$, called \emph{chain modules}, as follows: There are no cells of dimension larger than $2$, so we define $C_n \deq \bm 0$ for $n > 2$. For $n \in \{ 0, 1, 2 \}$, we define the chain module $C_n$ to be a module freely generated by the $n$-cells of the cell complex:
  \[ C_n \deq \genR{\Z}{X_n}^\oplus = \bigoplus_{e^n_\alpha \in X_n} \Z e^n_\alpha, \]
  that is the elements of $C_n$ are $\Z$-linear combinations of $n$-cells, and they are called \emph{$n$-chains}. We emphasize that $n$-cells are linearly independent, and thus a basis. Furthermore, since all sets of cells are finite by definition, all chain modules are \emph{finitely generated}.
  Recall from \Cref{def:chain-complex} that by convention, we define also $C_{-1} \deq \bm 0$.
\end{definition}

\begin{definition}[cellular boundary maps]
  \label{def:cellular-boundary-maps}
  Define a family of $\Z$-module homomorphisms $\{ \partial_n : C_n \to C_{n-1} \}_{n \in \N}$ between cellular chain modules from \Cref{def:cellular-chain-modules} as follows: For $n>2$, the module $C_n = \bm 0$, and hence the differential from it has to be the zero map: $\partial_n = 0$. Similarly, $C_{-1} = \bm 0$, so $\partial_0 = 0$.

  The remaining two differentials are defined by the abstract attaching maps. Recall that for a $1$-cell $e^1_\alpha \in X_1$ (an edge), the attaching map gives its source and target: $\varphi_1(e^1_\alpha) = (s(e^1_\alpha), t(e^1_\alpha))$. Define \[\partial_1(e^1_\alpha) \deq t(e^1_\alpha) - s(e^1_\alpha).\]
  For a $2$-cell $e^2_\alpha \in X_2$, the boundary is an abstract circle $\varphi_2(e^2_\alpha) = \lParen \pm_i e^1_i \rParen_{i=1}^\ell \in \mathcal{S}\overrightarrow{X_1}$, and we define
  \[
    \partial_2(e^2_\alpha) \deq \sum_{i=1}^\ell \pm_i e^1_i,
  \]
  that is we sum the elements in the tuple representing the circle, keeping their signs as well.
\end{definition}

\begin{theorem}
  The family of cellular chain modules from \Cref{def:cellular-chain-modules} together with the family of cellular boundary maps from \Cref{def:cellular-boundary-maps} form a chain complex, called a \emph{cellular chain complex}. That is, $\partial^2 = 0$.
\end{theorem}
\begin{proof}
  This is clear for those $n \in \N$ where $\partial_n = 0$ or $\partial_{n+1} = 0$, and thus $\partial_n \circ \partial_{n+1} = 0$. That is the case for $n \ge 2$, as well as $n = 0$. We only need to show that $\partial_1 \circ \partial_2 = 0$.

  Take some $e^2_\alpha \in X_2$, and write its boundary circle $\varphi_2(e^2_\alpha) = \lParen \pm_1 e^1_1, \dots, \pm_\ell e^1_\ell \rParen \in \mathcal{S}\overrightarrow{X_1}$. We compute the action of the composite $\partial_1\circ\partial_2$ on this:
  \[
    \partial_1 \circ \partial_2(e^2_\alpha)
    = \partial_1 \left( \sum_{i=1}^\ell \pm_i e^1_i \right)
    = \sum_{i=1}^\ell \pm_i \partial_1(e^1_i)
    = \sum_{i=1}^\ell \left( \pm_i t(e^1_i) \mp_i s(e^1_i) \right).
  \]
  Remember that the boundary of a $1$-cell is target minus source. Recall that $s(\pm e^1_i) = t(\mp e^1_i)$, i.e. the source of an edge when traversed in one direction is the target of the edge when traversed in the opposite direction. Rewrite the above as:
  \begin{align*}
    \sum_{i=1}^\ell
    & \left( \pm_i t(e^1_i) \mp_i s(e^1_i) \right)
      = \sum_{i=1}^\ell \left( t(\pm_i e^1_i) - s(\pm_i e^1_i) \right) \\
    & = t(\pm_1 e^1_1) - s(\pm_1 e^1_1) + t(\pm_2 e^1_2) - s(\pm_2 e^1_2)
      + \cdots + t(\pm_\ell e^1_\ell) - s(\pm_\ell e^1_\ell) \\
    & =
      \underbrace{\left[ t(\pm_1 e^1_1) - s(\pm_\ell e^1_\ell) \right]}_0
      + \underbrace{\left[ t(\pm_2 e^1_2) - s(\pm_1 e^1_1) \right]}_0
      + \cdots
      + \underbrace{\left[ t(\pm_\ell e^1_\ell) - s(\pm_{\ell-1} e^1_{\ell-1}) \right]}_0 \\
    & = 0.
  \end{align*}
  In the last step, we used the fact that $t(\pm_i e^1_i) = s(\pm_{i + 1} e^1_{i + 1})$ for each $i=1,\dots, \ell-1$ and $t(\pm_\ell e^1_\ell) = s(\pm_1 e^1_1)$, which comes from this being a circle. Therefore $\partial_1 \circ \partial_2 = 0$, and since this was the only nontrivial case, we conclude that $\partial^2 = 0$, and the family of cellular chain modules together with cellular boundary maps forms a chain complex.
\end{proof}

\begin{example}
  \label{ex:interpretation-of-sum-of-edges}
  Now for some interpretation. The abstract cell complex $(X_\bullet, \varphi_\bullet)$ only contains individual cells, whereas the chain modules contain linear combinations. Suppose we have two edges $e^1_1$ and $e^1_2$, such that $t(e^1_1) = s(e^1_2)$. Consider $e^1_1 + e^1_2$: its boundary is $\partial_1(e^1_1 + e^1_2) = t(e^1_2) - s(e^1_1)$. The sum $e^1_1 + e^1_2$ behaves as a single longer edge with endpoints $(s(e^1_1), t(e^1_2))$. Not all linear combinations have such a nice interpretation, for example
  \[ \partial_1(e^1_1 - e^1_2) = t(e^1_1) - s(e^1_1) - t(e^1_2) + s(e^1_2) = 2t(e^1_1) - s(e^1_1) - t(e^1_2), \]
  where the final step comes from the fact that $t(e^1_1) = s(e^1_2)$. This correctly shows the fact that the endpoint $t(e^1_1)$ appears twice as the target of an edge, but the sum itself does not have a nice interpretation as a longer edge. That is because we did not sum the edges \emph{coherently}: $(e^1_1, -e^1_2)$ is not a path in the graph $X^1$.
\end{example}

\diagramonright{0.7}{0.3}{%
  \vspace{1.5ex}
  \begin{example}
    \label{ex:cancellation-of-edge-between-faces}
    This works similarly for faces. Return to \Cref{ex:acc2-square}; we recall \Cref{fig:square-abs-cell-complex} on the right. Consider $e^2_1 + e^2_2$: the boundary is
    \[
      \partial_2(e^2_1 + e^2_2) = \underbrace{e^1_1 + \cancel{e^1_2} + e^1_3}_{\partial_2(e^2_1)} + \underbrace{e^1_4 + e^1_5 - \cancel{e^1_2}}_{\partial_2(e^2_2)} = e^1_1 + e^1_4 + e^1_5 + e^1_3,
    \]
    that is the boundary of the whole square, without the\hiddenparbreak
    \vspace{0.7ex}
  \end{example}
}{%
  \centering
  \begin{tikzpicture}
  \tikzstyle{vertex}=[fill=black,draw=black,circle,minimum size=4pt,inner sep=0]
  \tikzstyle{->-}=[
  decoration={
    markings,
    mark=at position 0.55 with {\arrow{Triangle[scale=1]}}},
  postaction={decorate},
  draw]
  \node[vertex, label={-135:$e^0_1$}] at (0,0) (v1) {};
  \node[vertex, label={135:$e^0_2$}] at (0,3) (v2) {};
  \node[vertex, label={45:$e^0_3$}] at (3,3) (v3) {};
  \node[vertex, label={-45:$e^0_4$}] at (3,0) (v4) {};
  \draw[->-] (v1.center) -- node[midway, left] {$e^1_1$} (v2.center);
  \draw[->-] (v2.center) -- (v4.center);
  \node at (1.4,2.2) {$e^1_2$};
  \draw[->-] (v4.center) -- node[midway, below] {$e^1_3$} (v1.center);
  \draw[->-] (v2.center) -- node[midway, above] {$e^1_4$} (v3.center);
  \draw[->-] (v3.center) -- node[midway, right] {$e^1_5$} (v4.center);
  \node at (1,0.7) (f1) {$e^2_1$};
  \node[scale=3] at (f1.center) {$\circlearrowright$};
  \node at (2.3,2) (f2) {$e^2_2$};
  \node[scale=3] at (f2.center) {$\circlearrowright$};
  \begin{pgfonlayer}{background}
    \path[fill=\colora, fill opacity=0.5] (0,3) -- (0,0) -- (3,0);
    \path[fill=\colorb, fill opacity=0.5] (0,3) -- (3,3) -- (3,0);
  \end{pgfonlayer}
\end{tikzpicture}
}
diagonal edge $e^1_2$. This has a nice interpretation that the sum $e^2_1 + e^2_2$ corresponds to the whole square. It works because we sum faces of the same orientation. Obviously, we can again take any linear combination of $2$-cells, but this will not have the same nice interpretation.

\begin{example}
  \label{ex:interpretation-sum-of-faces-in-cylinder}
  In \Cref{ex:cancellation-of-edge-between-faces}, $e^1_2$ canceled in $\partial_2(e^2_1 + e^2_2)$ because both faces include it in their boundary circles, but in different directions. The same kind of cancellation may happen when the same $2$-cell, or linear combination of $2$-cells, touches and edge multiple times. Consider the cylinder obtained by gluing opposite edges of a the square, see \cref{fig:cylinder-complex:original}. The attaching map of $2$-cells is:
  \begin{align*}
    \varphi_2(e^2_1) & = \lParen e^1_1, e^1_2, e^1_3 \rParen
    & \implies &
    & \partial_2(e^2_1) & = e^1_1 + e^1_2 + e^1_3, \\
    \varphi_2(e^2_2) & = \lParen e^1_4, -e^1_1, -e^1_2 \rParen
    & \implies &
    & \partial_2(e^2_1) & = e^1_4 -e^1_1 -e^1_2.
  \end{align*}
  Then the boundary of the sum is
  \[
    \partial_2(e^2_1 + e^2_2)
    = \cancel{e^1_1} + \bcancel{e^1_2} + e^1_3 + e^1_4 - \cancel{e^1_1} - \bcancel{e^1_2}
    = e^1_3 + e^1_4.
  \]
  This now corresponds to the whole square, even though such a cell does not exist in~$X_2$. If it did, the boundary circle would be $\lParen e^1_1, e^1_4, -e^1_1, e^1_3 \rParen$; see \cref{fig:cylinder-complex:sum}. The edge $e^1_2$ is not present in the abstract boundary circle, nor in $\partial_2(e^2_1 + e^2_2)$. The edge $e^1_1$ is also canceled in $\partial_2(e^2_1 + e^2_2)$, in the same way that $e^1_2$, but we keep it in the abstract complex description, because this is what the definition requires. However, as far as the chain complex can tell, the boundary consists of $e^1_3 + e^1_4$, in this case two disjoint loops. This does correspond to a cylinder without the top and bottom face.
\end{example}

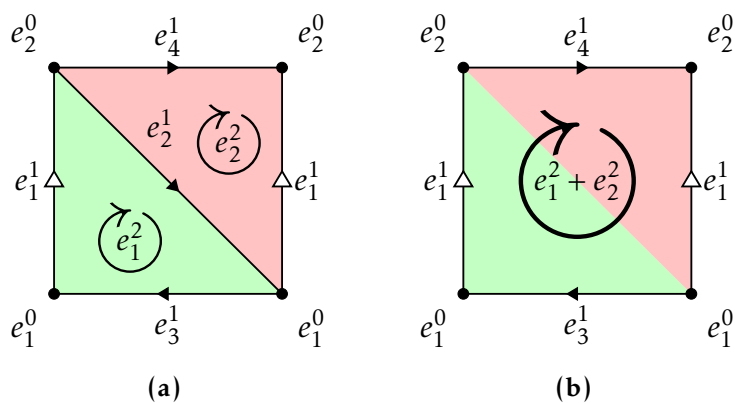
\begin{figure}[h]
  \centering
  \subfloat[\label{fig:cylinder-complex:original}]{\begin{tikzpicture}
  \tikzstyle{vertex}=[fill=black,draw=black,circle,minimum size=4pt,inner sep=0]
  \node[vertex, label={-135:$e^0_1$}] at (0,0) (v1) {};
  \node[vertex, label={135:$e^0_2$}] at (0,3) (v2) {};
  \node[vertex, label={45:$e^0_2$}] at (3,3) (v3) {};
  \node[vertex, label={-45:$e^0_1$}] at (3,0) (v4) {};
  \draw[-g>-] (v1.center) -- node[midway, left] {$e^1_1$} (v2.center);
  \draw[->-] (v2.center) -- (v4.center);
  \node at (1.4,2.2) {$e^1_2$};
  \draw[->-] (v4.center) -- node[midway, below] {$e^1_3$} (v1.center);
  \draw[->-] (v2.center) -- node[midway, above] {$e^1_4$} (v3.center);
  \draw[-g>-] (v4.center) -- node[midway, right] {$e^1_1$} (v3.center);
  \node at (1,0.7) (f1) {$e^2_1$};
  \node[scale=3] at (f1.center) {$\circlearrowright$};
  \node at (2.3,2) (f2) {$e^2_2$};
  \node[scale=3] at (f2.center) {$\circlearrowright$};
  \begin{pgfonlayer}{background}
    \path[fill=\colora, fill opacity=0.5] (0,3) -- (0,0) -- (3,0);
    \path[fill=\colorb, fill opacity=0.5] (0,3) -- (3,3) -- (3,0);
  \end{pgfonlayer}
\end{tikzpicture}}
  \qquad
  \subfloat[\label{fig:cylinder-complex:sum}]{\begin{tikzpicture}
  \tikzstyle{vertex}=[fill=black,draw=black,circle,minimum size=4pt,inner sep=0]
  \node[vertex, label={-135:$e^0_1$}] at (0,0) (v1) {};
  \node[vertex, label={135:$e^0_2$}] at (0,3) (v2) {};
  \node[vertex, label={45:$e^0_2$}] at (3,3) (v3) {};
  \node[vertex, label={-45:$e^0_1$}] at (3,0) (v4) {};
  \draw[-g>-] (v1.center) -- node[midway, left] {$e^1_1$} (v2.center);
  \draw[->-] (v4.center) -- node[midway, below] {$e^1_3$} (v1.center);
  \draw[->-] (v2.center) -- node[midway, above] {$e^1_4$} (v3.center);
  \draw[-g>-] (v4.center) -- node[midway, right] {$e^1_1$} (v3.center);
  \node at (1.5,1.5) (f1) {$e^2_1 + e^2_2$};
  \node[scale=5.5] at (f1.center) {$\circlearrowright$};
  \begin{pgfonlayer}{background}
    \path[fill=\colora, fill opacity=0.5] (0,3) -- (0,0) -- (3,0);
    \path[fill=\colorb, fill opacity=0.5] (0,3) -- (3,3) -- (3,0);
  \end{pgfonlayer}
\end{tikzpicture}}
  \caption[Cylinder complex]{Cylinder complex. \textbf{(a)}~The actual complex $(X_\bullet, \varphi_\bullet)$. We glued $e^1_1 \sim -e^1_5$ from \protect\cref{fig:square-abs-cell-complex}. If one were an ant living on this surface, one could start in the middle, walk right, and loop back from the left. \textbf{(b)}~The same space, but we show that $e^2_1 + e^2_2$ corresponds to a single cell that consists of the entire square. Such a cell is not actually contained in $X_2$, but a corresponding element is contained in the chain module $C_2$.}
  \label{fig:cylinder-complex}
\end{figure}

Having defined the abstract cell complexes and their chain complexes in such a closely-corresponding way allows us to use either structure to reason about a space. We use whichever is more convenient: the chain complex is incredibly useful in capturing the actual structure of the space independent of a particular abstract cell complex chosen to represent it, but we write the cell complex explicitly when we need the details about individual cells and their attaching. One of the things we obtain from chain complexes is homology, and this is important in the view of the topological space.

\section{Homology of a cell complex}
\label{sec:homology-of-cell-complex}

In \Cref{def:homology}, we define the $n^{\text{th}}$ homology module of a chain complex $C_\bullet$ as $H_n(C_\bullet) = \ker \partial_n / \im \partial_{n+1}$. We mentioned in \Cref{nota:cycles-and-boundaries} that we call $\ker \partial_n$ the space of \emph{$n$-cycles}, denoted $Z_n(C_\bullet)$, and we call $\im \partial_{n+1}$ the space of \emph{$n$-boundaries}, denoted $B_n(C_\bullet)$. There terms originate in studying a topological space using a chain complex, and we explain them here, also interpreting what homology is in this context.
While this obviously applies to general cell complexes (i.e. non-abstract and not restricted to dimension two), we choose to use our notion of an abstract $2$-dimensional cell complex. In part, this is because imagining surfaces, perhaps with edges and points drawn on them, is easy. The abstract $2$-dimensional cell complexes are also the setting for the rest of the work, so it is useful to build a good understanding of these.

In the following, let $X = (X_\bullet, \varphi_\bullet)$ be an abstract $2$-dimensional cell complex, and let $C_\bullet$ be its corresponding chain complex of $\Z$-modules from \Cref{sec:chain-complexes-from-cell-complexes}.

Elements of $Z_n(C_\bullet)$, the \emph{$n$-cycles}, correspond to $n$-dimensional subcomplexes (see \Cref{def:subcomplex}) of $X$ that have no boundary. This is because the differentials are interpreted as boundaries, and an element is in the kernel of a differential if it is mapped to zero. In our setting, $\partial_0 = 0$ and $\partial_n = 0$ for $n \ge 2$, so the only nontrivial case is $n=1$. Recall that $C_1$ is the space of oriented edges and their linear combinations; we saw their interpretation in \Cref{ex:interpretation-of-sum-of-edges}. A boundary of such a thing consists of endpoint vertices, and thus $1$-cycles are linear combinations of edges that have no endpoints. These are exactly the \emph{cycles} in the graph, which is where the term comes from. Note that they are not necessarily \emph{simple}, by which we mean they may visit the same vertex multiple times, follow an edge forward and then immediately backward, etc. Such cycles are linear combinations of simple cycles that do not do these things, and these form a basis of $Z_1(C_\bullet)$.

In our case, orientations are built into the cell complex, so it is more correct to say that elements of $Z_1(C_\bullet)$ correspond to abstract circles (recall \Cref{def:abstract-circle}). Note that self-loops are also included in this, because both their endpoints are the same, and they cancel.

\emph{Boundaries}, that is elements of $B_n(C_\bullet) = \im \partial_{n+1}$, are $n$-dimensional subspaces that form the boundaries of $(n+1)$-dimensional subspaces. That is to say $B_n(C_\bullet)$ is a $\Z$-module generated by boundaries of all $(n+1)$-cells; though note that $\partial_{n+1}$ is not necessarily a monomorphism, i.e. the set $\partial_{n+1}(X_{n+1})$ of images of $(n+1)$-cells is not necessarily linearly independent.

In the case of $n=1$, the elements of $B_1(C_\bullet)$ correspond to the boundary circles of faces. Let $e^2$ be a $2$-cell bounded by an abstract circle, e.g. a triangle $\lParen e^1_1, e^1_2, e^1_3 \rParen$. Then the differential of $e^2$ is, by definition, $\partial_2(e^2) = e^1_1 + e^1_2 + e^1_3$. Observe that such a triangle (or any abstract circle) belongs to $Z_1(C_\bullet)$. This is a defining property of a chain complex, and the topological interpretation is why that was chosen. Any boundary of a face from $X_2$, or sum thereof, itself has no boundary.

But suppose $Z_1(C_\bullet) \ne B_1(C_\bullet)$, so the homology $H_1(C_\bullet) = Z_1(C_\bullet) / B_1(C_\bullet)$ is non-zero. What are its elements?
The non-zero elements of $H_1(C_\bullet)$ correspond to cycles that are not boundaries of any face, and as such they describe \emph{holes} in the space. More correctly, the elements of $H_1(C_\bullet)$ are equivalence classes of these non-boundary cycles, up to addition of boundary cycles. If two cycles are in the same equivalence class, we call them \emph{homologous}. In the next section, we look at some examples.

\subsection{Examples}
\label{sec:homology-examples}

\begin{example}
  \label{ex:torus-nonHzero-cycle}
  \begin{figure}[h]
    \centering
    \subfloat[\label{fig:torus-nonH0-cycle-square}]{%
      \begin{tikzpicture}
  \newcommand*{\sqside}{4}

  \node[vertex, label={-135:$e^0_1$}] at (0,0) (v1) {};
  \node[vertex, label={135:$e^0_1$}] at (0,\sqside) (v2) {};
  \node[vertex, label={45:$e^0_1$}] at (\sqside,\sqside) (v3) {};
  \node[vertex, label={-45:$e^0_1$}] at (\sqside,0) (v4) {};
  \node[vertex, label={180:$e^0_2$}] at (0,\sqside/2) (v5) {};
  \node[vertex, label={0:$e^0_2$}] at (\sqside,\sqside/2) (v6) {};
  \node[vertex, label={90:$e^0_3$}] at (\sqside/2,\sqside/2) (v7) {};
  \draw[-g>-] (v1.center) -- node[midway, left] {$e^1_1$} (v5.center);
  \draw[-g>>-] (v5.center) -- node[midway, left] {$e^1_2$} (v2.center);
  \draw[-g>>>-] (v2.center) -- node[midway, above] {$e^1_3$} (v3.center);
  \draw[-g>-] (v4.center) -- node[midway, right] {$e^1_1$} (v6.center);
  \draw[-g>>-] (v6.center) -- node[midway, right] {$e^1_2$} (v3.center);
  \draw[-g>>>-] (v1.center) -- node[midway, below] {$e^1_3$} (v4.center);
  \draw[walk] (v5.center) -- node[midway, above] {$e^1_4$} (v7.center);
  \draw[walk] (v7.center) -- node[midway, above] {$e^1_5$} (v6.center);
  \node at (\sqside/2,0.8) (f1) {$e^2_1$};
  \node[scale=3] at (f1.center) {$\circlearrowright$};
  \node at (\sqside/2,\sqside-0.8) (f2) {$e^2_2$};
  \node[scale=3] at (f2.center) {$\circlearrowright$};

  \begin{pgfonlayer}{background}
    \path[fill=\colora, fill opacity=0.5] (v1.center) -- (v5.center) -- (v6.center) -- (v4.center) -- cycle;
    \path[fill=\colorb, fill opacity=0.5] (v2.center) -- (v3.center) -- (v6.center) -- (v5.center) -- cycle;
  \end{pgfonlayer}
\end{tikzpicture}}
    \qquad
    \subfloat[\label{fig:torus-nonH0-cycle-3d}]{%
      \input{torus_with_nonHzero_cycle_3d_labels}}
    \caption[Torus, which has circles that are not homologous to zero]{Torus, which has circles that are not homologous to zero.
      \textbf{(a)}~The gluing diagram. We graphically emphasize the circle $\lParen e^1_4, e^1_5 \rParen$ as \tikz[baseline=-1mm]{\draw[walk] (0,0) -- ++(0.5,0);}.
      \textbf{(b)}~A 3D picture of the torus glued. The dashed lines correspond to the edges where gluing occurs. The full lines correspond to the circle $\lParen e^1_4, e^1_5 \rParen$. Notice that there is no way to continuously deform this circle along the surface, such that it contracts to a point.}
    \label{fig:torus-nonH0-cycle}
  \end{figure}
  Suppose $X$ is a torus obtained by gluing a square as shown in \cref{fig:torus-nonH0-cycle-square}. Then the boundaries of faces are:
  \begin{align*}
    \varphi_2(e^2_1) & = \lParen e^1_1, e^1_4, e^1_5, -e^1_1, -e^1_3 \rParen
    & \implies &
    & \partial_2(e^2_1) & =  -e^1_3 + e^1_4 + e^1_5, \\
    \varphi_2(e^2_2) & = \lParen e^1_2, e^1_3, -e^1_2, -e^1_5, -e^1_4 \rParen
    & \implies &
    & \partial_2(e^2_1) & = +e^1_3 -e^1_4 -e^1_5.
  \end{align*}
  The boundary of their sum, respecting orientation, is $\partial_2(e^2_1 + e^2_2) = 0$.  We could of course sum them in a different way:
  \[
    \partial_2(a e^2_1 + be^2_2)
    = a(-e^1_3 + e^1_4 + e^1_5) + b(e^1_3 - e^1_4 - e^1_5)
    = -(a - b) e^1_3 + (a - b) e^1_4 + (a - b) e^1_5,
  \]
  for $a,b \in \Z$. Notice that we are describing the image of $\partial_2$. We will show that an element $e^1_4 + e^1_5$ corresponding to the circle $\lParen e^1_4, e^1_5 \rParen$, emphasized in \cref{fig:torus-nonH0-cycle}, is not in $\im\partial_2$. If it were, then there would exist $c \in \Z$, $c \ne 0$, such that $\partial_2(a e^2_1 + be^2_2) = c(e^1_4 + e^1_5)$. This would imply $a-b = 0$. But then $\partial_2(a e^2_1 + be^2_2) = 0$ and thus $c=0$, which is a contradiction. Conclude that $e^1_4 + e^1_5$ is indeed not in the image, meaning there is no linear combination of faces bounded by it. Hence it is not in the homology class $[0] = B_1(C_\bullet)$.

  In the following, we view the torus in \cref{fig:torus-nonH0-cycle-3d} as a continuous topological space. The homology class $[0]$ contains exactly the circles on the torus which can be continuously deformed to a point -- these are called \emph{contractible}. However, there is no way to contract the circle $\lParen e^1_4, e^1_5 \rParen$, because there is no surface along which we could deform it. It would have to pass through ``empty space,'' but this is not actually a part of our space, so it is not possible. The circle $\lParen e^1_4, e^1_5 \rParen$ is said to go around a \emph{hole} in the space. This is the meaning of $e^1_4 + e^1_5$ not being the boundary of any face (or linear combination of faces).

  In the example, we artificially created an unnecessary circle $\lParen e^1_4, e^1_5 \rParen$ that is not a boundary, so that we could explain this on a circle that is not used to define the gluing of the square (\cref{fig:torus-nonH0-cycle-square}) to make a torus. However, the circle $\lParen e^1_3 \rParen$ is also not homologous to zero, and in fact it is homologous to $\lParen e^1_4, e^1_5 \rParen$: we can add the boundary of $e^2_1$ or $e^2_2$ to transform one into the other; or in terms of deformation, we can slide $\lParen e^1_3 \rParen$ to $\lParen e^1_4, e^1_5 \rParen$.

  Another homology class of non-boundary circles corresponds to the circle $\lParen e^1_1, e^1_2 \rParen$. This is because the torus is a surface with no volume inside, so this circle also cannot be contracted. Thus we conclude that the first homology of this complex has rank two, i.e. it consists of two distinct homology classes:
  \[ H_1(C_\bullet) = \genR{\Z}{ e^1_1 + e^1_2, \ e^1_4 + e^1_5 }^\oplus \cong \Z^{\oplus 2}. \]
\end{example}

\begin{example}
  Now, we add another face to the complex from \Cref{ex:torus-nonHzero-cycle}: the new face $e^2_3$ has boundary $\lParen e^1_4, e^1_5 \rParen$, the circle that previously described a hole. We show this in \Cref{fig:torus-with-a-hole-filled}. With the addition of this face, the circle $\lParen e^1_4, e^1_5 \rParen$ can be contracted, and thus it is homologous to zero. The homology class represented by this circle is now, in fact, the zero class $[0]$. The first homology of this new complex is
  $H_1(C_\bullet) = \langle e^1_1 + e^1_2 \rangle_\Z^\oplus \cong \Z$ which has rank 1.
\end{example}

\begin{figure}[h]
  \centering
  \subfloat[]{%
    \begin{tikzpicture}
  \node[vertex, label={-135:$e^0_1$}] at (0,0) (v1) {};
  \node[vertex, label={135:$e^0_1$}] at (0,4) (v2) {};
  \node[vertex, label={45:$e^0_1$}] at (4,4) (v3) {};
  \node[vertex, label={-45:$e^0_1$}] at (4,0) (v4) {};
  \node[vertex, label={180:$e^0_2$}] at (0,4/2) (v5) {};
  \node[vertex, label={0:$e^0_2$}] at (4,4/2) (v6) {};
  \node[vertex, label={90:$e^0_3$}] at (4/2,4/2) (v7) {};
  \draw[-g>-] (v1.center) -- node[midway, left] {$e^1_1$} (v5.center);
  \draw[-g>>-] (v5.center) -- node[midway, left] {$e^1_2$} (v2.center);
  \draw[-g>>>-] (v2.center) -- node[midway, above] {$e^1_3$} (v3.center);
  \draw[-g>-] (v4.center) -- node[midway, right] {$e^1_1$} (v6.center);
  \draw[-g>>-] (v6.center) -- node[midway, right] {$e^1_2$} (v3.center);
  \draw[-g>>>-] (v1.center) -- node[midway, below] {$e^1_3$} (v4.center);
  \draw[walk] (v5.center) -- node[midway, above] {$e^1_4$} (v7.center)
  -- node[midway, above] {$e^1_5$} (v6.center);
  \node at (4-0.8,0.6) (f1) {$e^2_1$};
  \node[scale=3] at (f1.center) {$\circlearrowright$};
  \node at (0.8,4-0.8) (f2) {$e^2_2$};
  \node[scale=3] at (f2.center) {$\circlearrowright$};

  \draw[crawl] (v6.center)
  .. controls (4, 4/2-1) and (-1.3, 4/2-1.8)
  .. (v5.center);
  \node at (1, 1.425) (f3) {$e^2_3$};
  \node[scale=3] at (f3.center) {$\circlearrowright$};

  \begin{pgfonlayer}{background}
    \path[fill=\colora, fill opacity=0.5] (v1.center) -- (v5.center) -- (v6.center) -- (v4.center) -- cycle;
    \path[fill=\colorb, fill opacity=0.5] (v2.center) -- (v3.center) -- (v6.center) -- (v5.center) -- cycle;
    \path[fill=\colorc, fill opacity=0.5]  (v6.center)
    .. controls (4, 4/2-1) and (-1.3, 4/2-1.8)
    .. (v5.center);
  \end{pgfonlayer}
\end{tikzpicture}
  }
  \qquad
  \subfloat[]{%
    \input{torus_with_added_face_3d_labels}
  }
  \caption[Torus with a hole filled]{Torus with a hole filled.
    \textbf{(a)}~The gluing diagram. The only difference from \Cref{fig:torus-nonH0-cycle-square} is the addition of the face $e^2_3$ (blue), bounded by the circle $\lParen e^1_4, e^1_5 \rParen$ drawn as \tikz[baseline=-1mm]{\draw[walk] (0,0) -- ++(0.5,0);}. The edge drawn as \tikz[baseline=-1mm]{\draw[crawl] (0,0) -- ++(0.5,0);} is not real; we only use it to emphasize the boundary and orientation of the new face, but this is only an artifact of the way we have drawn the picture, with the vertex $e^0_2$ on the left and on the right. The edge \tikz[baseline=-1mm]{\draw[crawl] (0,0) -- ++(0.5,0);} is to be interpreted as the vertex $e^0_2$ stretched between the two places where we actually place the dot \tikz[baseline=-1mm]{\node[vertex] {};}.
    \textbf{(b)}~A 3D picture. To avoid clutter, we omit the labeling, which is the same as in \Cref{fig:torus-nonH0-cycle-3d}, with the addition of the face $e^2_3$ (blue). Notice that the circle $\lParen e^1_4, e^1_5 \rParen$ (full line) can be contracted.}
  \label{fig:torus-with-a-hole-filled}
\end{figure}

\subsection{Homology of a direct sum of complexes}

In \Cref{def:direct-sum-of-cell-complexes}, we defined the direct sum of cell complexes. We now show how homology behaves in respect to this. We need this to study $H_0(C_\bullet)$ in the following section, but also for merging two quantum error correcting codes in \Cref{sec:connected-sum-of-codes}.

\begin{lemma}
  \label{lem:direct-sum-of-complexes}
  Let $Y = (Y_\bullet, \upsilon_\bullet)$ and $Z = (Z_\bullet, \zeta_\bullet)$ be abstract cell complexes, and let $X = (X_\bullet, \chi_\bullet) \deq Y \oplus Z$. Denote their chain complexes $C_\bullet(Y)$, $C_\bullet(Z)$, and $C_\bullet(X)$. Then $C_\bullet(X) = C_\bullet(Y) \oplus C_\bullet(Z)$, the \emph{direct sum chain complex} which has components $C_n(X) = C_n(Y) \oplus C_n(Z)$ and differentials $\partial_n^{C_\bullet(X)} = \partial_n^{C_\bullet(Y)} \oplus \partial_n^{C_\bullet(Z)}$ for all $n$.
\end{lemma}
\begin{proof}
  Recall from  that the sets of cells of $X$ are $X_n = Y_n \sqcup Z_n$ for $n\in\{0,1,2\}$. Then the chain modules are related as
  \begin{equation}
    \label{eq:direct-sum-of-complexes-on-chain-modules}
    C_n(X) = \genR\Z{X_n}^\oplus = \genR\Z{Y_n \sqcup Z_n}^\oplus = \genR\Z{Y_n}^\oplus \oplus \genR\Z{Z_n}^\oplus = C_n(Y) \oplus C_n(Z),
  \end{equation}
  Take a differential $\partial_n^{C_\bullet(X)} : C_n(X) \to C_{n-1}(X)$, where both the domain and codomain are direct sums as in \cref{eq:direct-sum-of-complexes-on-chain-modules}. Then clearly this differential is compatible with the splitting, and
  \[
    \partial_n^{C_\bullet(X)} = \partial_n^{C_\bullet(Y)} \oplus \partial_n^{C_\bullet(Z)}.
    \qedhere
  \]
\end{proof}

\begin{theorem}
  \label{thm:direct-sum-of-complexes-on-homology}
  Let $X = Y \oplus Z$ be cell complexes as above, with their corresponding chain complexes. Then $H_n(C_\bullet(X)) = H_n(C_\bullet(Y)) \oplus H_n(C_\bullet(Z))$ for each $n$.
\end{theorem}
\begin{proof}
  Since the chain modules and differentials are all direct sums, it follows that the kernel and image are also direct sums:
  \begin{align*}
    \ker \partial_n^{C_\bullet(X)}
    & = \ker \partial_n^{C_\bullet(Y)} \oplus \ker \partial_n^{C_\bullet(Z)}
    & \text{and} &
    & \im \partial_n^{C_\bullet(X)}
    & = \im \partial_n^{C_\bullet(Y)} \oplus \im \partial_n^{C_\bullet(Z)},
  \end{align*}
  for all $n$. Then the homology modules are, for each $n$:
  \begin{align}
    \nonumber
    H_n(C_\bullet(X))
    & = \faktor{\ker \partial_n^{C_\bullet(X)}}{\im \partial_{n+1}^{C_\bullet(X)}}
      = \faktor{\left(
      \ker \partial_n^{C_\bullet(Y)} \oplus \ker \partial_n^{C_\bullet(Z)}
      \right)\;}{\;\left(
      \im \partial_n^{C_\bullet(Y)} \oplus \im \partial_n^{C_\bullet(Z)}
      \right)} \\
    \label{eq:H-direct-sum:direct-sum-quotient}
    & \cong \left(
      \faktor{\ker \partial_n^{C_\bullet(Y)}}{\im \partial_n^{C_\bullet(Y)}}
      \right) \oplus \left(
      \faktor{\ker \partial_n^{C_\bullet(Z)}}{\im \partial_n^{C_\bullet(Z)}}
      \right) \\
    \nonumber
    & = H_n(C_\bullet(Y)) \oplus H_n(C_\bullet(Z)).
  \end{align}
  In (\ref{eq:H-direct-sum:direct-sum-quotient}), we use the fact that $\im \partial_n^{C_\bullet(Y)} \le \ker \partial_n^{C_\bullet(Y)}$, and similarly for $C_\bullet(Z)$, to split the quotient into direct sum of two quotients, and thus conclude that the homology respects the direct sum.
\end{proof}

\subsection{$0^{\text{th}}$ homology module}
\label{sec:H0}

We saw in the examples in \Cref{sec:homology-examples} what the first homology module represents. Now, we look at $H_0(C_\bullet)$. The following results are important steps towards the Structure \Cref{thm:structure-theorem-for-qudit-logical-space}. They are standard results in the context of cellular (and other) homology theories. However, we still prove them below, to show how they work, and that they indeed hold for our construction of the abstract cell complex.

\begin{theorem}
  \label{thm:H0-is-free}
  If $C_\bullet$ is a chain complex corresponding to an abstract cell complex $(X_\bullet, \varphi_\bullet)$, then the $0^{\text{th}}$ homology module is free, and its rank is the number of connected components of $(X_\bullet, \varphi_\bullet)$.
\end{theorem}
In order to prove it, we need the following first:
\begin{lemma}
  \label{lem:connected-complex-has-H0-Z}
  If $X = (X_\bullet, \varphi_\bullet)$ is a \emph{connected} abstract cell complex, which means that there exists a path of edges between any two vertices, and $C_\bullet$ is its chain complex, then $H_0(C_\bullet) \cong \Z$.
\end{lemma}
In the statement of the lemma, we do not mention faces. The $1$-skeleton is essentially a graph, and if this is partitioned into disjoint connected components, then there is no way for a face to exist between them, because there is no way to define a boundary circle if there are no edges between the components. This means we can reason about connectedness just in terms of the $1$-skeleton.
\begin{proof}
  Assume $X$ is connected, and take two distinct points $e^0_1$ and $e^0_2 \in X_0$. Then there exists a path of oriented edges $(\pm_i e^1_i \in \overrightarrow{X_1})_{i=1}^\ell$ for $\ell \in \N$, such that $s(\pm_1 e^1_1) = e^0_1$ and $t(\pm_\ell e^1_\ell) = e^0_2$. This being a path means that $t(\pm_i e^1_i) = s(\pm_{i+1} e^1_{i+1})$ for $i = 1, \dots, \ell-1$.

  In the chain complex, the differential $\partial_0 = 0$, so its kernel is $C_0 = \genR\Z{X_0}^\oplus$, where its basis elements are all vertices. We now show that the quotient $C_0 / \im\partial_1 \cong \Z$. The connectedness means that for any pair of vertices $e^0_1$ and $e^0_2$, there exists an element $x \in C_1$ such that $\partial_1(x) = e^0_2 - e^0_1$ -- this $x$ corresponds to the path above. Choose some vertex~$e^0$, and change the basis to
  \[
    C_0 = \genR\Z{e^0} \oplus \genR\Z{ e^0_i - e^0 : e^0_i \in X_0 \setminus \{e^0\}}^\oplus.
  \]
  All of the basis elements in the right hand span are eliminated by quotienting over $\im\partial_1$. We are left with $H_0(C_\bullet) = C_0/\im\partial_1 \cong \langle e^0 \rangle_\Z \cong \Z$.
\end{proof}

\begin{proof}[Proof of \Cref{thm:H0-is-free}]
  Now, we simply put two results together. The connected components of the abstract cell complex $X$ are a family of complexes $\{X^{(i)}\}_i$ such that $X = \bigoplus_i X^{(i)}$. The superscript $(i)$ indexes the family and is not related to skeletons. It follows from \Cref{thm:direct-sum-of-complexes-on-homology} that $H_0(C_\bullet(X)) = \bigoplus_i H_0(C_\bullet(X^{(i)}))$, where $C_\bullet(X)$ and $\{C_\bullet(X^{(i)})\}_i$ are the corresponding chain complexes. By construction, all $X^{(i)}$ are connected, so each $H_0(C_\bullet(X^{(i)})) \cong \Z$ by \Cref{lem:connected-complex-has-H0-Z}. Then we conclude that $H_0(C_\bullet(X)) \cong \Z^{\oplus c}$, where $c$ is the number of connected components.
\end{proof}

\Cref{thm:H0-is-free} is important in its own right. However, the main reason we needed to show it was the following important result about the first homology taken over a ring $R$. The result below is a central step of our main \Cref{thm:structure-theorem-for-qudit-logical-space}.
\begin{corollary}
  \label{thm:ACC-H1-change-of-ring}
  Let $X$ be an abstract cell complex, and $C_\bullet$ its corresponding chain complex of $\Z$-modules. Let $R$ be a ring. Then the first homology of $C_\bullet$ over the ring~$R$ is $H_1^R(C_\bullet) \cong H_1(C_\bullet) \otimes_\Z R$, where $H_1(C_\bullet)$ is the usual homology over $\Z$.
\end{corollary}
\begin{proof}
  Recall from \Cref{lem:H1-change-of-ring} that if $H_n(C_\bullet)$ is free, then there is an isomorphism $H_{n+1}^R(C_\bullet) \cong H_{n+1}(C_\bullet) \otimes_\Z R$. As we have just shown in \Cref{thm:H0-is-free}, $H_0(C_\bullet)$ is free. The result follows immediately from this.
\end{proof}

\part[Homological Quantum Error Correction]{Homological Quantum Error Correction
  \\[4cm] \begin{RPhalfplane}{3}{0.2}{120}{30}
  \coordinate (v1) at (180:\crad/2);
  \coordinate (v2) at (90:\crad/2);
  \coordinate (v3) at (0:\crad/2);
  \coordinate (v4) at (-90:\crad/2);
  \node[vertex] at (v1) {};
  \node[vertex] at (v2) {};
  \node[vertex] at (v3) {};
  \node[vertex] at (v4) {};

  \newcommand*{\outangle}{15}
  \newcommand*{\outdist}{\crad+\arrowoffset}
  \coordinate (x1) at (180-\outangle:\outdist);
  \coordinate (x2) at (180+\outangle:\outdist);
  \coordinate (x3) at (\outangle:\outdist);
  \coordinate (x4) at (-\outangle:\outdist);

  \coordinate (y1) at (90+1.3*\outangle:\outdist);
  \coordinate (y2) at (-90+1.3*\outangle:\outdist);

  \draw[->-] (x1) to[out=-15, in=120] node[midway, above] {} (v1);
  \draw[->-] (v1) to[out=-120, in=15] node[midway, below] {} (x2);

  \draw[->-] (x3) to[out=180+15, in=60] node[midway, above] {} (v3);
  \draw[->-] (v3) to[out=-60, in=180-15] node[midway, below] {} (x4);

  \draw[->-] (v2) to[out=90, in=-90+1.3*\outangle] node[midway, below] {} (y1);

  \draw[->-] (y2) to[out=90+1.3*\outangle, in=-90] node[midway, above] {} (v4);

  \draw[->-] (v1) -- (v2);
  \draw[->-] (v2) -- (v3);
  \draw[->-] (v3) -- (v4);
  \draw[->-] (v4) -- (v1);

\end{RPhalfplane}
}
\label{part:homol-quant-error}
\chapter[CSS codes from cellulations]{CSS codes from cellulations \hspace*{\fill}\begin{tikzpicture}
  \tikzstyle{vertex}=[fill=black,draw=black,circle,minimum size=2pt,inner sep=0]
  \tikzstyle{compedge}=[
  postaction={decorate},
  line width=0.7pt,
  draw]
  \tikzstyle{->-}=[compedge,
  decoration={
    markings,
    mark=at position 0.7 with {\arrow{Triangle[scale=0.5]}}}]
  \tikzstyle{erroredge}=[red, decorate,
  decoration={crosses,
    segment length=2mm,
    raise=0pt,
    shape size=2mm}]
  \tikzstyle{errorvertex}=[fill=red, minimum size=3pt, inner sep=0pt]

  \pgfdeclarelayer{background}
  \pgfdeclarelayer{foreground}
  \pgfsetlayers{background,main,foreground}

  \newcommand*{\side}{3mm}
  \begin{pgfonlayer}{foreground}
    \node[vertex] (v1) at (0,0) {};
    \node[errorvertex] (v2) at (0,-\side) {};
    \node[vertex] (v3) at (\side,0) {};
    \node[vertex] (v4) at (\side,-\side) {};
    \node[errorvertex] (v5) at (2*\side,0) {};
    \node[vertex] (v6) at (2*\side,-\side) {};
  \end{pgfonlayer}
  \draw[->-] (v1.center) -- (v2.center);
  \draw[->-, red] (v3.center) -- (v4.center);
  \draw[->-, red] (v2.center) -- (v4.center);
  \draw[->-] (v1.center) -- (v3.center);
  \draw[->-, red] (v3.center) -- (v5.center);
  \draw[->-] (v4.center) -- (v6.center);
  \draw[->-] (v5.center) -- (v6.center);

  \begin{pgfonlayer}{background}
    \fill[\colora, fill opacity=1] (v1.center) -- (v2.center) -- (v4.center) -- (v3.center) -- cycle;
    \fill[\colorc, fill opacity=0.7] (v3.center) -- (v5.center) -- (v6.center) -- (v4.center) -- cycle;
  \end{pgfonlayer}

\end{tikzpicture}}
\label{chap:CSS-codes-from-cellulations}

We foreshadowed at the end of \Cref{sec:LDPC} that we aim to design LDPC codes by cellulating spaces. We saw in \Cref{ex:classical-code-is-chain-complex,ex:css-code-is-chain-complex} that a classical code, as well as a CSS code, can be interpreted as a chain complex. Then in \Cref{sec:chain-complexes-from-cell-complexes}, we construct chain complexes that represent abstract cell complexes. This now all comes together, and we finally explain the full story. In explaining this way to construct quantum error correcting codes, and the relation between a chain complex and stabilizer measurements, we follow the papers \cite{MR1951039,bravyi2013homological,cowtan2023css,vuillot2023homological}, with our own commentary and examples.

We start with qubit codes, following from \Cref{sec:css-codes}, and then using the abstract machinery built up in \Cref{chap:rings-and-modules,chap:homol-algebra}, we generalize to qudits where torsion appears and gives us more interesting logical spaces. We start by example.

\section{A simple toric code}
\label{sec:toric-code-basic}

\diagramonright{0.62}{0.38}{%
  \noindent
  We use the torus from \Cref{ex:torus-nonHzero-cycle}. We recall \Cref{fig:torus-nonH0-cycle-square} on the right, though we change the notation for convenience as follows:
  \begin{notation}
    From now on, we denote a $0$-cell (vertex) $e^0_\alpha$ as $\cvert_\alpha$, a $1$-cell (edge) $e^1_\alpha$ as $\cedge_\alpha$, and a $2$-cell (face) as $\cface_\alpha$.
  \end{notation}
  As a reminder, the white arrows on edges in the diagram imply gluing. Pairs of edges with the same number of white arrows are glued, and they are\hiddenparbreak
  \vspace{-0.75ex}
}{
  \centering
  \hspace{1ex}
  \begin{tikzpicture}
  \newcommand*{\sqside}{4}

  \node[vertex, label={-135:$\cvert_1$}] at (0,0) (v1) {};
  \node[vertex, label={135:$\cvert_1$}] at (0,\sqside) (v2) {};
  \node[vertex, label={45:$\cvert_1$}] at (\sqside,\sqside) (v3) {};
  \node[vertex, label={-45:$\cvert_1$}] at (\sqside,0) (v4) {};
  \node[vertex, label={180:$\cvert_2$}] at (0,\sqside/2) (v5) {};
  \node[vertex, label={0:$\cvert_2$}] at (\sqside,\sqside/2) (v6) {};
  \node[vertex, label={90:$\cvert_3$}] at (\sqside/2,\sqside/2) (v7) {};
  \draw[-g>-] (v1.center) -- node[midway, left] {$\cedge_1$} (v5.center);
  \draw[-g>>-] (v5.center) -- node[midway, left] {$\cedge_2$} (v2.center);
  \draw[-g>>>-] (v2.center) -- node[midway, above] {$\cedge_3$} (v3.center);
  \draw[-g>-] (v4.center) -- node[midway, right] {$\cedge_1$} (v6.center);
  \draw[-g>>-] (v6.center) -- node[midway, right] {$\cedge_2$} (v3.center);
  \draw[-g>>>-] (v1.center) -- node[midway, below] {$\cedge_3$} (v4.center);
  \draw[walk] (v5.center) -- node[midway, above] {$\cedge_4$} (v7.center);
  \draw[walk] (v7.center) -- node[midway, above] {$\cedge_5$} (v6.center);
  \node at (\sqside/2,0.8) (f1) {$\cface_1$};
  \node[scale=3] at (f1.center) {$\circlearrowright$};
  \node at (\sqside/2,\sqside-0.8) (f2) {$\cface_2$};
  \node[scale=3] at (f2.center) {$\circlearrowright$};

  \begin{pgfonlayer}{background}
    \path[fill=\colora, fill opacity=0.5] (v1.center) -- (v5.center) -- (v6.center) -- (v4.center) -- cycle;
    \path[fill=\colorb, fill opacity=0.5] (v2.center) -- (v3.center) -- (v6.center) -- (v5.center) -- cycle;
  \end{pgfonlayer}
\end{tikzpicture}
}

\noindent
glued in the orientation indicated by the arrows. Formally, this is a quotient by coherent equivalence, see \Cref{sec:gluing-complexes}.
The chain complex $C_\bullet$ that corresponds to this abstract cell complex is the following: the \mbox{$0$-chain} module, i.e. the space of vertices, is $C_0 = \genR\Z{\cvert_1, \cvert_2, \cvert_3}^\oplus$, the $1$-chain module (edges) is $C_1 = \genR\Z{\cedge_1, \dots, \cedge_5}^\oplus$, and finally the $2$-chain module (faces) is $C_2 = \genR\Z{\cface_1, \cface_2}$. These listed $n$-cells form bases of the chain modules, and we order the basis elements by their index (the subscript); this is the order we wrote them in the preceding expressions.
With this choice of ordering on the basis, we can now write the differentials as matrices:
\begin{align*}
  \partial_2
  & =
    \begin{pmatrix*}[r]
      0 & 0 \\ 0 & 0 \\ -1 & 1 \\ 1 & -1 \\ 1 & -1
    \end{pmatrix*}
  & \text{and} &
  & \partial_1
  & =
    \begin{pmatrix*}[r]
      -1 &  1 & 0 &  0 &  0 \\
      1  & -1 & 0 & -1 &  1 \\
      0  &  0 & 0 &  1 & -1
    \end{pmatrix*}.
\end{align*}
These matrices summarize the differentials computed from the attaching maps in \Cref{ex:torus-nonHzero-cycle}. As mentioned at the end of \Cref{sec:css-code-from-ortho-classical-codes}, we define these to be the parity check matrices corresponding to $X$-type and $Z$-type stabilizers. Specifically, we say that $P_X \deq \partial_1$ and $P_Z \deq \partial_2^\top$.

Recall that each row of a parity check matrix describes a parity check measurement, or in the context of CSS codes, a stabilizer measurement. For example, the first row of $P_X = \partial_1$ corresponds to $X^{-1} \otimes X \otimes \bm 1 \otimes \bm 1 \otimes \bm 1$. Here, the inverse is the adjoint, i.e. $X^{-1} = X^\dagger$. Working with qubit systems, the adjoint is actually $X$ itself, i.e. $X^\dagger = X$. It is useful then to perform a change of ring to the field $\Z_2$. This is done using \Cref{def:extension-of-scalars,def:change-of-ring-for-chain-complex} along the canonical quotient map $\Z \epito \Z_2$. This makes the chain modules into vector spaces over $\Z_2$, and the parity check matrices become the following:
\begin{align*}
  P_Z
  & =
    \begin{pmatrix}
      0 & 0 & 1 & 1 & 1 \\
      0 & 0 & 1 & 1 & 1
    \end{pmatrix},
  & \text{and} &
  & P_X
  & =
    \begin{pmatrix}
      1 & 1 & 0 & 0 & 0 \\
      1 & 1 & 0 & 1 & 1 \\
      0 & 0 & 0 & 1 & 1
    \end{pmatrix}.
\end{align*}
Finally, the stabilizer group defining this CSS code is generated by the following elements:
\[
  \begin{matrix}
    \bm1 & \otimes & \bm1 & \otimes & Z & \otimes & Z & \otimes & Z, \\
    X & \otimes & X & \otimes & \bm1 & \otimes & \bm1 & \otimes & \bm1, \\
    X & \otimes & X & \otimes & \bm1 & \otimes & X & \otimes & X, \\
    \bm1 & \otimes & \bm1 & \otimes & \bm1 & \otimes & X & \otimes & X.
  \end{matrix}
\]
We omit the second occurrence of $\bm1^{\otimes 2} \otimes Z^{\otimes 3}$, because this gives us no new information. The generating set above is not minimal though: $X \otimes X \otimes \bm1 \otimes X \otimes X$ is a product of other two elements. The rows corresponding to those two have lower Hamming weight, so we could choose them to be the generators, and omit the higher-weight element. However, the correspondence to the abstract cell complex would be partially lost, and this may lower our ability to correct errors. In addition, this would mean that attempts at reasoning about the code using the cell complex would fail -- as we will see later, omitting $X \otimes X \otimes \bm1 \otimes X \otimes X$ would lead to an error syndrome that only flags a single vertex, and this makes it unclear what the error is.

The syndrome measurement for this is in \Cref{fig:syndrome-measurement-toric}. The $Z$-type stabilizer is measured by applying a series of controlled $X$ operators, where the data qubit is the control, and the ancilla is the target. To measure $X$-type stabilizers, we essentially change the basis in the measurement of $Z$-type stabilizers. We do not derive it here, but observe in the figure that this means we apply the Hadamard~$H$ before and after the entangling operations (this essentially means preparing $\neket +$ and measuring in the $X$-basis), and the CNOTs are controlled on the ancillae. More on the measurement of stabilizers can be found in ref. \cite[Appendix A]{cowtan2022qudit}

We do not show encoding separately, because this is another large area of quantum error correction using CSS codes, and we will not go into the details on this. However, one way to prepare an encoded state is to perform all the stabilizer measurements and correct errors, as we would usually do during error correction.\cite{lodyga2015}

\begin{figure}[h]
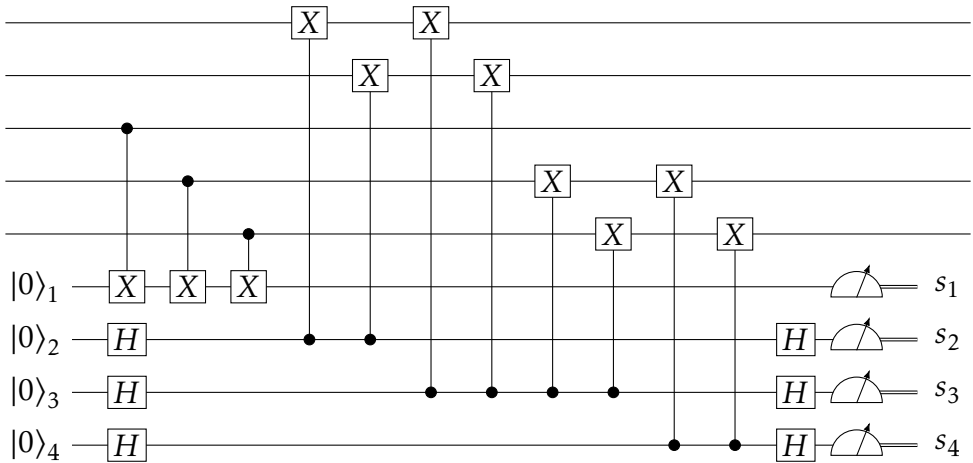

  \centering
  \begin{qcirc}{15}{8}{5}
    \cgate252X
    \cgate353X
    \cgate454X
    \gate  62H
    \cgate605X
    \cgate616X
    \gate  6{13}H
    \gate  72H
    \cgate707X
    \cgate718X
    \cgate739X
    \cgate74{10}X
    \gate  7{13}H
    \gate  82H
    \cgate83{11}X
    \cgate84{12}X
    \gate  8{13}H
  \end{qcirc}
  \caption[Syndrome measurement circuit for the example toric code]{Syndrome measurement circuit for the example toric code.}
  \label{fig:syndrome-measurement-toric}
\end{figure}

\subsection{Logical space}
\label{sec:logical-space}

Recall that codewords, i.e. states considered correct by the code, or undetectable errors, are those corresponding to the kernel of a parity check matrix. We now generalize this idea to our chain complex:
\begin{equation}
  \label{eq:Z-type-complex}
  \begin{tikzcd}[ampersand replacement=\&, row sep=-1]
    \cdots \& {\bm 0} \& {\Z_2^{\oplus m_Z}} \& {\Z_2^{\oplus n}} \& {\Z_2^{\oplus m_X}} \& {\bm 0.}
    \arrow["0", from=1-2, to=1-3]
    \arrow["{P_Z^\top}", from=1-3, to=1-4]
    \arrow["{P_X}", from=1-4, to=1-5]
    \arrow["0", from=1-1, to=1-2]
    \arrow["0", from=1-5, to=1-6]
  \end{tikzcd}
\end{equation}
In the context of CSS codes, we represent codewords by operators from our stabilizer group, and errors are other operators. That is, by (code)words, we do not mean the state vectors, but the operators. In \eqref{eq:Z-type-complex}, we already denote the chain modules as free and finitely generated $\Z_2$-modules with their corresponding ranks. This makes it clear that $C_1 = \Z_2^{\oplus n}$ corresponds to operators on physical qubits, and the other two chain modules somehow correspond to syndromes.

The space $C_1$ corresponds to $Z$-type operators that act on our system. Recall from \Cref{def:binary-rep-of-stabs} that a vector $\vec v \in C_1 = \Z_2^n$ represents $\mathcal{Z}(\vec v) = \bigotimes_{i=1}^n Z^{v_i}$. It is perhaps not intuitive why $C_1$ represents $Z$-type operators, and not $X$-type. The reason is that $\partial_1 = P_X$ describes $X$-type stabilizer measurements, and these commute with any $X$-type errors. However, they anti-commute with detectable $Z$-type errors, represented by elements in $C_1 \setminus \ker P_X$, and this is what we can measure. Thus we call the complex in \eqref{eq:Z-type-complex} a \emph{$Z$-type chain complex}.

\subsubsection*{Kernel and image}

The kernel of $P_X$ corresponds to $Z$-type operators which commute with all \mbox{$X$-type} stabilizers, i.e. undetectable $Z$-type errors (or perhaps logical operations, as we see later). Of course, the matrix $P_Z^\top$ also features in the complex. Its image consists of $Z$-type stabilizers, i.e. products of $Z$-type generators defined by $P_Z$. With this interpretation, the chain complex property that $P_X \circ P_Z^\top = 0$ indeed makes sense: the $Z$-type subgroup of stabilizers by definition commutes with the $X$-type subgroup.

Recall from \Cref{ex:torus-nonHzero-cycle} which is the basis for the present code that in terms of the abstract cell complex, the kernel of $P_X$ contains the linear combinations
\[ \vec v = a( \cedge_1 + \cedge_2 ) + b \cedge_3 + c( \cedge_4 + \cedge_5 )  \]
for $a,b,c \in \Z_2$. These correspond to abstract circles. Such a linear combination represents the $Z$-type operator
\[ \mathcal{Z}(\vec v) = Z^a \otimes Z^a \otimes Z^b \otimes Z^c \otimes \Z^c. \]
Notice this important idea: our operators are $1$-chains in the abstract cell complex! We can thus think of these two very different things interchangeably.

\subsubsection*{Kernel, but not image}

The elements $\ker P_X \setminus \im P_Z^\top$ represent those $Z$-type operators that are not $Z$-type stabilizers, i.e. they can act non-trivially on the state, but they still commute with all $X$-type stabilizers which means they are not (detectable) errors. These are \emph{$Z$-type logical operators}. However, not all of them are, in fact, independent operators. Take a $Z$-type logical operator $A$, and take a $Z$-type stabilizer $B$. Then $A$ and $AB = BA$ are different operators, but their action on the state is the same. Thus we work with \emph{equivalence classes of logical operators}, i.e. of elements of $\ker P_X \setminus \im P_Z^\top$ up to composition with elements of $\im P_Z^\top$.\footnote{We abuse terminology here: we should say: ``up to addition of elements of $\im P_Z^\top$,'' because these are vectors in $\Z_2^n$ representing operators. But the addition in $\Z_2^n$ does, indeed, represent composition of the represented operators: $\mathcal{Z}(\vec v + \vec w) = \mathcal Z(\vec v) \circ \mathcal Z(\vec w)$.} What we have just described is the quotient
\[
  \faktor{\ker P_X}{\im P_Z^\top} = \faktor{\ker \partial_1}{\im \partial_2} = H_1(C_\bullet),\] the first homology of the $Z$-type complex. Note that this is a vector space over $\Z_2$, because we have already changed the ring of scalars of the complex to $\Z_2$. Conclude that the first homology describes the equivalence classes of $Z$-type logical operators, and this is why we call this error correcting code \emph{homological}.

As these homology classes are the only non-trivial logical operators, this also gives us the size of the logical space -- number of logical qubits encoded by this code. This is the rank of $H_1(C_\bullet)$, which for the case of torus is $2$ (see \Cref{ex:torus-nonHzero-cycle}).

The connection with the abstract cell complex and its first homology is clear: logical operators are circles not homologous to zero, i.e. those that go around the holes of the torus. They are equivalent up to addition of circles that are boundaries of faces, which are in $\im P_Z^\top$, and correspond to the $Z$-type stabilizers.

\subsubsection*{Choice of representatives}

The logical space is the space of homology classes. However, to perform actual computation, we need concrete operators. That is, we need to choose a basis of representatives of these equivalence classes. This choice depends on the cell complex used to generate the code, and on what the architecture allows us to do, which is why we only mention this briefly.

One thing we may consider is the weight of the chosen operators, which is the Hamming weight of their representing vectors in $\Z_2^n$. In the ideal case, the weight is large enough that such an operator is unlikely to be applied as a random error (this is related to code distance), yet low enough that the process of applying it does not generate too much noise.\cite{vuillot2023homological,MR3832800}

\subsection{Syndrome measurement}
\label{sec:homological-CSS-error-correction}

We now describe how error detection, and perhaps correction, work in this setting. We still work with the $Z$-type complex.
By construction, $Z$-type errors correspond to the elements of $C_1 \setminus \ker P_X$.
Since the kernel contains $1$-chains without a boundary, this means the errors are exactly those $1$-chains which have a nonzero boundary. Concretely, these correspond to  subcomplexes that are made of edges, and are not closed, meaning they have endpoints. An example of a $Z$-type error is $Z \otimes \bm1^{\otimes 4}$, which corresponds to the $1$-chain $\cedge_1 \in C_1$. We display this in \Cref{fig:toric-example-Z-type-syndrome}.

\begin{figure}[h!]
  \centering
  \newcommand*{\sqside}{4}
  \begin{tikzpicture}
  \pgfdeclarelayer{background}
  \pgfdeclarelayer{foreground}
  \pgfsetlayers{background,main,foreground}

  \begin{pgfonlayer}{foreground}
    \node[errorvertex, label={-135:$\cvert_1$}] at (0,0) (v1) {};
    \node[errorvertexsubtle, label={135:$\cvert_1$}] at (0,\sqside) (v2) {};
    \node[errorvertexsubtle, label={45:$\cvert_1$}] at (\sqside,\sqside) (v3) {};
    \node[errorvertexsubtle, label={-45:$\cvert_1$}] at (\sqside,0) (v4) {};
    \node[errorvertex, label={180:$\cvert_2$}] at (0,\sqside/2) (v5) {};
    \node[errorvertexsubtle, label={0:$\cvert_2$}] at (\sqside,\sqside/2) (v6) {};
    \node[vertex, label={90:$\cvert_3$}] at (\sqside/2,\sqside/2) (v7) {};
  \end{pgfonlayer}

  \draw[-g>-] (v1.center) -- node[midway, left] {$\cedge_1$} (v5.center);
  \draw[erroredge] (v1.center) -- (v5.center);
  \draw[-g>>-] (v5.center) -- node[midway, left] {$\cedge_2$} (v2.center);
  \draw[-g>>>-] (v2.center) -- node[midway, above] {$\cedge_3$} (v3.center);
  \draw[-g>-] (v4.center) -- node[midway, right] {$\cedge_1$} (v6.center);
  \draw[erroredgesubtle] (v4.center) -- (v6.center);
  \draw[-g>>-] (v6.center) -- node[midway, right] {$\cedge_2$} (v3.center);
  \draw[-g>>>-] (v1.center) -- node[midway, below] {$\cedge_3$} (v4.center);
  \draw[->-] (v5.center) -- node[midway, above] {$\cedge_4$} (v7.center);
  \draw[->-] (v7.center) -- node[midway, above] {$\cedge_5$} (v6.center);
  \node at (\sqside/2,0.8) (f1) {$\cface_1$};
  \node[scale=3] at (f1.center) {$\circlearrowright$};
  \node at (\sqside/2,\sqside-0.8) (f2) {$\cface_2$};
  \node[scale=3] at (f2.center) {$\circlearrowright$};

  \begin{pgfonlayer}{background}
    \path[fill=\colora, fill opacity=0.7] (v1.center) -- (v5.center) -- (v6.center) -- (v4.center) -- cycle;
    \path[fill=\colorb, fill opacity=0.7] (v2.center) -- (v3.center) -- (v6.center) -- (v5.center) -- cycle;
  \end{pgfonlayer}
\end{tikzpicture}
  \caption[Example of $Z$-type syndrome in the toric presented code]{Example of $Z$-type syndrome in the presented toric code.
    The abstract cell complex with the error corresponding to $\cedge_1$. We mark $\cedge_1$ by the red crosses~\tikz[baseline=-1mm]{\draw[erroredge](0,0) -- (4mm,0);}. We also emphasize the endpoints $\cvert_1$ and $\cvert_2$ by big red squares \tikz{\node[errorvertex]{};}. Note that these vertices and edges repeat throughout the diagram due to gluing; we mark the repeats more subtly, and we focus attention to one of the occurrences.
  }
  \label{fig:toric-example-Z-type-syndrome}
\end{figure}

The fact that these have endpoints is the essence of obtaining the syndromes. As seen in \Cref{fig:syndrome-measurement-toric}, each stabilizer generator, or equivalently each row of $P_X$, tells us which qubits to (non-destructively) measure. Since by construction $C_1$ corresponds to the physical space, and $C_0$ to syndromes of $X$-type measurements, we associate measurements to vertices of the complex. The syndrome $\vec s \in \Z_2^{m_X}$ has component $s_i = 1$ if there is an error on a qubit touching vertex $\cvert_i$. Note that if there are errors on multiple qubits, they may cancel -- the syndrome measurement counts the number of errors on incident edges, but in the field $\Z_2$, i.e. modulo $2$. This is, again, related to the fact that errors are paths (or tree, graphs, etc.), and we can only detect their endpoints.

To make this concrete, suppose again that we have an error $Z \otimes \bm1^{\otimes 4}$, corresponding to $\cedge_1$. The syndrome is $\vec s = (0,1,1,0)^\top$. The code sees an error, but it is not clear where exactly it occurred:
\begin{itemize}
\item the correct guess is a single error on the first qubit, corresponding to $\cedge_1$,
\item but equally likely, though incorrectly, it could guess that there is a single error on $\cedge_2$;
\item and finally, there could be errors on $\cedge_1$, $\cedge_2$, $\cedge_4$ and $\cedge_5$. This case is less probable, because it requires multiple simultaneous errors.
\end{itemize}

\subsection{Error correction}

Suppose that we guess correctly that the error is either on $\cedge_1$ or $\cedge_2$. If we knew with certainty that it is on $\cedge_1$, we could just apply $Z$ there, and the error would be fixed. However, in a homological code, we can correct an error even if we do not guess correctly where it is: by applying $Z$ operators, we may extend the $1$-chain of errors. The strategy is to create a circle that is homologous to zero, i.e. one that is contractible. Then even if we do not correctly identify the qubit where error occurred, we correct it. For example, we could apply $\bm1 \otimes \bm1 \otimes Z\otimes Z \otimes Z$, that is, apply $Z$ on edges $\cedge_3$, $\cedge_4$ and $\cedge_5$. This is inefficient, but closes the circle, and eliminates the error; see \Cref{fig:toric-example-Z-type-syndrome:altfix}.

However, it is possible to go wrong with this. Not knowing where exactly the error occurred, we might guess it was at $\cedge_2$. Then applying $Z$ on $\cedge_2$ closes the circle $\lParen \cedge_1, \cedge_2 \rParen$, see \Cref{fig:toric-example-Z-type-syndrome:fix}. This circle is not in the trivial homology class. That means the circle is a logical operator: by attempting error correction, we cause a logical operation by mistake.

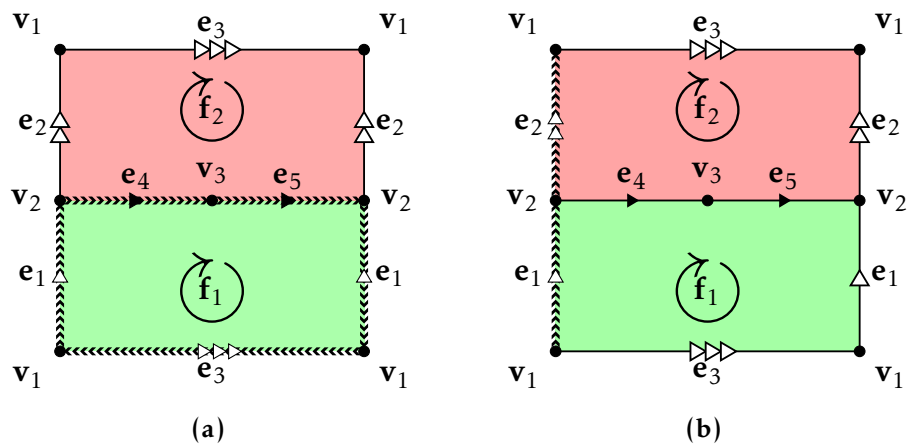
\begin{figure}[h]
  \centering
  \newcommand*{\sqside}{4}
  \subfloat[\label{fig:toric-example-Z-type-syndrome:altfix}]{%
    \begin{tikzpicture}

  \pgfdeclarelayer{background}
  \pgfdeclarelayer{foreground}
  \pgfsetlayers{background,main,foreground}

  \begin{pgfonlayer}{foreground}
    \node[vertex, label={-135:$\cvert_1$}] at (0,0) (v1) {};
    \node[vertex, label={135:$\cvert_1$}] at (0,\sqside) (v2) {};
    \node[vertex, label={45:$\cvert_1$}] at (\sqside,\sqside) (v3) {};
    \node[vertex, label={-45:$\cvert_1$}] at (\sqside,0) (v4) {};
    \node[vertex, label={180:$\cvert_2$}] at (0,\sqside/2) (v5) {};
    \node[vertex, label={0:$\cvert_2$}] at (\sqside,\sqside/2) (v6) {};
    \node[vertex, label={90:$\cvert_3$}] at (\sqside/2,\sqside/2) (v7) {};
  \end{pgfonlayer}

  \draw[walk-g>-] (v1.center) -- node[midway, left] {$\cedge_1$} (v5.center);
  \draw[-g>>-] (v5.center) -- node[midway, left] {$\cedge_2$} (v2.center);
  \draw[-g>>>-] (v2.center) -- node[midway, above] {$\cedge_3$} (v3.center);
  \draw[walk-g<-] (v4.center) -- node[midway, right] {$\cedge_1$} (v6.center);
  \draw[-g>>-] (v6.center) -- node[midway, right] {$\cedge_2$} (v3.center);
  \draw[walk-g<3-] (v1.center) -- node[midway, below] {$\cedge_3$} (v4.center);
  \draw[walk->-] (v5.center) -- node[midway, above] {$\cedge_4$} (v7.center);
  \draw[walk->-] (v7.center) -- node[midway, above] {$\cedge_5$} (v6.center);
  \node at (\sqside/2,0.8) (f1) {$\cface_1$};
  \node[scale=3] at (f1.center) {$\circlearrowright$};
  \node at (\sqside/2,\sqside-0.8) (f2) {$\cface_2$};
  \node[scale=3] at (f2.center) {$\circlearrowright$};

  \begin{pgfonlayer}{background}
    \path[fill=\colora, fill opacity=0.7] (v1.center) -- (v5.center) -- (v6.center) -- (v4.center) -- cycle;
    \path[fill=\colorb, fill opacity=0.7] (v2.center) -- (v3.center) -- (v6.center) -- (v5.center) -- cycle;
  \end{pgfonlayer}
\end{tikzpicture}}
  \qquad
  \subfloat[\label{fig:toric-example-Z-type-syndrome:fix}]{%
    \begin{tikzpicture}

  \pgfdeclarelayer{background}
  \pgfdeclarelayer{foreground}
  \pgfsetlayers{background,main,foreground}

  \begin{pgfonlayer}{foreground}
    \node[vertex, label={-135:$\cvert_1$}] at (0,0) (v1) {};
    \node[vertex, label={135:$\cvert_1$}] at (0,\sqside) (v2) {};
    \node[vertex, label={45:$\cvert_1$}] at (\sqside,\sqside) (v3) {};
    \node[vertex, label={-45:$\cvert_1$}] at (\sqside,0) (v4) {};
    \node[vertex, label={180:$\cvert_2$}] at (0,\sqside/2) (v5) {};
    \node[vertex, label={0:$\cvert_2$}] at (\sqside,\sqside/2) (v6) {};
    \node[vertex, label={90:$\cvert_3$}] at (\sqside/2,\sqside/2) (v7) {};
  \end{pgfonlayer}

  \draw[walk-g>-] (v1.center) -- node[midway, left] {$\cedge_1$} (v5.center);
  \draw[walk-g>>-] (v5.center) -- node[midway, left] {$\cedge_2$} (v2.center);
  \draw[-g>>>-] (v2.center) -- node[midway, above] {$\cedge_3$} (v3.center);
  \draw[-g>-] (v4.center) -- node[midway, right] {$\cedge_1$} (v6.center);
  \draw[-g>>-] (v6.center) -- node[midway, right] {$\cedge_2$} (v3.center);
  \draw[-g>>>-] (v1.center) -- node[midway, below] {$\cedge_3$} (v4.center);
  \draw[->-] (v5.center) -- node[midway, above] {$\cedge_4$} (v7.center);
  \draw[->-] (v7.center) -- node[midway, above] {$\cedge_5$} (v6.center);
  \node at (\sqside/2,0.8) (f1) {$\cface_1$};
  \node[scale=3] at (f1.center) {$\circlearrowright$};
  \node at (\sqside/2,\sqside-0.8) (f2) {$\cface_2$};
  \node[scale=3] at (f2.center) {$\circlearrowright$};

  \begin{pgfonlayer}{background}
    \path[fill=\colora, fill opacity=0.75] (v1.center) -- (v5.center) -- (v6.center) -- (v4.center) -- cycle;
    \path[fill=\colorb, fill opacity=0.75] (v2.center) -- (v3.center) -- (v6.center) -- (v5.center) -- cycle;
  \end{pgfonlayer}
\end{tikzpicture}}
  \caption[Two of the possible choices of correction for the syndrome]{Two of the possible choices of correction the syndrome. We show them as the circles, denoted \tikz[baseline=-1mm]{\draw[walk](0,0) -- (3mm,0);}, obtained by error correction, and we no longer show the error markings.
}
\end{figure}

\subsection{$X$-type complex}

In the previous steps, we obtained $X$-type measurements, $Z$-type errors, and equivalence classes of $Z$-type logical operators. We have shown how the code deals with $Z$-type errors. However, the choice to have $P_X = \partial_1$ and $P_Z^\top = \partial_2$ was arbitrary, and analogous results follow from swapping them; this is how we obtain the corresponding $X$-type complex:\footnote{A note on notation: in both, the chain complex, and this cochain complex, we could just write the three non-trivial vector spaces, the parity matrices, and omit all the zero spaces and morphisms. However, we present it in this form to make the connection with general (co)chain complexes clear.}
\[
  \begin{tikzcd}[ampersand replacement=\&,column sep=scriptsize]
    \cdots \& {\bm 0} \& {C^2} \& {C^1} \& {C^0} \& {\bm 0}
    \arrow["0"', from=1-3, to=1-2]
    \arrow["{P_Z}"', from=1-4, to=1-3]
    \arrow["{P_X^\top}"', from=1-5, to=1-4]
    \arrow["0"', from=1-2, to=1-1]
    \arrow["0"', from=1-6, to=1-5]
  \end{tikzcd}
\]
This is an instance of a \emph{cochain complex}, the dual notion to a chain complex. Specifically, the above cochain complex is the dual of the chain complex in \eqref{eq:Z-type-complex}. We define this formally below, but we limit ourselves to the present case, instead of full generality.
\begin{definition}[dual chain complex and cohomology over $\Z_2$]
  Let $C_\bullet$ be a chain complex of vector spaces over $\Z_2$. Its \emph{dual} is the cochain complex $C^\bullet$, where we have the chain module $C^n \deq C_n^*$ for each $n\in\N$, and where the differentials are $\partial^n : C^{n-1} \to C^{n}$ that are the duals of the chain differentials: $\partial^n \deq \partial_n^*$. We work with finitely generated modules, so these are represented by matrices; in those terms, we have $\partial^n = \partial_n^\top$. This can be interpreted as another chain complex, because $\partial^{n+1} \circ \partial^{n} = 0$. The homology of this new chain complex is called \emph{cohomology}:
  \[
    H^n(C^\bullet)
    \deq \faktor{\ker \partial^{n+1}}{\im \partial^n}
    \overset{\text{(fin. dim.)}}=
    \faktor{\ker \partial_{n+1}^\top}{\im \partial_n^\top}.
    \qedhere
  \]
\end{definition}
We formulate the space of (equivalence classes) of $X$-type logical operators as the first cohomology $H^1(C^\bullet)$ in analogy to deriving the $Z$-type logical operators as the first homology $H_1(C_\bullet)$.

The $X$-type errors, measured by $Z$-type stabilizers, are also $1$-chains in the complex. However, these are not the typical paths, trees, or graphs of edges. The $Z$-type measurements are associated with faces of the complex: we measure the qubits corresponding to the edges in the boundary of a face. Then $X$-type errors correspond to paths of edges as in \Cref{fig:X-type-syndrom}, and syndromes are measured at their endpoint faces.

\begin{figure}[h]
  \centering
  \begin{tikzpicture}
  \foreach \x[evaluate={\v={int(\x+1)}; \w={int(\x+4)}}] in {0,...,3}{
    \coordinate (v-\x-0) at (2*\x,0);
    \coordinate (v-\x-1) at (2*\x,-2);
    \node[vertex, label={90:$\cvert_{\v}$}] at (v-\x-0) {};
    \node[vertex, label={-90:$\cvert_{\w}$}] at (v-\x-1) {};
  }
  \foreach \x in {0,...,3}{
    \draw[->-] (v-\x-0) -- (v-\x-1);
  }
  \foreach \x [evaluate={\n=int(\x+1)}] in {0,...,2}{
    \draw[->-] (v-\x-0) -- (v-\n-0);
    \draw[->-] (v-\x-1) -- (v-\n-1);
  }
  \begin{pgfonlayer}{background}
    \fill[\colora, fill opacity=0.7] (v-0-0) -- (v-1-0) -- (v-1-1) -- (v-0-1) -- cycle;
    \fill[\colord, fill opacity=0.7] (v-1-0) -- (v-2-0) -- (v-2-1) -- (v-1-1) -- cycle;
    \fill[\colorc, fill opacity=0.7] (v-2-0) -- (v-3-0) -- (v-3-1) -- (v-2-1) -- cycle;
    \draw[erroredge] (v-1-0) -- (v-1-1);
    \draw[erroredge] (v-2-0) -- (v-2-1);
  \end{pgfonlayer}

  \foreach \x [evaluate={\f=int(\x+1)}] in {0,...,2}{
    \node (f\f) at (1+2*\x,-0.7) {$\cface_\f$};
    \node[scale=3, rotate=-60] at (f\f.center) {$\circlearrowright$};
  }
  \node[errorcirc] (c1) at (1,-1.5) {};
  \node[errorcirc] (c3) at (5,-1.5) {};
  \begin{pgfonlayer}{background}
    \draw[line width=1pt, dashed] (c1.center) -- (c3.center);
  \end{pgfonlayer}
\end{tikzpicture}
  \caption[$X$-type syndrome in a small lattice code]{$X$-type syndrome in a small lattice code. The errors occur on qubits corresponding to edges marked with crosses \tikz{\draw[erroredge] (0,0) -- (0.3,0);}. The syndrome component corresponding to $\cface_2$ is $0$, because these errors cancel out. The components corresponding to $\cface_1$ and $\cface_3$ are $1$, indicated by the red circles \tikz{\node[errorcirc, scale=0.5]{};}. We draw a dashed line between these, to indicate how the error corresponds to a path across edges. This path is not a cell in the complex.}
  \label{fig:X-type-syndrom}
\end{figure}

An interesting way to look at the cochain complex, when the abstract cell complex describes a $2$-manifold, is that it describes an abstract cell complex with faces and vertices switching places. This idea can also be see in \Cref{fig:X-type-syndrom}. However, this is picture is not the right way to think about the cochain complex in general, and the fact that it is possible here is a special case.

The choice of representatives the cohomology classes, i.e. elements of $H^1(C^\bullet)$, is determined by the choice of representatives of $H_1(C_\bullet)$. Equivalently, choosing concrete $Z$-type logical operators determines the $X$-type logical operators.\cite{cowtan2023css}

\section{Qudits and their operators}
\label{sec:qudits-and-ops}

The aim is now to generalize the preceding to the case of qudits, where the integer $d>1$ is arbitrary. In particular, it is not necessarily a prime, so $\Z_d$ is not necessarily a field, and $\Z_d$-modules are no longer necessarily vector spaces. In particular, this allows for torsion within the first homology, which is where we get extra logical space to play with. Apart from theoretical interest, there is a practical one: quantum computing with qudit systems seems to be a viable avenue. This is in part also because we struggle to keep many qudits stable enough to do computation, but for $d>2$, we can pack more information into a system of fewer objects. Furthermore, there is linear optical quantum computing, and this is inherently qudit based.

In this section, we briefly introduce qudits.
We follow the paper \cite{sarkar2023qudit}.

\begin{definition}[qudit]
  \label{def:qudit}
  Let $d > 1$ be an integer. A \emph{qudit} is a $d$-state system living in a $d$-dimensional state space. The standard basis of this is written $\neket 0$, \dots, $\neket {d-1}$, which is the eigenbasis for the generalized $Z$ operator defined below.
  Their duals are, as expected, $\nebra 0$, \dots, $\nebra {d-1}$, with $\braket ij = \delta_{i,j}$.
\end{definition}

\begin{definition}[qudit Pauli group]
  \label{def:qudit-pauli-group}
  Define the \emph{qudit Pauli $X$} as the operator that maps $\neket{i} \mapsto \neket{i+1 \pmod d}$ for $i =0,\dots,d-1$. Define the \emph{qudit Pauli $Z$} that maps $\neket{i} \mapsto \omega^i \neket{i}$, where $\omega \deq e^{2\pi i / d}$.
  Similarly to \Cref{def:1-qubit-Pauli-group}, we use these to generate the group of qudit Pauli operators $\Pauli_1$, and we define the $n$-qudit Pauli group $\Pauli_n$ as the tensor products of elements of $\Pauli_1$ analogously to \Cref{def:n-qubit-Pauli-group}.
\end{definition}

Note that we do not mention the $Y$ Pauli above. In the qubit case, this is just $iXZ$, and ignoring global phases, it corresponds to $XZ$ which means it is not needed as a generator of the qubit Pauli group. It is similarly defined in the qudit case, and as such we do not need it.
Observe also that if $d=2$, these indeed become the familiar qubit Paulis, as expected.

All of the previous reasoning about stabilizers and similar is transported to the case of qudits. We will build codes where stabilizers may include higher powers of $Z$ or $X$, for example $Z^2 \otimes Z \otimes \bm 1$. Note that it is no longer necessarily the case that $Z^2 = X^2 = \bm 1$.

\begin{lemma}[powers of qudit Pauli operators]
  \label{lem:powers-of-qudit-Paulis}
  Let $m,n \in \Z$. For $d>1$, the power $X^m = X^{m \mod d}$, and similarly $Z^m = Z^{m \mod d}$. Consequently, we have that the product of powers behaves as $X^mX^n = X^{m+n \mod d}$, the inverse is $X^\dagger = X^{-1} = X^{d-1}$, and analogous equations hold for $Z$.
\end{lemma}
This follows directly from \Cref{def:qudit-pauli-group}, so we omit the proof. Note that the above means that $\Z_d$, the ring of modular arithmetic with modulus $d$, is the natural setting for reasoning about qudits.

\begin{lemma}[commutation]
  \label{lem:qudit-pauli-commutation}
  The qudit Pauli $X$ and $Z$ commute as follows: \[ZX = \omega XZ. \qedhere \]
\end{lemma}
\begin{proof}
  Write $X = \sum_{i=0}^{d-1} \neket{i+1 \mod d}\nebra{i}$ and $Z = \sum_{j=0}^{d-1} \omega^j \neket j \nebra j$. This is just a different way to write their definition. Compute $XZ$:
  \[ XZ = \sum_{i=0}^{d-1} \sum_{j=0}^{d-1} \neket{i+1 \mod d} \nebra i \omega^j \neket j \nebra j = \sum_{i=0}^{d-1} \omega^i \neket{i+1 \mod d}\nebra i. \]
  Now, compute $ZX$:
  \[ ZX = \sum_{i=0}^{d-1} \sum_{j=0}^{d-1}  \omega^j \neket j \nebraket j {i+1 \mod d} \nebra i = \sum_{i=0}^{d-1} \omega^{i+1} \neket{i+i \mod d} \nebra i = \omega XZ. \]
  Thus we conclude that $ZX = \omega XZ$.
\end{proof}

\begin{lemma}[commutation of powers of Pauli operators]
  \label{lem:higher-powers-commutation-Pauli}
  Let $m,n \in \N$. Then $Z^m X^n = \omega^{mn} X^n Z^m$.
\end{lemma}
\begin{proof}
  We prove this by strong induction. We start with two base cases: If $m =0$ or $n=0$, then the statement is trivially true. If $m=n=1$, then it follows from \Cref{lem:qudit-pauli-commutation}.

  Now, we follow by induction on $m$. Take an arbitrary $n$, and suppose that for each $m' \le m$, we have that $Z^{m'} X^n = \omega^{m' n} X^n Z^{m'}$. We will show this is then also the case for $m+1$:
  \[ Z^{m+1} X^n = Z Z^m X^n = Z \omega ^{mn} X^n Z^m = \omega^{mn} \omega^{1n} X^n Z Z^m = \omega^{(m+1)n} X^n Z^{m+1}, \]
  where we applied the inductive hypothesis first on $Z^m X^n$ and then $Z X^n$. We follow this by an analogous strong induction on $n$. Take an arbitrary $m$, and suppose that for each $n' \le n$, we have that $Z^m X^{n'} = \omega^{mn'} X^{n'} Z^m$. Then
  \[ Z^m X^{n+1} = Z^m X^n X = \omega^{mn} X^n Z^m X = \omega^{mn} X^n \omega^{m1} X Z^m = \omega^{m(n+1)} X^{n+1} Z^m. \]
  The lemma then follows by strong induction on $m$ and $n$.
\end{proof}

\section{Stabilizers of qudit systems}
\label{sec:stab-qudit-syst}

We generalize \Cref{def:binary-rep-of-stabs} to handle qudits:
\begin{definition}[representation of qudit Pauli operators]
  Let $n \in \N$, and suppose we have a system of $n$ qudits, where $d>1$. We represent the operators acting on this system as vectors in $\Z_d^{\oplus n}$.
  Let $\vec v$ be an element of $\Z_d^{\oplus n}$. Then the corresponding $X$- and $Z$-type stabilizers, where this is now the qudit Paulis, are:
  \begin{align*}
    \mathcal{X}(\vec v) &\deq \bigotimes_{i=1}^n X^{v_i} = X^{v_1} \otimes \cdots \otimes X^{v_n},
    & \text{and} &
    & \mathcal{Z}(\vec v) & \deq \bigotimes_{i=1}^n Z^{v_i} = Z^{v_1} \otimes \cdots \otimes Z^{v_n}.
    \qedhere
  \end{align*}
\end{definition}

Having shown \Cref{lem:higher-powers-commutation-Pauli}, we now generalize \Cref{thm:stabilizers-commute-iff-vectors-orthogonal} from qubits to qudits. This is an important result that will allow us to write a CSS code on qudits again as a chain complex.
\begin{theorem}
  \label{thm:qudit-stab-commute-iff-rep-vectors-ortho}
  Let $n \in \N$ be the number of qudits in a system, and let $\vec z$ and $\vec x$ be vectors $\Z_d^{\oplus n}$. Then the operators $\mathcal{Z}(\vec z)$ and $\mathcal{X}(\vec x)$ commute if and only if $\langle \vec z , \vec x \rangle = 0$, where this is the usual inner product in $\Z_d^{\oplus n}$. The equality is also in $\Z_d$, so it holds modulo $d$.
\end{theorem}
\begin{proof}
  Suppose for now that $\vec z, \vec x \in \Z^{\oplus n}$, i.e. their components are arbitrary integers. This is for the sake of recovering the equality $\langle \vec z , \vec x \rangle \equiv 0 \pmod d$, with the modulo~$d$ explicit. The operators $\mathcal{Z}(\vec z)$ and  $\mathcal{X}(\vec x)$ can be obviously extended to $\Z^{\oplus n}$, and \Cref{lem:higher-powers-commutation-Pauli} clearly also works -- this is due to \Cref{lem:powers-of-qudit-Paulis}.
  We evaluate the commutator, and apply \Cref{lem:higher-powers-commutation-Pauli} on each tensor factor $Z^{z_i} X^{x_i}$:
  \[
    [\mathcal{Z}(\vec z), \mathcal{X}(\vec x)]
    = \bigotimes_{i=1}^n Z^{z_i} X^{x_i} - \bigotimes_{j=1}^n X^{x_j} Z^{z_j}
    = \bigotimes_{i=1}^n \omega^{z_i x_i} X^{x_i} Z^{z_i} - \bigotimes_{j=1}^n X^{x_j} Z^{z_j}.
  \]
  We collect the factors $\omega^{z_i x_i}$ by linearity of the tensor product, and then take the tensor product out of the parenthesis:
  \[
    [\mathcal{Z}(\vec z), \mathcal{X}(\vec x)]
    = \left(\prod_{\ell=1}^n \omega^{z_\ell x_\ell} \right) \bigotimes_{i=1}^n X^{x_i} Z^{z_i} - \bigotimes_{j=1}^n X^{x_j} Z^{z_j}
    = \left(\omega^{\sum_{\ell=1}^n z_\ell x_\ell} - 1 \right) \bigotimes_{i=1}^n X^{x_i} Z^{z_i}.
  \]
  Recall that $\omega \deq e^{2 \pi i/d}$, and for an integer $q \in \Z$, $\omega^q = 1$ if and only if $q \in d\Z$.
  It follows that the commutator is zero if and only if $\sum_{\ell=1}^n z_\ell x_\ell \in d\Z$, or equivalently, $\sum_{\ell=1}^n z_\ell x_\ell \equiv 0 \pmod d$. This is exactly the statement that the operators commute if and only if their representing vectors are orthogonal in $\Z_d^{\oplus n}$.
\end{proof}

Then defining two parity check matrices $P_X$ and $P_Z$ over qudits, in analogy to the qubit case, means that once again we will have $P_X \circ P_Z^\top = 0$, and we can use all the machinery of homological algebra developed so far.

\chapter[Torsion]{Torsion \hspace*{\fill}\begin{tikzpicture}
  \newcommand*{\crad}{0.3}
  \newcommand*{\arrowoffset}{0.05}
  \newcommand*{\arrowstart}{120}
  \newcommand*{\arrowsec}{30}

  \tikzstyle{vertex}=[fill=black,draw=black,circle,minimum size=2pt,inner sep=0]
  \tikzstyle{compedge}=[
  postaction={decorate},
  line width=0.7pt,
  draw]
  \tikzstyle{-g>-}=[compedge,
  decoration={
    markings,
    mark=at position 0.55 with {\arrow{Triangle[scale=0.3, fill=white]}}}]

  \coordinate (o) at (0,0);
  \begin{pgfonlayer}{background}
    \draw[dashed, line width=0.7pt, postaction={decorate}, decoration={
      markings,
      mark=at position 0.35 with {\arrow{Triangle[scale=0.8, fill=white]}},
      mark=at position 0.85 with {\arrow{Triangle[scale=0.8, fill=white]}}
    }] (o) circle (\crad);
  \end{pgfonlayer}

  \coordinate (v1) at (180:2*\crad/3);
  \coordinate (v2) at (o);
  \coordinate (v3) at (0:2*\crad/3);
  \node[vertex] at (v1) {};
  \node[vertex] at (v2) {};
  \node[vertex] at (v3) {};

  \newcommand*{\outangle}{15}
  \coordinate (x1) at (180-\outangle:0.35);
  \coordinate (x2) at (180+\outangle:0.35);
  \coordinate (x3) at (\outangle:0.35);
  \coordinate (x4) at (-\outangle:0.35);

  \draw (x1) -- (v1);
  \draw (v1) -- (x2);

  \draw (x3) -- (v3);
  \draw (v3) -- (x4);

  \draw (v1) .. controls (-3*\crad/8,\crad/4) and (-\crad/8,\crad/4) .. (v2);
  \draw (v2) .. controls (-\crad/8,-\crad/4) and (-3*\crad/8,-\crad/4) .. (v1);

  \draw (v2) .. controls (\crad/8,\crad/4) and (3*\crad/8,\crad/4) .. (v3);
  \draw (v3) .. controls (3*\crad/8,-\crad/4) and (\crad/8,-\crad/4) .. (v2);
\end{tikzpicture}}
\label{chap:torsion}

We build qudit CSS codes in close analogy to the qubit codes. They will be defined using chain complexes of $\Z_d$-modules arising from abstract cell complexes. Note that this implies change of ring from the original chain complex of $\Z$-modules that represents a cell complex, in analogy to changing the ring to $\Z_2$ for qubit codes before.
Error detection and correction are also analogous. The difference is that in the case of qudits, both errors and stabilizers may contain higher powers of qudit Pauli operators.
The thing that is not at all analogous is the presence of torsion (see \Cref{sec:torsion}). Torsion does not exist in vector spaces, so qudits with prime $d$, like qubits with $d=2$, do not have it. In particular, the first homology of the complex which represents the logical space will contain torsion, and this makes the codes interesting for us.

\section{First homology}

We first say something general about this. The main theorem presented here is original, and it transports a result from ref. \cite{vuillot2023homological} about decomposing the logical space to the case of qudits. It is a specific use of the well established Universal Coefficient Theorem~\ref{thm:UCT}.

We wish to work with chain complexes of $\Z_d$-modules. However, as a first step, we come back to the familiar $\Z$-modules, because $\Z$ is a PID which means the Structure \Cref{thm:structure-thm-PID-mod} applies. That gives us the following decomposition:
\begin{equation}
  \label{eq:H1-split-Z}
  H_1(C_\bullet) \cong \Z^{\oplus k'} \oplus \bigoplus_{i=1}^{k''} \faktor{\Z}{d_i\Z},
\end{equation}
where $d_i$ are non-zero integers that are not invertible (i.e. $d_i \ne \pm 1$), and they form a divisibility chain: $d_1$ divides $d_2$, which in turn divides $d_3$, etc.
We write \Cref{eq:H1-split-Z} it in the presentation of \cite{vuillot2023homological}.
The free part $\Z^{\oplus k'}$ means that we have $k'$ logical qudits of dimension $d$, i.e. living in a $d$-dimensional state space as usual.\footnote{In \cite{vuillot2023homological}, the codes acted on rotor systems -- these can be understood as qudits with $d = \infty$, and this is why the authors reasoned in $\Z$.} The torsion part means that there are $k''$ additional logical qudits, but they have dimensions $d_i$. This torsion part does not exist if $d$ is prime, because then $\Z_d$ is a field. The total number of logical objects is $k = k' + k''$. We now follow this with our result:

\begin{theorem}[Structure Theorem for the Qudit Logical Space]
  \label{thm:structure-theorem-for-qudit-logical-space}
  Suppose we have a system of qudits of dimension $d$, and we have a CSS code described by an abstract cell complex acting on it. Let $C_\bullet$ be the chain complex of $\Z$-modules corresponding to the cell complex. Then the logical space of the code, the first homology \emph{over $\Z_d$}, has the following decomposition:
  \[ H_1^{\Z_d}(C_\bullet) \cong \Z_d^{\oplus k'} \oplus \bigoplus_{i=1}^{k''} \Z_{d_i}, \]
  where $k', k'' \in \N$ and $d_1, \dots, d_{k''}$ are non-zero and non-invertible elements of $\Z_d$, such that $d_i$ divides $d_{i+1}$ for $i=1,\dots,k''-1$. The numbers $k', k''$ here may be different than those in \cref{eq:H1-split-Z}.\footnote{Note that \Cref*{thm:structure-theorem-for-qudit-logical-space} could be generalized: we use the Universal Coefficient Theorem which itself does not require the chain complex to be defined using a topological space. We impose this requirement to set $\Tor_1(H_0(C_\bullet), R) = \bm 0$ because we want to obtain a simple decomposition of $H_1^R(C_\bullet)$ that is similar to the decomposition of a finitely generated module over a PID. However, it is possible to relax this requirement to obtain a more general, though slightly more complicated, theorem.}
\end{theorem}
\begin{proof}
  By the Structure \Cref{thm:structure-thm-PID-mod}, we have the decomposition of the first homology over $\Z$ as in \cref{eq:H1-split-Z}:
  \begin{equation}
    \label{eq:strthm:H1-Z}
    H_1^\Z(C_\bullet) \cong \Z^{\oplus s} \oplus \bigoplus_{i=1}^{t} \faktor{\Z}{q_i\Z},
  \end{equation}
  where $s,t \in \N$ and $q_i \in \Z$ are non-zero, non-invertible (that is, $q_i \ne \pm 1$), and they form a divisibility chain as before. Notice that we now call the rank of the free part $s$ instead of $k'$, and the number of torsion parts $t$ instead of $k''$. We also label the generators of the ideals $q_i$ instead of $d_i$. This is because when moving to homology over $\Z_d$, the new decomposition may have a different rank of the free part, a different number of torsion parts, and different generators of the quotients.

  We now express the first homology over $\Z_d$ by change of ring. By \Cref{thm:ACC-H1-change-of-ring}, for a chain complex arising from an (abstract) cell complex, this is simply
  \[
    H_1^{\Z_d}(C_\bullet) \cong H_1^\Z(C_\bullet) \otimes_\Z \Z_d.
  \]
  What happens there is an extension of scalars along the canonical quotient epimorphism $\Z \epito \Z_d$. Now, we expand the $H_1^\Z(C_\bullet)$ by \cref{eq:strthm:H1-Z}, and use the distributivity of the tensor product over direct sum from \cref{eq:tensor-distributes}:
  \[
    H_1^{\Z_d}(C_\bullet) \cong \left(
      \Z^{\oplus s} \oplus \bigoplus_{i=1}^{t} \Z_{q_i}
    \right) \otimes_\Z \Z_d
    \cong
    \left( \Z \otimes_\Z \Z_d \right)^{\oplus s} \oplus \bigoplus_{i=1}^t Z_{q_i} \otimes_\Z \Z_d.
  \]
  We use the property from \cref{eq:tensor-Za-Zb-gcd} to evaluate these tensor products of $\Z$ and its quotients. Recall that this is $\Z_a \otimes_\Z \Z_b \cong \Z_{\gcd(a,b)}$ for $a,b \in \N$. Note that $\Z \cong \Z_0$, which means that $\Z \otimes_\Z \Z_d \cong \Z_d$ (as expected from \cref{eq:tensor-monoid}). The module $\Z_d$ is a free $\Z_d$-module, so we have a contribution $\Z_d^{\oplus s}$ to the free part of $H_1^{\Z_d}$.

  The quotients become $\Z_{q_i} \otimes_\Z \Z_d \cong \Z_{\gcd(q_i, d)}$. If $q_i$ and $d$ are coprime, meaning that $\gcd(q_i, d) = 1$, then this tensor product is $\Z_1 = \bm 0$, the zero module. In the other extreme, if $\gcd(q_i,d) = d$, then the tensor product is $\Z_d$, another free part. Finally, if~$1 < \gcd(q_i,d) < d$, then $\Z_{\gcd(q_i,d)}$ is a torsion part of $H_1^{\Z_d}(C_\bullet)$. Thus the decomposition is
  \[
    H_1^{\Z_d}(C_\bullet) \cong
    \Z_d^{\oplus s} \oplus \bigoplus_{i=1}^t Z_{\gcd(q_i,d)}
    \eqd \Z_d^{\oplus k'} \oplus \bigoplus_{i=1}^{k''} \Z_{d_i}.
  \]
  We collect the free parts, of which there may be more in $H_1^{\Z_d}$ than in $H_1^\Z$, so $k' \ge s$. Similarly, some of the torsion parts of $H_1^\Z$ may have become free or zero in $H_1^{\Z_d}$, so $k'' \le t$. The generators of ideals $d_i$ correspond to the greatest common divisors that are not $1$ nor $d$ from above.
\end{proof}

The Structure Theorem allows us to compute the logical spaces in terms of~$\Z_d$. This corresponds more closely to the properties of the physical systems we are using, i.e. qudits with finite $d$. In addition, the proof above shows explicitly how to compute $H_1^{\Z_d}(C_\bullet)$ from the decomposition of $H_1^\Z(C_\bullet)$ over $\Z$. This allows us to work with chain complexes of $\Z$-modules, and compute their homologies over $\Z$ as well, and then change the base ring to $\Z_d$ at the end.
That is very useful for computational purposes: computing in a PID like $\Z$ is much more convenient than in a non-PID like $\Z_d$. In particular, Computer Algebra Systems such as Sage\cite{sagemath} are able to deal with $\Z$ much better (see \Cref{cha:jupyter-sage}).

\section{Projective plane}

We now show an example of a topological space where torsion occurs in first homology, and hence derive a CSS code that has torsion in its logical space. The space we use the is \emph{real projective plane}, denoted also $\R\mathcal{P}^2$, a $2$-dimensional non-orientable manifold. Imagining it is difficult, because it cannot be embedded into $3$-dimensional space without self-intersection. It can be described in terms of the gluing diagram in \Cref{fig:gs-RP2-relabeld}. This is the same as in \Cref{fig:gluing-square:proj}, though we relabel the cells for convenience: vertices are labeled $\{\mathbf v_i\}_i$, edges $\{\mathbf{e}_i\}_i$, and faces $\{ \mathbf{f}_i \}_i$ (in~this case only $\mathbf{f}_1$).
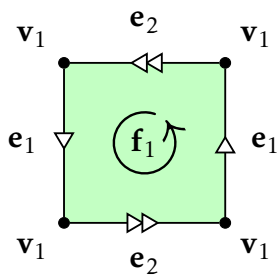
\begin{figure}[h]
  \centering
  \begin{tikzpicture}
  \node[vertex, label={45+90*1:$\mathbf v_1$}] at (45+90*1:1.5) (v1) {};
  \node[vertex, label={45+90*2:$\mathbf v_1$}] at (45+90*2:1.5) (v2) {};
  \node[vertex, label={45+90*3:$\mathbf v_1$}] at (45+90*3:1.5) (v3) {};
  \node[vertex, label={45+90*4:$\mathbf v_1$}] at (45+90*4:1.5) (v4) {};

  \draw[-g>-] (v1) -- (v2);
  \node at (90*2:1.6) {$\mathbf e_1$};

  \draw[-g>>-] (v2) -- (v3);
  \node at (90*3:1.6) {$\mathbf e_2$};

  \draw[-g>-] (v3) -- (v4);
  \node at (90*4:1.6) {$\mathbf e_1$};

  \draw[-g>>-] (v4) -- (v1);
  \node at (90*1:1.6) {$\mathbf e_2$};

  \node (f) at (0,0) {$\mathbf f_1$};
  \node[scale=3, rotate=-45] at (f.center) {$\circlearrowleft$};
  \begin{pgfonlayer}{background}
    \fill[\colora, fill opacity=0.5] (v1.center) -- (v2.center) -- (v3.center) -- (v4.center) -- cycle;
  \end{pgfonlayer}
\end{tikzpicture}
  \caption[Gluing square representation of the real projective plane]{Gluing square representation of the real projective plane.}
  \label{fig:gs-RP2-relabeld}
\end{figure}

In our case, however, we choose a different representation: the \emph{half-sphere model} of the projective plane. This is conventional in the literature on homological CSS codes based on the projective plane; see for example \cite{vuillot2023homological,10.1063/1.2731356}. It is also more general than the gluing square, specifically because there need not be a square subcomplex encompassing neatly the space that we want. In this model, we have a~half of a $2$-sphere, say the upper hemisphere, and we identify the opposite points of the equator. We may also view this as a $2$-disc, with the opposite points of the boundary identified.\footnote{Note that the gluing square is, in fact, compatible with this construction: opposite edges are glued in the opposite direction, and if we forget about the points and edges, we get the disc with antipodal points identified.} In either case, we refer to the circle on the boundary, whose antipodal points get identified, as the \emph{gluing boundary}. It is not a part of an abstract cell complex defined in this space. We show this model in \Cref{fig:RP2-half-sphere}.

\begin{figure}[h]
  \centering
  \input{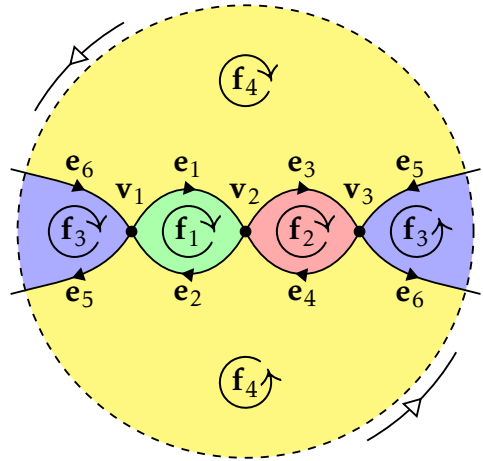}
  \caption[A cellulation of the real projective plane, using the half-sphere model]{A cellulation of the real projective plane, using the half-sphere model. The half-sphere model is essentially a disc with antipodal points identified: for example, notice that the points where $\mathbf{e}_6$ crosses the circle are an angle of $180^\circ$ apart. The identification of antipodal points reverses orientations as one crosses the circle: for example, the face $\mathbf f_4$ is oriented clockwise in the upper part, but anticlockwise in the lower part of the diagram.}
  \label{fig:RP2-half-sphere}
\end{figure}

\subsection{Cellulation}
\label{sec:cellulation}

In \Cref{ex:cyl-Mob,ex:gs-torus,ex:gs-Klein-RP2}, we started with an abstract cell complex, glued its parts, and saw what came out. Here, we switch direction: we have a space first, and we construct an abstract cell complex to match it. This assignment of a cell complex to a space is called a \emph{cellulation}.\cite{HatcherAllen2002At} For this, the half-sphere model of $\R\mathcal{P}^2$ is very useful, because the half-plane itself has no cells by itself. The construction is not trivial, and one can easily make a mistake. We provide a practical guide below:

\subsubsection*{Vertices and edges}

As usual with a cell complex, vertices come first. We place these as we like into the half-sphere. Next, we draw edges between vertices, choosing also their orientation. This is straightforward if the edge does not cross the boundary of the half-sphere. In \Cref{fig:RP2-half-sphere}, these are edges $\mathbf e_1, \dots, \mathbf e_4$.

If an edge does cross the boundary, it continues from the antipodal point, i.e. from a point on the dashed circle exactly $180^\circ$ away. We draw it in its chosen orientation from a vertex into the boundary, and then again from the antipodal point of the boundary. In the figure, these are $\mathbf e_5$ and $\mathbf e_6$.

Notice that crossing the boundary switches left and right. Suppose there are two ants, labeled $A$ and $B$, who live on this surface. They walk side-by-side, $A$ on the left of $B$. Suppose both $A$ and $B$ are walking left from $\mathbf v_1$, the ant $A$ along $\mathbf e_5$, and $B$ along $-\mathbf e_6$ (i.e. opposite to the direction of the edge). As they cross the dashed line, they emerge on the right-hand side, but now $B$, walking along $\mathbf e_6$, is on the left of $A$, still on $\mathbf e_5$.

\subsubsection*{Faces}

This leads us to faces, which are the tricky part. Notice in the picture that the face $\mathbf{f}_4$ appears in two regions of the circle, but it has opposite orientation in each. The same is true for $\mathbf{f}_3$. We can state this as a rule: the orientation of a face is reversed when it crosses the gluing boundary.

We can also state this in a constructive way. Suppose we have built the \mbox{$1$-skeleton}, and wish to add the face $\mathbf{f}_3$. We mark the left-hand region of the picture as a part of this face, and declare the orientation there to be clockwise $\bm\circlearrowright$. Then we choose a vertex incident on this face; in this case, the only choice is $\mathbf{v}_1$. Respecting the clockwise orientation, we follow the edge $\mathbf{e}_5$ from $\mathbf{v}_1$. Since the face is on the right of this edge, we think of this as walking on the right of this edge, within the face $\mathbf{f}_3$. As we cross the glued boundary, we emerge in the right-hand part of the picture. We are on the left of the edge $\mathbf{e}_5$, but still within $\cface_3$. That must mean that in this region, the face is oriented counter-clockwise $\bm\circlearrowleft$.\footnote{Another way to think about this is that we start of on top of the plane (in front of the page), and as we walk across the glued boundary, we cross over below the plane (behind the page).}

\begin{note}
  The above description is possible due to the requirement that in our abstract cell complex (or even a less restricted general cell complex), a face must be bounded by a circle, and hence a face can only exist in a connected component of the $1$-skeleton. We might define more general topological constructions, where this requirement is relaxed. Then, for example, we could have a face bounded by two circles with no edge between them. These are more difficult to deal with, and in that case, we have to fall back on the rule that crossing the gluing boundary reverses the orientation of a face.
\end{note}

\subsection{CSS code}

We construct a CSS code from the cellulation of $\R\mathcal{P}^2$ from \Cref{fig:RP2-half-sphere}. The chain $\Z$-modules are $C_0 = \genR\Z{\mathbf{v}_1,\dots, \mathbf{v}_3}^\oplus$, $C_1 = \genR\Z{\mathbf{e}_1, \dots, \mathbf{e}_6}^\oplus$, and $C_2 = \genR\Z{\mathbf{f}_1, \dots, \mathbf{f}_4}^\oplus$. We choose the natural ordering of these bases, so that we may write down matrices. The differentials are:
\begin{align}
  \label{eq:RP2-code-matrices}
  P_Z^\top = \partial_2
  & =
    \begin{pmatrix*}[r]
            1 & 0 & 0 & -1 \\
            1 & 0 & 0 & 1 \\
            0 & 1 & 0 & -1 \\
            0 & 1 & 0 & 1 \\
            0 & 0 & 1 & 1 \\
            0 & 0 & 1 & -1
    \end{pmatrix*}
  & \text{and} &
  & P_X = \partial_1
  & =
    \begin{pmatrix*}[r]
      -1 & 1 & 0 & 0 & -1 & 1 \\
      1 & -1 & -1 & 1 & 0 & 0 \\
      0 & 0 & 1 & -1 & 1 & -1
    \end{pmatrix*}.
\end{align}
This gives us the following generators of the stabilizer group, which also prescribe which qudits to measure and how:
\[
  \begin{array}[h]{lclclclclcl}
    Z & \otimes & Z & \otimes & \bm1 & \otimes & \bm1 & \otimes & \bm1 & \otimes & \bm1 \\
    \bm1 & \otimes & \bm1 & \otimes & Z & \otimes & Z & \otimes & \bm1 & \otimes & \bm1 \\
    \bm1 & \otimes & \bm1 & \otimes & \bm1 & \otimes & \bm1 & \otimes & Z & \otimes & Z \\
    Z^{-1} & \otimes & Z & \otimes & Z^{-1} & \otimes & Z & \otimes & Z & \otimes & Z^{-1} \\
    X^{-1} & \otimes & X & \otimes & \bm1 & \otimes & \bm1 & \otimes & X^{-1} & \otimes & X \\
    X & \otimes & X^{-1} & \otimes & X^{-1} & \otimes & X & \otimes & \bm1 & \otimes & \bm1 \\
    \bm1 & \otimes & \bm1 & \otimes & X & \otimes & X^{-1} & \otimes & X & \otimes & X^{-1}
  \end{array}
\]
The measurement process itself is architecture-dependent, so do we not go into much detail. In general, we need an entangling operation, a qudit analogue of the CNOT. One option is the operator that maps $\neket i \otimes \neket j \mapsto \neket i \otimes \neket{i + j \mod d}$. Furthermore, we need a way to switch between the $Z$ and $X$ bases. On qubits, this is the Hadamard operator, which is itself an instance of a Fourier transform. This is generalized to the case of qudits; we still denote it $H$. In the general qudit case, its square $H^2$ maps $\neket i \mapsto \neket {-i\mod d}$. Apart from using this for $X$-type measurements as in the qubit case, we also apply $H^2$ before and after CNOT on those qudits, where the stabilizer generator acts as $X^{-1}$ or $Z^{-1}$, i.e. an inverse Pauli operator. For more details, see ref. \cite{cowtan2022qudit}.

Similarly to the case of a qubit CSS code, the main idea is to use the non-commutation of the qudit $X$ and $Z$ operators. A measurement outcome of $\omega^a$ corresponds to a component of the syndrome vector with value $a \in \Z_d$. Observe that for qubits ($d=2$), this is exactly the statement that measuring $(-1)^a$ corresponds to syndrome component of $a \in \Z_2$.

\subsubsection*{Syndrome decoding}

In our example code in \Cref{fig:RP2-half-sphere}, suppose an error $Z \otimes \bm 1 \otimes Z^2 \otimes \bm1^{\otimes 3}$ occurred. This corresponds to the $1$-chain $\cedge_1 + 2\cedge_3$. Its syndrome of $X$-type measurements is $(-1, 1, 0)^\top + 2(0,-1,1)^\top = (-1, -1, 2)^\top$. In general, if an error $Z^{a}$ occurs on an edge $\cedge_\beta$, then this contributes $-a$ to the syndrome component corresponding to $s(\cedge_\beta)$ and $+a$ to $t(\cedge_\beta)$. However, it is not easy to work out the $1$-chain backward from the syndrome. Even with qubits, it may be unclear what the error is. In the case of qudits, the search space is now much larger, and so is the complexity of decoding. Qudit error decoding is a large field of its own, and we do not go into this. See for example the papers \cite{Delfosse2021unionfinddecoders,bravyi2013homological,Watson_2015}.

\subsection{Logical space}

We now compute the first homology of the complex in \Cref{fig:RP2-half-sphere}. First, we compute the kernel of $\partial_1 = P_X$ and image of $\partial_1 = P_Z^\top$, the matrices from \Cref{eq:RP2-code-matrices}. We have used Sage\cite{sagemath} for this computation, and we display an excerpt from the corresponding Jupyter\cite{Kluyver:2016aa} notebook in \Cref{cha:jupyter-sage}. The results are:
\begin{equation}
  \label{eq:RP2-code-ker-im}
  \left\{
  \begin{aligned}
    \ker P_X & = \genR\Z{ \mathbf{e}_1 + \mathbf{e}_2,\ \hphantom{2(}\mathbf{e}_2 + \mathbf{e}_4 - \mathbf{e}_6,\hphantom{)}\ \mathbf{e}_3 + \mathbf{e}_4,\ \mathbf{e}_5 + \mathbf{e}_6 }, \\
    \im P_Z^\top & = \genR\Z{ \mathbf{e}_1 + \mathbf{e}_2,\ 2(\mathbf{e}_2 + \mathbf{e}_4 - \mathbf{e}_6),\ \mathbf{e}_3 + \mathbf{e}_4,\ \mathbf{e}_5 + \mathbf{e}_6 }.
  \end{aligned}\right.
\end{equation}
Notice that while the kernel corresponds to circles in the complex, and as such includes $\mathbf{e}_2 + \mathbf{e}_4 - \mathbf{e}_6$, the image does not include this element. This is because $\mathbf{e}_2 + \mathbf{e}_4 - \mathbf{e}_6$ is not a boundary of a face. However, its multiple $2(\mathbf{e}_2 + \mathbf{e}_4 - \mathbf{e}_6)$ is the boundary of $\vec x \deq \mathbf{f}_1 + \mathbf{f}_2 - \mathbf{f}_3 + \mathbf{f}_4$. This is strange on first look, and it illustrates why working with the matrices is useful: we can compute first, and interpret later. This method is also much more reliable, because it does not depend on topological intuition; in fact, this is how the author gained that intuition in the first place.

The element $\vec x$ represents the sum of all of the faces, in a way that respects their orientation. We interpret it as merging the faces across edges that separate them by forgetting those edges. Specifically, this element is interpreted as merging across $\mathbf{e}_1$, $\mathbf{e}_3$, and $\mathbf{e}_5$, that is the edges not present in $\partial_2(\vec x)$. It is important to note that we do not delete all the edges. In fact, the circle $\lParen \cedge_4, \cedge_2, -\cedge_6 \rParen$ must exist, because the face $\vec x$ must be bounded by a circle. We show this in~\Cref{fig:RP2-circ-merge}.

\begin{figure}[h]
  \centering
  \input{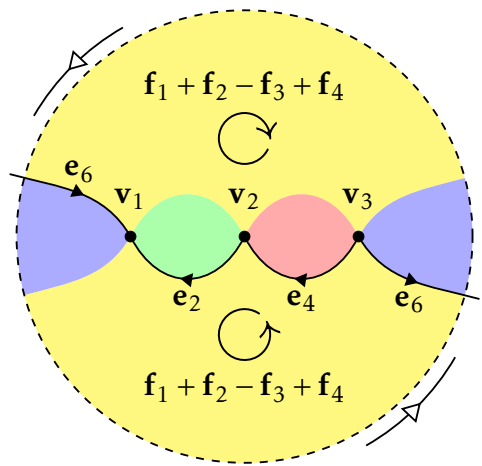}
  \caption[The complex from \protect\Cref*{fig:RP2-half-sphere} with faces merged]{The complex from \protect\Cref{fig:RP2-half-sphere} with faces merged.}
  \label{fig:RP2-circ-merge}
\end{figure}

The boundary of $\vec x$ is $2(\mathbf{e}_2 + \mathbf{e}_4 - \mathbf{e}_6)$, and not just $\mathbf{e}_2 + \mathbf{e}_4 - \mathbf{e}_6$. To see this, imagine once again we are standing inside the face $\vec x$, and walk along its boundary, following its orientation. Start at $\mathbf{v}_1$, and stand just above it: the face is oriented clockwise here, so we walk along $-\mathbf{e}_6$, i.e. opposite the direction of this edge. Walking like this, we have the boundary of $\vec x$ on the left. We cross the gluing boundary, and afterwards we have the boundary of $\vec x$ on the right. After reaching $\mathbf{v}_1$ again, we have walked along $-\mathbf{e}_6 + \mathbf{e}_4 + \mathbf{e}_2$. However, we are still not back at the start, because we are on the wrong side of the boundary. We keep going, make another loop, and when we reach $\mathbf{v}_1$ again, the total path walked is $2(\mathbf{e}_2 + \mathbf{e}_4 - \mathbf{e}_6)$. In this sense, the face touches the boundary twice, but without cancellation. This is the essence of torsion.

It remains to compute the first homology using the kernel and image in \cref{eq:RP2-code-ker-im}. We wrote them in such form that the cancellation in the quotient below is obvious:
\[
  H_1(C_\bullet) = \faktor{\ker P_X}{\im P_Z^\top}
  = \frac{\genR\Z{ \mathbf{e}_1 + \mathbf{e}_2,\ \hphantom{2(}\mathbf{e}_2 + \mathbf{e}_4 - \mathbf{e}_6,\hphantom{)}\ \mathbf{e}_3 + \mathbf{e}_4,\ \mathbf{e}_5 + \mathbf{e}_6 }}{\genR\Z{ \mathbf{e}_1 + \mathbf{e}_2,\ 2(\mathbf{e}_2 + \mathbf{e}_4 - \mathbf{e}_6),\ \mathbf{e}_3 + \mathbf{e}_4,\ \mathbf{e}_5 + \mathbf{e}_6 }}
  \cong \faktor{\Z}{2\Z}.
\]
The torsion of $\R\mathcal{P}^2$ is known to be $\Z_2$,\cite{HatcherAllen2002At} so this is as expected. The single non-trivial homology class of logical  operators we get from this is represented by $\mathbf{e}_2 + \mathbf{e}_4 - \mathbf{e}_6$. Recall that this is a $Z$-type complex, so the corresponding operator is:
\[ \bm 1 \otimes Z \otimes \bm 1 \otimes Z \otimes \bm 1 \otimes Z^{-1}. \]

\section{Torsion of higher order}
\label{sec:larger-torsion}

As we have seen, the origin of torsion in $H_1(C_\bullet)$ is a face whose boundary traverses some circle of edges multiple times. Above, we saw the case of $\Z_2$ torsion. Here, we briefly show how to construct a space with torsion $\Z_q$, where $q>2$. Specifically, we define an abstract cell complex engineered to have torsion $\Z_3$, by gluing a hexagon (in general a $2q$-gon) shown in~\Cref{fig:3-torsion-complex:hex}.\cite{HatcherAllen2002At}
\begin{figure}[h]
  \centering
  \subfloat[\label{fig:3-torsion-complex:hex}]{\begin{tikzpicture}
  \foreach \i [evaluate={\v=int(mod(\i,2)+1)}] in{0,...,6}{
    \coordinate (x\i) at (-60*\i:2);
    \node[vertex, label={-60*\i:$\mathbf{v}_\v$}] at (x\i) {};
  }
  \foreach \i [evaluate={\n=int(\i+1); \j=mod(\n+1,6)}] in{0,2,4}{
    \draw[-g>-] (x\i) -- (x\n);
    \node at (-60*\i-30:2.3) {$\mathbf e_1$};
    \draw[-g>>-] (x\n) -- (x\j);
    \node at (-60*\n-30:2.3) {$\mathbf e_2$};
  }

  \node (f1) at (0,0) {$\mathbf f_1$};
  \node[scale=4, rotate=-60] at (f1.center) {$\circlearrowright$};
  \begin{pgfonlayer}{background}
    \fill[\colora, fill opacity=0.7] (x0) -- (x1) -- (x2) -- (x3) -- (x4) -- (x5) -- cycle;
  \end{pgfonlayer}
\end{tikzpicture}}
  \qquad
  \subfloat[\label{fig:3-torsion-complex:e1}]{\begin{tikzpicture}[x={(1,0)}, y={(0.3, 0.5)}, z={(0,1.5)}]
  \newcommand*{\ylen}{3}
  \newcommand*{\crad}{2.5}
  \newcommand*{\aoff}{-50}
  \foreach \y in {0,1}{
    \coordinate (e1-\y) at (0, \y*\ylen, 0);
    \foreach \a in {0,1,2}{
      \coordinate (x-\a-\y)
      at ({cos(120*\a+\aoff)*\crad}, \y*\ylen, {sin(120*\a+\aoff)*\crad});
    }
  }
  \node[vertex, label={-120:$\cvert_2$}] at (e1-0) {};
  \node[vertex, label={90:$\cvert_1$}] at (e1-1) {};
  \draw[-g>-] (e1-1) -- node[midway, left] {$\cedge_1$} (e1-0);

  \foreach \i [evaluate={\a=int(mod(\i+1, 3))}] in {0,1,2}{
    \draw[-g>>-] (x-\a-1) -- (e1-1);
    \draw[-g>>-] (e1-0) -- (x-\a-0);
    \begin{pgfonlayer}{background}
      \fill[\colora, fill opacity=0.7] (x-\a-1) -- (e1-1) -- (e1-0) -- (x-\a-0) -- cycle;
    \end{pgfonlayer}
  }
\end{tikzpicture}}
  \caption[A complex with torsion $\Z_3$]{A complex with torsion $\Z_3$. \textbf{(a)}~Three edges of a hexagon are glued together to become the edge $\cedge_1$, and analogously for $\cedge_2$. \textbf{(b)}~Picturing the realization of the whole space is not easy. We show what the gluing of $\cedge_1$ looks like.}
  \label{fig:3-torsion-complex}
\end{figure}
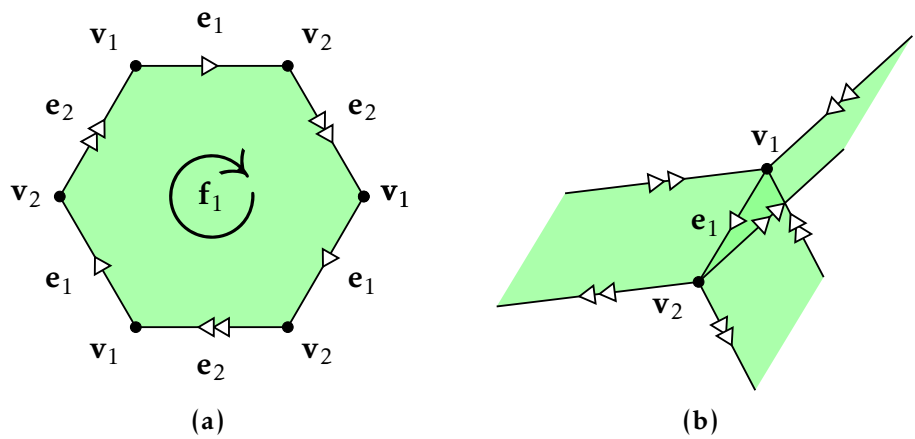

In that complex, the face $\mathbf{f}_1$ is bounded by the circle $\lParen \mathbf{e}_1, \mathbf{e}_2, \mathbf{e}_1, \mathbf{e}_2, \mathbf{e}_1, \mathbf{e}_2 \rParen$ that contains each of the two edges three times, in the same orientation. The differential is $\partial_2(\mathbf{f}_1) = 3(\mathbf{e}_1 + \mathbf{e}_2)$. Observe that $\ker\partial_1$ is generated by the simple cycle $\lParen \mathbf{e}_1, \mathbf{e}_2 \rParen$. Then we have
\[ H_1(C_\bullet) = \faktor{\ker \partial_1}{\im \partial_2} = \faktor{\genR\Z{\mathbf{e}_1 + \mathbf{e}_2}}{\genR\Z{3(\mathbf{e}_1 + \mathbf{e}_2)}} \cong \Z_3. \]

This method generalizes to any $q \ge 2$. However, only for $q=2$ can such a space be realized as a $2$-manifold. In the example, the neighbourhood of (the topological realization~of) any point on the edge $\cedge_1$ cannot possibly be homeomorphic to $\R^2$, as shown in \Cref{fig:3-torsion-complex:e1}.

\chapter{Conclusion and related work}
\label{chap:conclusion}

In this work, we have explored homological quantum error correction, and we derived error correcting codes from cellulations of topological spaces. We explained this in worked examples and provided commentary on the results.
In terms of contributions, we defined our own abstracted and restricted version of a cell complex to match our needs. As a main result, we proved \Cref{thm:structure-theorem-for-qudit-logical-space}, and results leading up to it.
Now, we finish off by mentioning other work in the literature, especially things that we would like to explore next. The following would be the natural directions for this project to continue.

\section{Regular cellulations}

There is an ongoing line of research into \emph{regular} cellulations of various spaces, also called \emph{uniform tilings}; see for example \cite{Breuckmann_2017,7456305}. Regularity means that each face is a regular $r$-gon, and each vertex has $s$ of these around it. This is represented by the Schl\"afli symbol~$\lBrace r, s \rBrace$. For example, we display the real projective plane with a regular cellulation with symbol $\lBrace 3, 5 \rBrace$, i.e. each vertex surrounded by five triangles, in \Cref{fig:3-5-RP2}.

\begin{figure}[h]
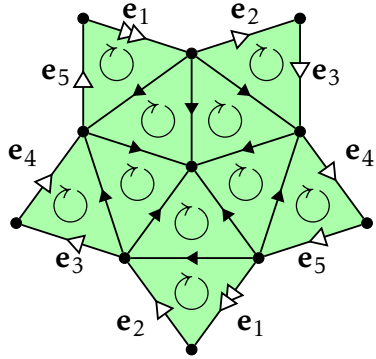

  \centering
  \drawRPHIH[1.5cm]{alpha beta gamma delta epsilon etta theta iota kappa lambda}
  \caption[A regular cellulation of the real projective plane with Schl\"afli symbol~$\lBrace 3, 5 \rBrace$]{A regular cellulation of the real projective plane with Schl\"afli symbol~$\lBrace 3, 5 \rBrace$. We abuse notation slightly and draw a single white triangle to indicate gluing for all of the edges except $\cedge_1$, which is indicated by two white triangles. By this, we do not mean that there is an edge that is glued eight times, this is simply a notational convenience. Furthermore, we also label the edges to make this clear. Observe that this gluing makes the complex a real projective plane, and it is a cellulation of the half-plane model, similar to \protect\Cref{fig:RP2-half-sphere}.}
  \label{fig:3-5-RP2}
\end{figure}
The advantage of regularity is that these codes are by design LDPC (see \Cref{sec:LDPC}). Recall this means that the Hamming weight of each stabilizer, as well as the number of stabilizer measurements that each physical qudit is involved in, are both bounded from above by a constant. For an $\lBrace r, s \rBrace$ regular cellulation, each $X$-type stabilizer (vertex) involves $s$ qudits (edges). Similarly, each $Z$-type stabilizer (face) touches $r$ qudits.

A particular class of regular cellulations are those that tile the hyperbolic plane. Note that this is no longer just a topological notion: we are imposing a metric on some space. A metric does not change the homology itself, however, it allows us to cellulate the space in a regular and controlled way. We then cut holes out of the plane, or perform gluings, to create non-trivial first homology. More can be found in ref. \cite{Breuckmann_2017,7456305}.

\section{Fault-tolerant computation and transversals}

The encoding of logical operators is itself a problem that needs solving. Of course, we get operators that correspond to $1$-chains representing the equivalence classes in the first homology of the chain complex. However, these are, essentially, just $X$ and $Z$ acting on logical qudits. For fault-tolerant computation, we need more operators. In particular, we need an entangling operation, say an encoded analogue of CNOT, and we need a change of basis, that is an encoded Hadamard.

There exists a class of fault-tolerant operators that are easy to implement. Let $A$ be a physical operator acting on an un-encoded system in a certain way. If it is possible to encode $A$ as a logical operator $\overline A$ acting on our logical qudits in essentially the same way as $A$ acts on un-encoded qudits, then it is called a \emph{transversal operator}. To make this concrete, suppose our codespace contains two logical qudits of the same kind encoded in disjoint blocks of physical qudits. This essentially means that we have two copies of the same smaller code with a single logical qudit, and these run in parallel. If we are able to perform logical CNOT between these two logical qudits by applying the physical CNOT between the corresponding qudits of the two blocks, then CNOT is a transversal of this code. It is not always clear how to find transversals, or how to construct codes given some desired transversal operation. More on this in refs. \cite{gottesman1997stabilizer,10.1063/1.1499754,breuckmann2022foldtransversal}

\section{Code surgery}

Throughout the work, we have mostly dealt with a single code at a time. In \Cref{chap:joining-codes}, we give some thoughts on joining two codes, either using the tensor product of chain complexes, or a connected sum of the underlying topological spaces (which corresponds to a quotient of the direct sum of the topological complexes). More generally, one can use the tools of topological \emph{surgery theory} which studies how a space can be split up, and the parts recombined in a possibly novel way. In the context of quantum error correction, this is another way of constructing new codes with favourable properties. In particular, surgery of LDPC codes yields new codes which are also LDPC. More can be read in ref. \cite{cowtan2023css}.


\appendix
\part{Appendix}
\chapter[Joining Codes]{Joining Codes \hspace*{\fill}\begin{tikzpicture}
  \tikzstyle{vertex}=[fill=black,draw=black,circle,minimum size=2pt,inner sep=0]
  \tikzstyle{compedge}=[
  postaction={decorate},
  line width=0.7pt,
  draw]
  \tikzstyle{->-}=[compedge,
  decoration={
    markings,
    mark=at position 0.7 with {\arrow{Triangle[scale=0.5]}}}]

  \newcommand*{\side}{3mm}
  \node[vertex] (v1) at (0,0) {};
  \node[vertex] (v2) at (0,-\side) {};
  \node[vertex] (v3) at (\side/2,\side/2) {};
  \node[vertex] (v4) at (3*\side/2,\side/2) {};
  \node[vertex] (v13) at (\side/2,0) {};
  \node[vertex] (v14) at (3*\side/2,0) {};
  \node[vertex] (v23) at (\side/2,-\side) {};
  \node[vertex] (v24) at (3*\side/2,-\side) {};

  \draw[->-] (v1.center) -- (v2.center);
  \draw[->-] (v3.center) -- (v4.center);
  \draw[->-] (v13.center) -- (v23.center);
  \draw[->-] (v13.center) -- (v14.center);
  \draw[->-] (v14.center) -- (v24.center);
  \draw[->-] (v23.center) -- (v24.center);

  \begin{pgfonlayer}{background}
    \fill[\colora, fill opacity=1] (v13.center) -- (v14.center) -- (v24.center) -- (v23.center) -- cycle;
  \end{pgfonlayer}
\end{tikzpicture}}
\label{chap:joining-codes}

In this chapter, we briefly explore what can be done when we have two or more codes obtained from chain complexes corresponding abstract cell complexes (or in general from any chain complexes), and we wish to merge them. We only show the two simple options, the \emph{tensor product}, and the \emph{connected sum}. This is an extension of the main body of the work, and not all details have been worked out yet.

\section{Tensor product}
\label{sec:tensor-product-codes}

In this section, we define a tensor product of two codes based on cellulations. We focus mainly on the product of two classical codes, which will give us a quantum code; but at the end, we mention more general products. We describe the tensor product using examples.

\subsection{The $[3,1,3]$ classical code as a cell complex}
\label{sec:3-1-3-returns}

We come back to the $[3,1,3]$ classical code from \Cref{sec:3-1-3-code}, and we interpret is as an abstract cell complex. Recall that the parity matrix is
\[ P =
  \begin{pmatrix*}
    1 & 1 & 0 \\
    1 & 0 & 1
  \end{pmatrix*} : \Z_2^3 \to \Z_2^2
\]
We wish to interpret this as a differential $\partial_1$ of a length 1 chain complex $C_\bullet$, corresponding to a $1$-dimensional abstract cell complex $X = (X_\bullet, \varphi_\bullet)$. We change the ring to $\Z$, and we add minus signs to some of the entries. The choice we make is:
\[ \partial_1 =
  \begin{pmatrix*}[r]
    -1 & -1 & 0 \\
    1 & 0 & -1
  \end{pmatrix*} : C_1 \to C_0
\]
We assume the ordered bases are $(\cedge_1, \cedge_2, \cedge_3)$ for $C_1$, and $(\cvert_1, \cvert_2)$ for $C_0$. With this labeling, the first column simply says that $\varphi_1(\cedge_1) = (\cvert_1, \cvert_2)$. However, the other two edges are tricky. The matrix gives their starting vertices, but does say anything about their targets.

One way to deal with this is to interpret those edges as shooting off to infinity. This breaks the rules, as it is now not a valid abstract cell complex. But we might generalize our notion of a topological complex to include this scenario, which now corresponds to a bi-infinite line, as shown in \Cref{fig:3-1-3-as-cell-complexes-infty}.

\begin{figure}[h]
  \centering
  \newcommand*{\edgewidth}{0.75pt}
  \subfloat[\label{fig:3-1-3-as-cell-complexes-infty}]{%
    \begin{tikzpicture}
  \newcommand*{\dottedlen}{0.35}
  \tikzstyle{vertex}=[fill=black,draw=black,circle,minimum size=4pt,inner sep=0]
  \tikzstyle{->-}=[
  decoration={
    markings,
    mark=at position 0.55 with {\arrow{Triangle[scale=1.2]}}},
  postaction={decorate},
  draw,
  line width=\edgewidth]

  \node[vertex, label={-90:$\cvert_1$}] (v1) at (0,0) {};
  \node[vertex, label={-90:$\cvert_2$}] (v2) at (1.5,0) {};
  \coordinate (l) at (-1.3,0);
  \coordinate (ll) at ($(l)-(\dottedlen,0)$);
  \coordinate (r) at (2.5,0);
  \coordinate (rr) at ($(r)+(\dottedlen,0)$);

  \draw[->-] (v1.center) -- node[midway,above] {$\cedge_1$} (v2.center);
  \draw[->-] (v1.center) -- node[midway,above] {$\cedge_2$} (l);
  \draw[->-] (v2.center) -- node[midway,above] {$\cedge_3$} (r);
  \draw[dotted, line width=\edgewidth] (r) -- (rr);
  \draw[dotted, line width=\edgewidth] (l) -- (ll);
\end{tikzpicture}
  }
  \qquad
  \subfloat[\label{fig:3-1-3-as-cell-complexes-circle}]{
    \begin{tikzpicture}
  \newcommand*{\radius}{1}
  \tikzstyle{vertex}=[fill=black,draw=black,circle,minimum size=4pt,inner sep=0]
  \tikzstyle{->-}=[
  decoration={
    markings,
    mark=at position 0.55 with {\arrow{Triangle[scale=1.2]}}},
  postaction={decorate},
  draw,
  line width=\edgewidth]

  \node[vertex, label={150:$\cvert_1$}] (v1) at (150:\radius) {};
  \node[vertex, label={30:$\cvert_2$}] (v2) at (30:\radius) {};
  \node[vertex, label={-90:$\cvert_\infty$}] (vinf) at (-90:\radius) {};

  \draw[->-] (v1.center)
  to [out=30, in=150]
  node[midway,above] {$\cedge_1$}
  (v2.center);

  \draw[->-] (v1.center)
  to [out=-90, in=150]
  node[midway,below left] {$\cedge_2$}
  (vinf.center);

  \draw[->-] (v2.center)
  to [out=-90, in=30]
  node[midway,below right] {$\cedge_3$}
  (vinf.center);

\end{tikzpicture}
  }
  \caption[Interpretations of the $\lbrack 3,1,3 \rbrack$ repetition code as abstract cell complexes]{Interpretations of the $[3,1,3]$ repetition code as abstract cell complexes.}
  \label{fig:3-1-3-as-cell-complexes}
\end{figure}
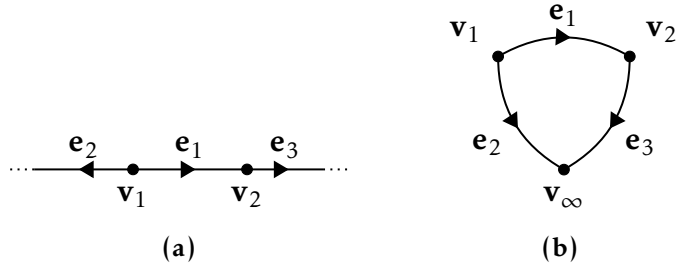

However, there is another way to capture this construction. This one does not require us to break the rules, and it will be much more useful when defining the product of codes in \Cref{sec:3-1-3-tensor}. We adjoin a \emph{point at infinity}, labeled $\cvert_\infty$, and connect all single-endpoint edges to it. This means changing the chain module $C_0$ to $\Z_2^{\oplus 3}$ by adding a new basis vector. The new ordered basis is $(\cvert_1, \cvert_2, \cvert_\infty)$. Correspondingly, we change the differential:
  \[
    \partial_1 \deq
    \begin{tikzpicture}[baseline=(current bounding box.center)]
      \matrix (mat)
      [matrix of math nodes,
      left delimiter={\lparen},
      right delimiter={\rparen},
      column sep=0.9em,
      row sep=1ex,
      nodes={inner sep=0pt, outer sep=0pt}]
      {
        -1 & -1 & \hphantom-0 \\
        \hphantom-1 & \hphantom-0 & -1 \\[1ex]
        \hphantom-0 & \hphantom-1 & \hphantom-1 \\
      };
      \draw[line width=0.3mm] ($0.5*(mat-2-1.south west)+0.5*(mat-3-1.north west)-(0.05,0)$)
      -- ($0.5*(mat-2-3.south east)+0.5*(mat-3-3.north east)+(0.05,0)$);
    \end{tikzpicture}
  \]
  This corresponds to setting $\varphi_2(\cedge_2) = (\cvert_1, \cvert_\infty)$ and $\varphi_2(\cedge_3) = (\cvert_2, \cvert_\infty)$. The new abstract cell complex forms a circle, and we display it in \Cref{fig:3-1-3-as-cell-complexes-circle}.

  Note that the parity matrix $P = \partial_1$ is different than before, and this is formally a different code under \Cref{con:different-codes}. However, it is essentially the same, as we have not changed the codespace.

\subsection{Tensor product of complexes}
\label{sec:3-1-3-tensor}

We now define the tensor product of an abstract cell complexes, and of chain complexes. For more details, see \cite{HatcherAllen2002At}. These two constructions correspond, and allow us to construct new codes from old.

\begin{definition}[tensor product cell complex]
  Let $X = (X_\bullet, \varphi_\bullet)$ and $Y = (Y_\bullet, \psi_\bullet)$ be abstract cell complexes of dimension at most $1$, i.e. we require $X_2 = Y_2 = \varnothing$. Define its \emph{tensor product},\footnote{This is usually called just the \emph{product}.} denoted $X \otimes Y$, as the abstract cell complex $(P_\bullet, \pi_\bullet)$, where $P_n \deq \{ e^{i}_\alpha \otimes e^j_\beta : i+j = n, e^i_\alpha \in X_i, e^j_\beta \in Y_j \}$. The elements $e^i_\alpha \otimes e^j_\beta$ are just pairs of cells from each complex, and $\otimes$ is just a chosen notation. Observe that a product of an $n$- and an $m$-cell is an $(n+m)$-cell; we restrict the dimension of $X$ and $Y$, so that the product $P = X \otimes Y$ is at most $2$-dimensional. The attaching map for $1$-cells acts as $\pi_1(e^1_\alpha \otimes e^0_\beta) = (s(e^1_\alpha) \otimes e^0_\beta, t(e^1_\alpha) \otimes e^0_\beta)$, and analogously on $e^0_\alpha \otimes e^1_\beta$. The attaching of product $2$-cells is defined analogously. We show an example in \Cref{fig:tensor-product-cell-complex}.
\end{definition}

\begin{figure}[h]
  \centering
  \begin{tikzpicture}
    \newcommand*{\side}{3}
    \node[vertex, label={180:$e^0_1$}] (v1) at (-\side/3,0) {};
    \node[vertex, label={180:$e^0_2$}] (v2) at (-\side/3,-\side) {};
    \node[vertex, label={90:$e^0_3$}] (v3) at (\side/2,5*\side/6) {};
    \node[vertex, label={90:$e^0_4$}] (v4) at (3*\side/2,5*\side/6) {};
    \node[vertex, label={90+45:$e^0_1 \otimes e^0_3$}] (v13) at (\side/2,0) {};
    \node[vertex, label={90-45:$e^0_1 \otimes e^0_4$}] (v14) at (3*\side/2,0) {};
    \node[vertex, label={-90-45:$e^0_2 \otimes e^0_3$}] (v23) at (\side/2,-\side) {};
    \node[vertex, label={-90+45:$e^0_2 \otimes e^0_4$}] (v24) at (3*\side/2,-\side) {};

    \draw[->-] (v1.center) -- node [midway, left] {$e^1_1$} (v2.center);
    \draw[->-] (v3.center) -- node [midway, above] {$e^1_2$} (v4.center);
    \draw[->-] (v13.center) -- node [midway, left] {$e^1_1 \otimes e^0_3$} (v23.center);
    \draw[->-] (v13.center) -- node [midway, above] {$e^0_1 \otimes e^1_2$} (v14.center);
    \draw[->-] (v14.center) -- node [midway, right] {$e^1_1 \otimes e^0_4$} (v24.center);
    \draw[->-] (v23.center) -- node [midway, below] {$e^0_2 \otimes e^1_2$} (v24.center);

    \node (f) at (\side, -\side/2) {$e^1_1 \otimes e^1_2$};
    \node[scale=5] at (f.center) {$\circlearrowleft$};

    \begin{pgfonlayer}{background}
      \fill[\colora, fill opacity=1] (v13.center) -- (v14.center) -- (v24.center) -- (v23.center) -- cycle;
    \end{pgfonlayer}
\end{tikzpicture}
  \caption[{Example of the tensor product of cell complexes}]{Example of the tensor product of cell complexes. The two factor complexes are the lines on the left and top, and they are positioned in a way to make the structure of the tensor product clear.}
  \label{fig:tensor-product-cell-complex}
\end{figure}
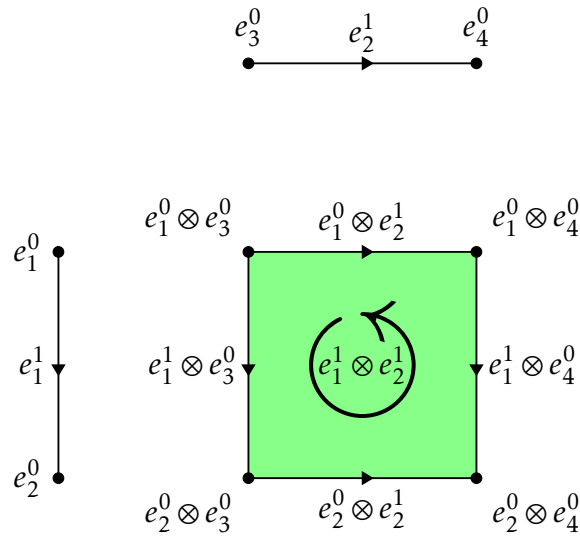

Correspondingly, the chain complex describing this is the following:
\begin{definition}[tensor product of chain complexes\cite{HatcherAllen2002At}]
  \label{def:tensor-product-chain-complex}
  Let $C_\bullet$ and $D_\bullet$ be chain complexes of $R$-modules. Define their \emph{tensor product} $C_\bullet \otimes D_\bullet$ as the chain complex $P_\bullet$ with components
  \[ P_n \deq \bigoplus_{i + j = n} C_i \otimes_R D_j, \]
  and differentials are linearly extended from
  \[
    \partial_n^{P_\bullet}\left( \bigoplus_{i + j = n} x_i \otimes y_j \right)
    \deq \sum_{i + j = n} \partial_i^{C_\bullet}(x_i) \otimes y_j + (-1)^i x_i \otimes \partial_j^{D_\bullet}(y_j),
  \]
  where $n \in \N$, and each $x_i \in C_i$, and $y_j \in D_j$.
\end{definition}

Note that conventionally $C_{-1} = D_{-1} = \bm 0$. By convention, $P_{-1}$ should also be the zero module. We check that this is the case following \Cref{def:tensor-product-chain-complex}:
\[
  P_{-1} = (C_{-1} \otimes_R D_0) \oplus (C_0 \otimes_R D_{-1}) = (\bm 0 \otimes_R D_0) \oplus (C_0 \otimes_R \bm 0) = \bm 0 \oplus \bm 0 = \bm 0,
\]
exactly as expected. Above, we used the property from \cref{eq:tensor-zero}.

With the above definitions, we can define the tensor product of error correcting codes. If we bound the factors to dimension one, i.e. classical codes, we immediately get a quantum code out of this. For the interpretation of the $[3,1,3]$ code from \Cref{fig:3-1-3-as-cell-complexes-infty} which is a line, the product is a plane. This is not an interesting code, because it has trivial first homology.
On the other hand, this is exactly why we constructed also \Cref{fig:3-1-3-as-cell-complexes-circle}. The product of two circles is a torus, and now this is a viable quantum code with logical space (first homology) of rank 2.

This may be further generalized to products of quantum codes. Then, the chain complex has length four, and our abstract cell complexes can no longer interpret it (though, general cell complexes can). Then we have to choose how to assign parity check matrices, so that we obtain generators of the stabilizer group. One way is to truncate the complex to length two by deleting one end. Another option is to see the extra differentials as defining \emph{meta-checks}, which detect measurement errors. More on these can be found in ref. \cite{bravyi2013homological}.

\section{Connected sum}
\label{sec:connected-sum-of-codes}

Finally, we briefly mention the connected sum of codes. In \Cref{def:direct-sum-of-cell-complexes}, we have defined the direct sum of cell complexes, which corresponds to the direct sum of chain complexes. We also sketched how to glue them. This clearly also corresponds to an operation that takes two codes and merges them.

This can be formalized as an actual sum of abstract cell complexes, but we can also use a more general connected sum (see \cite{HatcherAllen2002At}) of topological spaces and then cellulate it. This has different effects on the logical space. This concept may be generalized to code surgery, see \cite{cowtan2023css,cowtan2022qudit}.

\chapter{Jupyter + Sage notebook}
\label{cha:jupyter-sage}
We display in \Cref{fig:sage-setup,fig:sage-ntb} a portion of our Jupyter\cite{Kluyver:2016aa} notebook where we have used Sage\cite{sagemath} to analyze the complex from \Cref{fig:RP2-half-sphere}. We use the Sage standard library to find $\ker \partial_1$ and $\im \partial_2$ as matrices over $\Z$, and we further use Sage to change the basis, so that the generating sets of these spaces are in row echelon form. The function \texttt{on\_basis\_from\_dict} is our own utility, written to more easily translate differentials computed by hand into Sage. It is displayed in \Cref{fig:sage-setup}.

\begin{figure}[h]
  \centering
  \includegraphics[page={1},scale=0.7]{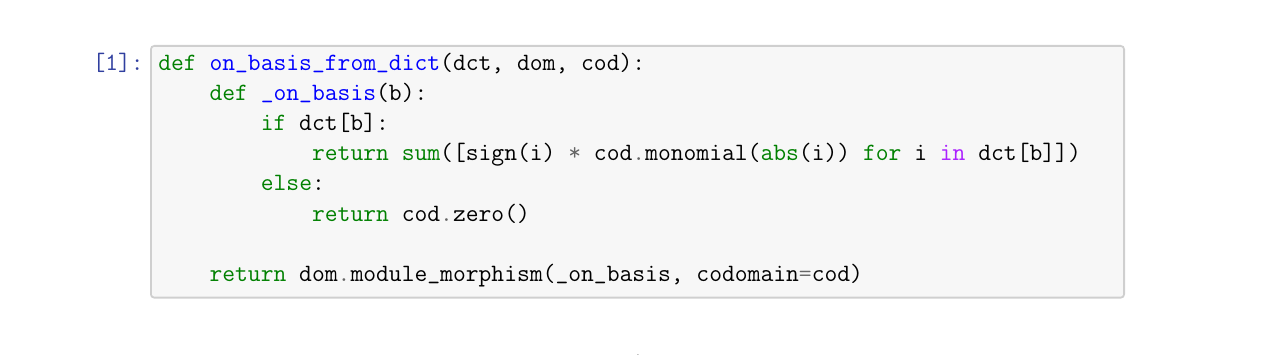}
  \caption[The custom utility function on\_basis\_from\_dict]{The custom utility function \texttt{on\_basis\_from\_dict}. It allows to conveniently input a linear transformation $M$, defined by its action on the standard basis $\vec e_i \mapsto \sum_j M_{j,i} \vec e_j$, as a dictionary. It works in the specific case that the entries of $M$ as a matrix are all $0$ or $\pm 1$. The dictionary has keys corresponding to input basis vectors, and the values are lists of output basis vectors with nonzero coefficients. These are written using their integer indices, and they are given the sign corresponding to the matrix value. For example, if we have $\vec e_i \mapsto \vec e_1 - \vec e_3$ , the key-value pair is $\{ i : [+1, -3] \}$, omitting $\vec e_2$ because it does not occur.}
  \label{fig:sage-setup}
\end{figure}

\begin{figure}[h]
  \centering
  \includegraphics[page={1},scale=0.7]{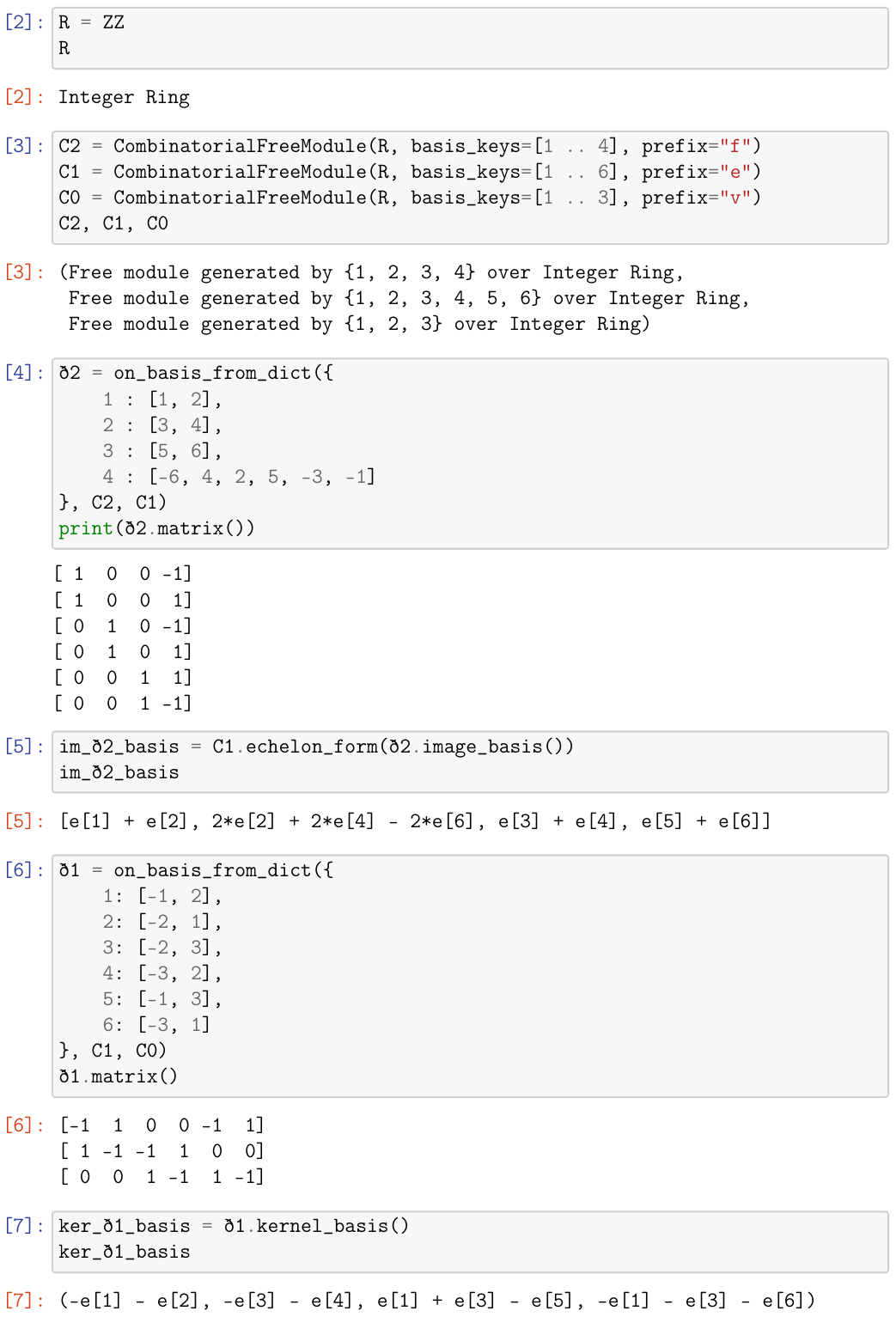}
  \caption[Computation of $\ker \partial_1$ and $\im \partial_2$ of \protect\Cref*{fig:RP2-half-sphere}]{Computation of $\ker \partial_1$ and $\im \partial_2$ of \protect\Cref{fig:RP2-half-sphere}. We define the ring of scalars to be $\Z$, and define three free and finitely generated modules $C_2, C_1, C_0$ by listing their bases. Then we define the differentials $\partial_2$ and $\partial_1$ using the function from \protect\Cref{fig:sage-setup}. Finally, we ask Sage to compute the bases of, respectively, their image and kernel, and write them in row echelon form, so that may easily compute the quotient $\ker \partial_1 / \im \partial_2$.}
  \label{fig:sage-ntb}
\end{figure}

\addcontentsline{toc}{chapter}{Bibliography}
\bibliography{bibliography}
\bibliographystyle{alpha}

\end{document}